\documentclass{article}

\usepackage[final]{neurips_2025}

\usepackage[Symbol]{upgreek}
\usepackage{xcolor} 
\usepackage[most]{tcolorbox}
\usepackage{soul}
\usepackage{pifont}
\usepackage[most]{tcolorbox}
\usepackage{colortbl}
\usepackage{bbding}
\usepackage{caption}
\usepackage{wrapfig}
\usepackage{color,xcolor}
\usepackage{microtype}
\usepackage{graphicx}
\usepackage{subfigure}
\usepackage{booktabs} 
\usepackage{multirow}
\usepackage{makecell}
\usepackage{diagbox}
\usepackage{amssymb}
\usepackage{mathtools}
\usepackage{natbib}
\usepackage{amsthm}

\usepackage[toc,page,header]{appendix}
\usepackage{minitoc}
\usepackage{enumitem}
\usepackage{amsmath}
\usepackage{tikz}
\usepackage[ruled,linesnumbered,vlined]{algorithm2e}

\def \R {\mathbb{R}}

\def \W {\mathbf{W}}
\def \one {\mathbf{1}}

\def \x {\mathbf{x}}

\def \E {\mathbb{E}}

\def \p {\mathbf{p}}

\def \z {\mathbf{z}}

\def \y {\mathbf{y}}

\def \F {\mathcal{F}}

\def \N {\mathcal{N}}
\def \V {\mathcal{V}}

\def \C {\mathcal{C}}

\usepackage{hyperref}
\usepackage[capitalize,noabbrev]{cleveref}
\hypersetup{colorlinks={true},linkcolor={blue},citecolor={blue}}

\newtheorem{theorem}{Theorem}

\newtheorem{lemma}{Lemma}

\newtheorem{definition}{Definition}
\newtheorem{assumption}{Assumption}
\newtheorem{remark}{Remark}

\title{Effective Policy Learning for Multi-Agent Online Coordination Beyond Submodular Objectives}

\author{ 
Qixin Zhang$^1$
\quad
Yan Sun$^2$
\quad
Can Jin$^4$
\quad
Xikun Zhang$^1$
\\
\vspace{0.2cm}
\textbf{
Yao Shu$^3$
\quad
Puning Zhao$^5$
\quad
 Li Shen$^5$
\quad
Dacheng Tao$^1\thanks{Corresponding author: Dacheng Tao.}$
}\\
\small{$^1$College of Computing and Data Science, Nanyang Technological University, Singapore}\\
\small{$^2$The University of Sydney}
\quad
\small{$^3$Hong Kong University of Science and Technology(Guangzhou)}\\
\quad 
\small{$^4$Rutgers University}
\quad
\small$^5${Shenzhen Campus of Sun Yat-sen University}\\
{\tt\small
\{qixin.zhang, xikun.zhang, dacheng.tao\}@ntu.edu.sg}; sun9899@uni.sydney.edu.au\\
{\tt\small yaoshu@hkust-gz.edu.cn; \{zhaopn,shenli6\}@mail.sysu.edu.cn}
}

\begin{document}
	
\doparttoc % 
\faketableofcontents %
\maketitle
\begin{abstract}
In this paper, we present two effective  policy learning algorithms for multi-agent online coordination(MA-OC) problem. The first one, \texttt{MA-SPL},  not only can achieve the optimal $(1-\frac{c}{e})$-approximation guarantee for the MA-OC problem with submodular objectives but also can handle the unexplored  $\alpha$-weakly DR-submodular and $(\gamma,\beta)$-weakly submodular scenarios, where $c$ is the curvature of the investigated submodular functions, $\alpha$ denotes the diminishing-return(DR) ratio and the tuple $(\gamma,\beta)$ represents the submodularity ratios.  Subsequently, in order to reduce the reliance on the unknown parameters $\alpha,\gamma,\beta$ inherent in the \texttt{MA-SPL} algorithm, we further introduce the second online algorithm named \texttt{MA-MPL}. This \texttt{MA-MPL} algorithm is entirely \emph{parameter-free} and simultaneously can maintain the same approximation ratio as the first  \texttt{MA-SPL} algorithm. The core of our \texttt{MA-SPL} and \texttt{MA-MPL} algorithms is a novel continuous-relaxation technique termed as \emph{policy-based continuous extension}. Compared with the well-established \emph{multi-linear extension}, a notable advantage of this new \emph{policy-based continuous extension} is its ability to provide a lossless rounding scheme for any set function, thereby enabling us to tackle the challenging weakly submodular objectives. Finally, extensive simulations are conducted to validate the effectiveness of our proposed algorithms.
\end{abstract}

\section{Introduction}%These applications span a wide range of tasks,
Coordinating multiple autonomous agents to cooperatively complete complex tasks in time-varying environments is a significant challenge with extensive applications in machine learning, robotics and control, including target tracking~\citep{corah2021scalable,zhou2023robust,zhou2019sensor,zhou2018resilient}, area monitoring~\citep{hashemi2019submodular,li2023submodularity,schlotfeldt2021resilient},  multi-path planning~\citep{shi2023robust,singh2007efficient,singh2009nonmyopic}, mobile sensor placement~\citep{krause2008efficient,krause2008near,robey2021optimal}, collaborative data selection~\citep{mirzasoleiman2016distributed,FedSub} environmental mapping~\citep{atanasov2015decentralized,liu2021distributed} and task assignment~\citep{qu2019distributed,arslan2007autonomous}. Motivated by these diverse real-world scenarios,  thus this paper delves into the \underline{M}ulti-\underline{A}gent \underline{O}nline \underline{C}oordination(MA-OC) problem.

 Prior to this, numerous studies have demonstrated that the utility functions associated with a wide range of multi-agent coordination scenarios often exhibit a \emph{diminishing-return}(DR) property. Specifically, as the number of agents increases, the marginal gain in benefit will tend to decrease. For instance, in area monitoring using a fleet of unmanned aerial vehicles(UAVs), due to the overlaps of sensing ranges, the increment of the total monitored area from adding an additional UAV typically becomes less and less as the team of UAVs expands. Note that the \emph{diminishing-return} property is also known as \emph{submodularity} in mathematics~\citep{bach2013learning,fujishige2005submodular}. Consequently, a vast majority of research regarding multi-agent coordination focus on \emph{submodular} objectives~\citep{gharesifard2017distributed,grimsman2018impact,kia2025submodular,krause2008robust,marden2016role,qu2019distributed,rezazadeh2023distributed,robey2021optimal,xu2023bandit,xu2023online,zhang2025nearoptimal,zhou2023robust,zhou2022risk}. 
 
However, recent works~\citep{hashemi2019submodular,hashemi2020randomized,jawaid2015submodularity,kaya2025randomized,ye2023maximization} observed that there also exist many multi-agent coordination scenarios inducing utility functions that are \emph{close-to-submodular}, but \emph{not strictly submodular}. A notable example is the employment of a swarm of UAVs to track multiple moving objects. In this scenario, each UAV needs to
periodically determine its moving direction and speed. Particularly, under the Kalman filter framework, \citep{hashemi2019submodular,hashemi2020randomized,kaya2025randomized}  pointed out that the aforementioned trajectory selection problem of UAVs can be formulated as a multi-agent variant of the general online \emph{weakly} submodular maximization problems. It is worth noting that the currently well-known algorithms for weakly submodular maximization are highly dependent on the discrete local search~\citep{chen2018weakly,gatmiry2018non,harshaw2019submodular,JOGO-Lu,thiery2022two}. So far, however, how to extend this local search into online settings still remains an uncharted territory.  Thus, almost all existing studies regarding multi-agent coordination with weakly submodular objectives choose to neglect the changing environment and focus on simple offline scenarios~\citep{hashemi2019submodular,hashemi2020randomized,kaya2025randomized,ye2023maximization}.  In view of all this, a natural question arises:
 \begin{center}\vspace{-0.25em}
 	\textbf{Q1:} \emph{Is it possible to design an effective online algorithm for MA-OC problem with weakly submodular objectives?}
 \end{center}\vspace{-0.25em}
 
 In addition, we also find that when the objective function is exactly \emph{submodular}, the state-of-the-art algorithms~\citep{rezazadeh2023distributed,zhang2025nearoptimal} for the MA-OC problem only can guarantee a sub-optimal $(\frac{1-e^{-c}}{c})$-approximation, which mismatches the best possible $(1-\frac{c}{e})$-approximation established in the works~\citep{sviridenko2017optimal,harvey2020improved} for \emph{single-agent} submodular maximization. Here, $c$ denotes the curvature of the investigated submodular functions. Given this drawback, another question comes to our mind, that is, 
 \begin{center}\vspace{-0.25em}
 	\textbf{Q2:} \emph{Is it possible to achieve the optimal $(1-\frac{c}{e})$-approximation ratio for MA-OC problem \\with submodular objectives?}
 \end{center}\vspace{-0.25em}
In the subsequent sections of this paper, we will provide an affirmative answer to these two questions by presenting an effective online algorithm named \texttt{MA-SPL}, which not only can achieve the optimal $(1-\frac{c}{e})$-approximation for the MA-OC problem with submodular objectives but also can address the previously unexplored  weakly submodular scenarios. The core of our proposed \texttt{MA-SPL} algorithm is a novel continuous-relaxation framework named \emph{policy-based continuous extension}, which can efficiently transform the discrete set function maximization  problem into a solvable continuous optimization task. Furthermore, compared to the well-established \emph{multi-linear extension}~\citep{calinescu2011maximizing}, a notable advantage of our proposed policy-based continuous extension is that it can provide a lossless rounding scheme for any set objective function. In contrast, all known lossless rounding schemes for the \emph{multi-linear extension} require the set objective function to be \emph{submodular}~\citep{calinescu2011maximizing,chekuri2010submodular,chekuri2014submodular}. Moreover, to eliminate the dependence of both DR ratio and submodularity ratio inherent in our proposed \texttt{MA-SPL} algorithm, we further present a \emph{parameter-free} online algorithm termed as \texttt{MA-MPL} for the MA-OC problem with general weakly submodular objectives. 

In summary, we make the following contributions:
\begin{list}{\labelitemi}{\leftmargin=0.5em \itemindent=-0.0em \itemsep=.1em}
	\vspace{-.05in}
	\item[$\bullet$] This paper introduces an innovative continuous-relaxation technique named \emph{policy-based continuous extension} for the general multi-agent coordination problem. Furthermore, we conduct an in-depth exploration of the differentiability, monotonicity and submodularity of our proposed policy-based continuous extension. More importantly, when the investigated set objective function is submodular or weakly submodular, we design three different surrogate functions for our policy-based continuous extension. The stationary points of these three surrogate functions can yield a better approximation guarantee than those of the original policy-based continuous extension itself.
	\item[$\bullet$] Building on these surrogate functions, we then propose a novel online algorithm named \texttt{MA-SPL} for the concerned MA-OC problem. Moreover, we also verify that, when the set objective function is monotone submodular with curvature $c$, $\alpha$-weakly DR-submodular or $(\gamma,\beta)$-weakly submodular, our proposed \texttt{MA-SPL} can achieve an approximation ratio of $(1-\frac{c}{e})$, $(1-e^{-\alpha})$ or $(\frac{\gamma^{2}(1-e^{-(\beta(1-\gamma)+\gamma^2)})}{\beta(1-\gamma)+\gamma^2})$ with a dynamic regret bound of $\mathcal{O}\left(\sqrt{\frac{\mathcal{P}_{T}T}{1-\tau}}\right)$ to the best comparator in hindsight, respectively. Here, $\mathcal{P}_{T}$ is the deviation of maximizer sequence, $\tau$ is the spectral gap of the network, $\alpha$ denotes the DR ratio, $\beta$ represents the upper submodularity ratio and $\gamma$ is the lower submodularity ratio. %To the best of our knowledge, this is the first result that not only can achieve the tight  $(1-\frac{c}{e})$-approximation to the MA-OC problem with submodular objectives but also successfully handles the unexplored  weakly submodular settings. 
	\item[$\bullet$] To eliminate the dependence of both DR ratio and submodularity ratio of our \texttt{MA-SPL} algorithm, we next present a \emph{parameter-free} online algorithm named \texttt{MA-MPL} for the MA-OC problem. The cornerstone of this \texttt{MA-MPL} algorithm is a novel inequality between our proposed policy-based continuous extension and the original set objective function. Moreover, when the objective function is monotone $\alpha$-weakly DR-submodular or $(\gamma,\beta)$-weakly submodular, our proposed \texttt{MA-MPL} algorithm also can enjoy the same approximation ratio  with a regret bound of $\mathcal{O}\left(d(G)\sqrt{\mathcal{P}_{T}T}\right)$ to the
	best comparator in hindsight, where $d(G)$ is the diameter of the corresponding communication graph $G$.
	\item[$\bullet$] We conduct numerical experiments to verify the effectiveness of our proposed algorithms.
\end{list}
\textbf{Related Work.} Due to space limitations, the comprehensive literature review is placed in \cref{appendix:related_work}. In particular, we present a detailed comparison of our proposed \texttt{MA-SPL} algorithm and \texttt{MA-MPL}  algorithm with existing studies on multi-agent online coordination in \cref{tab:Comparison}.
\begin{table}[t]
\vspace{-1.0em}
	\caption{\small Comparison of the different algorithms for $T$-round MA-OC problem. Note that 	`Approx.' denotes the obtained approximation ratio, `\#Com.' represents the number of communication, `\#Queries' denotes the number of queries to the set objective functions, `D-Regret' denotes the dynamic regret bound, `Proj-free' indicates whether the method does not require projection, `Para-free'  indicates whether the method does not require prior knowledge of curvature $c$ and parameters $\alpha,\gamma,\beta$, $\mathcal{P}_{T}$ is the deviation of maximizer sequence, $\tau$ is the spectral gap of the weight matrix, $d(G)$ is the diameter of the graph $G$, $\phi(\gamma,\beta)\triangleq\beta(1-\gamma)+\gamma^{2}$ and $\kappa\triangleq\sum_{i=1}^{n}\kappa_{i}$.}\label{tab:Comparison}
	\vspace{0.45em}
	\resizebox{1.01\textwidth}{!}{
		\setlength{\tabcolsep}{2.5mm}{
			\begin{tabular}{ccccccccc}
				\toprule[1.3pt]
				Method& Utility& Para-free& Proj-free&Graph&Approx.&\#Com.&\#Queries&D-Regret\\\hline
				OSG~\citep{xu2023online,grimsman2018impact}& Submodular&\ding{52}&\ding{52}&\textbf{Complete}&$\left(\frac{1}{1+c}\right)$&$\mathcal{O}(T)$&$\mathcal{O}(\kappa T)$&$\widetilde{\mathcal{O}}(\sqrt{\mathcal{P}_{T}T})$\\
				MA-OSMA~\citep{zhang2025nearoptimal}& Submodular&\ding{56}&\ding{56}&Connected&$\left(\frac{1-e^{-c}}{c}\right)$&$\mathcal{O}(T)$&$\mathcal{O}(\kappa T)$&$\mathcal{O}\left(\sqrt{\frac{\mathcal{P}_{T}T}{1-\tau}}\right)$\\
				MA-OSEA~\citep{zhang2025nearoptimal}&Submodular&\ding{56}&\ding{52}&Connected&$\left(\frac{1-e^{-c}}{c}\right)$&$\mathcal{O}(T)$&$\mathcal{O}(\kappa T)$&$\widetilde{\mathcal{O}}\left(\sqrt{\frac{\mathcal{P}_{T}T}{1-\tau}}\right)$\\ 
				\midrule[1.3pt]\rowcolor{gray!25}
				\multirow{3}{*}{\Gape{\makecell[c]{\large\texttt{MA-SPL}\\\ }}}&Submodular&\ding{52}&\ding{56}&Connected& $\left(1-\frac{c}{e}\right)$&$\mathcal{O}(T)$&$\mathcal{O}(\kappa T)$& $\mathcal{O}\left(\sqrt{\frac{\mathcal{P}_{T}T}{1-\tau}}\right)$\\\rowcolor{gray!25}
				\cellcolor{gray!25}\large(Algorithm~\ref{alg:SPL})&$\alpha$-weakly DR-Sub&\ding{56}&\ding{56}&Connected& $\left(1-e^{-\alpha}\right)$&$\mathcal{O}(T)$&$\mathcal{O}(\kappa T)$& $\mathcal{O}\left(\sqrt{\frac{\mathcal{P}_{T}T}{1-\tau}}\right)$\\\rowcolor{gray!25}
				&$(\gamma,\beta)$-weakly Sub&\ding{56}&\ding{56}&Connected& $\left(\frac{\gamma^{2}(1-e^{-\phi(\gamma,\beta)})}{\phi(\gamma,\beta)}\right)$&$\mathcal{O}(T)$&$\mathcal{O}(\kappa T)$& $\mathcal{O}\left(\sqrt{\frac{\mathcal{P}_{T}T}{1-\tau}}\right)$\\
				\midrule[1.3pt]\rowcolor{gray!25}
				\multirow{2}{*}{\Gape{\makecell[c]{\large\texttt{MA-MPL}\\\ }}}&$\alpha$-weakly DR-Sub&\ding{52}&\ding{52}&Connected& $\left(1-e^{-\alpha}\right)$&$\mathcal{O}(T^{3/2})$&$\mathcal{O}(\kappa T^{5/2})$& $\mathcal{O}\left(d(G)\sqrt{\mathcal{P}_{T}T}\right)$\\\rowcolor{gray!25}
				\cellcolor{gray!25}\large(Algorithm~\ref{alg:MPL})&$(\gamma,\beta)$-weakly Sub&\ding{52}&\ding{52}&Connected& $\left(\frac{\gamma^{2}(1-e^{-\phi(\gamma,\beta)})}{\phi(\gamma,\beta)}\right)$&$\mathcal{O}(T^{3/2})$&$\mathcal{O}(\kappa T^{5/2})$& $\mathcal{O}\left(d(G)\sqrt{\mathcal{P}_{T}T}\right)$\\
				\midrule[1.3pt]	\end{tabular}}}\vspace{-1.1em}
\end{table}
\section{Problem Setup}\label{sec:problem}
This section will provide a detailed introduction to multi-agent online coordination(MA-OC) problem.

In MA-OC problem, we generally consider a collection of $n$ distinct agents, indexed by the set $\N\triangleq\{1,\dots,n\}$ and interconnected through an undirected network $G(\N,\mathcal{E})$. Here, $\mathcal{E}\subseteq\N\times\N$ represents the possible communication links among agents. Additionally, each agent $i\in\N$ is endowed with a unique set of actions $\V_{i}\triangleq\{v_{i,1},\dots,v_{i,\kappa_{i}}\}$, meaning that these action sets are mutually disjoint, i.e., $\V_{i}\cap\V_{j}=\emptyset$ for any $i,j\in\N$. In the process of multi-agent online coordination, at every time spot $t\in[T]$, each agent $i\in\N$ will separately select one action $a_{i}(t)$ from its own action set $\V_{i}$. After committing to these choices, the environment will reveal a utility set function $f_{t}$ defined over the aggregated action space $\V\triangleq\cup_{i\in\N}\V_{i}$. Then, the agents receive the utility $f_{t}(\cup_{i\in\N}\{a_{i}(t)\})$. As a result, the goal of agents is to maximize their cumulative reward as much as possible. Specifically, at each time step $t\in[T]$, we need to address the following set function maximization problem in a multi-agent collaborative manner: 
\begin{equation}\label{equ_problem}
	\max f_{t}(S),\ \ \text{ s.t.}\ S\subseteq\V\ \text{and}\ |S\cap\V_{i}|\le1,\forall i\in\N.
\end{equation} 
Furthermore, in numerous practical applications regarding  MA-OC problem,  each agent usually has a limited perceptual range, allowing it to sense only the environmental changes in its immediate surroundings. For example, in target tracking scenarios involving a swarm of unmanned aerial vehicles(UAVs), each UAV only can perceive the targets within its sensing radius, leaving those outside this range undetected. To model this limitation, several studies on the MA-OC problem~\citep{zhang2025nearoptimal,xu2023online,xu2023bandit,robey2021optimal,rezazadeh2023distributed} adopt a local feedback model. More specifically, after  $f_{t}$ is revealed, each agent $i\in\N$  is only permitted to query a local marginal oracle $\mathcal{Q}_{t}^{i}:\V_{i}\times2^\V\rightarrow\R_{+}$ defined as  $\mathcal{Q}_{t}^{i}(a,S)\triangleq f_{t}(a|S)\triangleq \big(f_{t}(S\cup\{a\})-f_{t}(S)\big)$  for any $a\in\V_{i}$ and $S\subseteq\V$. This implies that, at each time $t\in[T]$, agents only can receive the marginal evaluations about the actions within their individual action set, rather than the full information of $f_{t}$. We also impose this local-feedback constraint in this paper.

Before going into the details, it is crucial to emphasize that, in many real-world scenarios, there exists such an (approximate) local marginal feedback oracle $\mathcal{Q}_{t}^{i}$ for each agent $i\in\N$ after $f_{t}$ is revealed. Typically, in addition to the decision-making process, agents often utilize various off-the-shelf learning algorithms to estimate the marginal contributions of their available actions based on the observed and collected information (See \citep{corah2021scalable}). Moreover, the local information available to one agent is often insufficient for precisely assessing the actions of other agents who are not in close vicinity. Given this fact, confining each agent to the marginal estimations of the actions within its own action set also can further  reduce the accumulation of learning errors. %There are also notable exceptions, such as the multi-agent task assignment problem formulated in \cite{qu2019distributed}, where each agent $i\in\N$ exactly can obtain the marginal gains of the actions within its own action set $\V_{i}$.

\textbf{Dynamic $\rho$-Regret:} Generally speaking, the set function maximization problem~\eqref{equ_problem} is \textbf{NP}-hard~\citep{natarajan1995sparse, feige1998threshold}, indicating that no polynomial-time algorithms can solve it optimally. Thus, this paper employs the \emph{dynamic $\rho$-regret}~\citep{chen2018online,kakade2007playing,streeter2008online,xu2023online,zhang2025nearoptimal,zinkevich2003online} to measure the performance of our proposed algorithms for MA-OC problem. In \emph{dynamic $\rho$-regret}, the algorithm is compared against a sequence of local maximizers with scale parameter $\rho\in[0,1]$, i.e., $R^{*}_{\rho}(T)\triangleq\rho\sum_{t=1}^{T}f_{t}(\mathcal{A}_{t}^{*})-\sum_{t=1}^{T}f_{t}\big(\cup_{i\in\N}\{a_{i}(t)\}\big)$, where $\mathcal{A}_{t}^{*}$ is the optimal solution of problem~\eqref{equ_problem} and  $a_{i}(t)$ is the action chosen by agent $i$ at time $t$.
\section{Preliminaries}
In this section, we introduce some basic concepts and the frequently used notations.

\textbf{Notations.} For any positive integer $n$, the symbol $[n]$ denotes the set $\{1,\dots, n\}$. $\mathbf{0}_{p}$ and $\mathbf{1}_{p}$ denote the $p$-dimensional vector whose all components are $0$ and $1$, respectively. Moreover, $\|\cdot\|_{1}$ and $\|\cdot\|_{2}$ stand for the $L_{1}$ norm and $L_{2}$  norm for vectors, respectively. We also use $\Delta_{m}$ to represent the standard $m$-dimensional simplex, that is, $\Delta_{m}\triangleq\{(x_{1},\dots,x_{m})|\sum_{i=1}^{m}x_{i}\le 1,\ x_{i}\ge0,\forall i\in[m]\}$.

\textbf{Submodularity and Curvature.} Let $\V$ be a finite ground set and $f:2^{\V}\rightarrow\R_{+}$ be a set function mapping any subset of $\V$ to a non-negative real number. Then, for any two subsets $S,T\subseteq\V$, we  denote by $f(T|S)$ the marginal contribution of adding the elements of $T$ to $S$, i.e., $f(T|S)\triangleq f(T\cup S)-f(S)$. In particular, when $T$ is a singleton set $\{v\}$, we also use $f(v|S)$ to represent $f(\{v\}|S)$. Therefore,  we say a set function $f$ is \emph{submodular} if and only if it satisfies the \emph{diminishing-return property}~\cite{nemhauser1978analysis,fujishige2005submodular,fisher1978analysis}, that is, $f(v|S)\ge f(v|T)$ for any $S\subseteq T\subseteq\V$ and $v\in\V\setminus T$. To precisely characterize the diminishing-return property, \citep{conforti1984submodular,feldman2021guess,sviridenko2017optimal,vondrak2010submodularity} introduced the concept of \emph{curvature} for submodular functions, which is defined as $c\triangleq1-\min_{S\subseteq\V, v\notin S}\frac{f(S\cup\{v\})-f(S)}{f(\{v\})-f(\emptyset)}$.

\textbf{Monotonicity.} A set function $f:2^{\V}\rightarrow\R_{+}$ is \emph{monotone} if and only if $f(S)\le f(T)$ for any $S\subseteq T\subseteq\V$. Moreover, in this paper, we suppose the set function $f$ is \emph{normalized}, that is, $f(\emptyset)=0$.

\textbf{Weak Submodularity.} A set function $f:2^{\V}\rightarrow\R_{+}$ is said to be \emph{$\gamma$-weakly submodular from below} for some $\gamma\in(0,1]$ if and only if $\sum_{v\in T\setminus S}f(v|S)\ge \gamma\big(f(T)-f(S)\big)$ for any two subsets $S\subseteq T\subseteq\V$, where $\gamma$ is called as the \emph{lower submodularity ratio}~\citep{das2011submodular,das2018approximate,chen2018weakly}. Similarly, we also can define the \emph{weak submodularity from above}, that is,  a set function $f:2^{\V}\rightarrow\R_{+}$ is \emph{$\beta$-weakly submodular from above} for some $\beta\ge1$ if and only if $\sum_{v\in T\setminus S}f(v|T-\{v\})\le\beta\big(f(T)-f(S)\big), \forall S\subseteq T\subseteq\V$, where $\beta$ is the \emph{upper submodularity ratio}. When a set function $f$ is both \emph{$\gamma$-weakly submodular from below} and \emph{$\beta$-weakly submodular from above}, we say it is \emph{$(\gamma,\beta)$-weakly submodular}~\citep{thiery2022two}.

\textbf{Weak DR-submodularity.}  A set function $f:2^{\V}\rightarrow\R_{+}$ is \emph{$\alpha$-weakly DR-submodular} for some $\alpha\in(0,1]$ if and only if  $f(v|S)\ge\alpha f(v|T)$ for any two subsets $S\subseteq T\subseteq\V$ and $v\in\V\setminus T$. In particular,  $\alpha$ is often called as the diminishing-return(DR) ratio~\citep{bogunovic2018robust,kuhnle2018fast,gatmiry2018non,GONG202116,JOGO-Lu}. It is worth noting that, from the previous definition of weak submodularity, we can infer that an $\alpha$-weakly DR-submodular function automatically satisfies the conditions for being $(\alpha, \frac{1}{\alpha})$-weakly submodular. %Moreover, when $\alpha=1$,  weakly DR-submodular objectives will reduce to the standard submodular functions.
\section{Policy-based Continuous-Relaxation Framework}\label{sec:Continuous-Relaxation}
Before presenting our proposed algorithms for the MA-OC problem, we firstly explore the offline set function maximization problem~\eqref{equ_problem}. In recent years, compared to discrete optimization, the field of continuous optimization has made significant advancements, yielding a broad spectrum of effective algorithmic frameworks and theoretical tools. Consequently,  one promising strategy to addressing the set function maximization problem~\eqref{equ_problem} is to convert it into a solvable continuous optimization problem throughout \emph{continuous-relaxation} techniques. 

A well-known continuous-relaxation framework  is the \emph{multi-linear extension}~\citep{calinescu2011maximizing}, which was introduced for maximizing submodular set functions. Regrettably, this relaxation framework cannot be directly applied to the general set function maximization problem~\eqref{equ_problem}, as most existing lossless rounding schemes for  \emph{multi-linear extension}—such as pipage rounding~\citep{ageev2004pipage}, swap rounding~\citep{chekuri2010submodular}, and contention resolution~\citep{chekuri2014submodular}—rely heavily on the \emph{submodular} assumption. Note that lossless rounding schemes refer to methods that convert the obtained continuous solution into a feasible discrete solution without any loss in terms of the objective function value. To date, how to losslessly round the \emph{multi-linear extension} of non-submodular set functions, e.g. $(\gamma,\beta)$-weakly submodular and $\alpha$-weakly DR-submodular functions, still remains an open question~\citep{thiery2022two}. To overcome this hurdle, we will introduce an innovative continuous-relaxation technique in the subsequent part of this section.

\subsection{Policy-based Continuous Extension}
From the previous description about the MA-OC problem provided in \cref{sec:problem}, we can view the set function maximization problem~\eqref{equ_problem} as a variant of multi-agent cooperative game~\citep{albrecht2024multi,semsar2009multi,wang2022shaq,zhang2021multi}. Inspired by this viewpoint, we naturally consider whether each agent $i\in\N$ can learn a policy $\uppi_{i}\triangleq(\pi_{i,1},\dots,\pi_{i,\kappa_{i}})$ over its individual action space $\V_{i}\triangleq\{v_{i,1},\dots,v_{i,\kappa_{i}}\}$ and then utilizes this policy $\uppi_{i}$ to make decision, where each $\pi_{i,m}$ represents the probability of agent $i$ taking the action  $v_{i,m},\forall m\in[\kappa_{i}]$. Based on this idea, if letting $a_{i}\in\V_{i}\cup\{\emptyset\}$ denote the random action chosen by each policy $\uppi_{i},\forall i\in\N$, then we can obtain the following policy-based continuous extension, namely, 
\begin{definition}\label{def_extension}	
If $\uppi_{i}\triangleq(\pi_{i,1},\dots,\pi_{i,\kappa_{i}})\in\Delta_{\kappa_{i}}$ for any $i\in\N$, then the policy-based continuous extension $F_{t}:\prod_{i=1}^{n}\Delta_{\kappa_{i}}\rightarrow\R_{+}$ for the set function maximization problem~\eqref{equ_problem} can be  defined as:
	\begin{equation}\label{equ_extension}
		\begin{aligned}
			&F_{t}(\uppi_{1},\dots,\uppi_{n})\triangleq\sum_{a_{i}\in\V_{i}\cup\{\emptyset\},\forall i\in\N}\Big(f_{t}\big(\cup_{i=1}^{n}\{a_{i}\}\big)\prod_{i=1}^{n}p(a_{i}|\uppi_{i})\Big),
		\end{aligned}
	\end{equation}where $p(\cdot|\uppi_{i})$ is a probability distribution  over the set $\V_{i}\cup\{\emptyset\}$, that is, 	$p(v_{i,m}|\uppi_{i}) =\pi_{i,m},\forall i\in[n],\forall m\in[\kappa_{i}]$ and $p(\emptyset | \uppi_{i}) = 1 - \sum_{m=1}^{K_{i}} \pi_{i,m},\forall i\in\N$.
\end{definition}
\begin{remark}
It is noteworthy that, in Eq.\eqref{equ_extension}, with the probability $1 - \sum_{m=1}^{\kappa_{i}} \pi_{i,m}$, the policy $\uppi_{i}$ will not  pick  any action from $\V_{i}$, which means there is a possibility that no action will be chosen, i.e., $\emptyset$. 
\end{remark}
\begin{remark}
The definition in Eq.\eqref{equ_extension} highlights a notable advantage of our proposed policy-based continuous extension: it does not assign probabilities to any subset that violates the constraint of problem~\eqref{equ_problem}. As a result, for any set function $f_{t}$  and  any  $(\uppi_{1},\dots,\uppi_{n})\in\prod_{i=1}^{n}\Delta_{\kappa_{i}}$, throughout the \cref{def_extension}, we can easily generate a subset, i.e.,  $\cup_{i=1}^{n}\{a_{i}\}$, that adheres to the constraints of problem~\eqref{equ_problem} while ensuring $\E\left(f_{t}\left(\cup_{i=1}^{n}\{a_{i}\}\right)\right)=F_{t}(\uppi_{1},\dots,\uppi_{n})$. In contrast, all known lossless rounding schemes of the multi-linear extension require the function $f_{t}$ to be submodular~\citep{calinescu2011maximizing,chekuri2010submodular,chekuri2014submodular}.
\end{remark}
\begin{remark}
Notably, we observe that the works~\citep{sahin2020sets,zhou2025improved} have introduced two relaxation techniques for lattice submodular and $k$-submodular functions, both of which are analogous to our policy-based continuous extension $F_{t}$. However, it is crucial to emphasize that there exist notable differences between our work and \citep{sahin2020sets,zhou2025improved}: \textbf{a)} The lattice formulation typically requires an `order' relationship among different actions of the same agent.  However, in our multi-agent coordination problem, we do not impose any specific order on the decisions. %Without this predefined order, we even cannot ensure that the aforementioned $c_{t}$ is a lattice submodular or $k$-submodular function when $f_{t}$ is submodular. 
Consequently, the results and algorithms from \citep{sahin2020sets,zhou2025improved} are not directly applicable to our scenario. \textbf{b)} In addition to submodularity, our paper also considers weak submodularity and allows for varying sizes of action sets among agents.
\end{remark}
With the policy-based continuous extension $F_{t}$ defined in Eq.\eqref{equ_extension}, the set function maximization problem~\eqref{equ_problem} can be naturally relaxed into a continuous maximization task, i.e., 
\begin{equation}\label{equ_Relaxation}
	\max F_{t}(\uppi_{1},\dots,\uppi_{n}),\ \ \text{ s.t.}\ \ \|\uppi_{i}\|_{1}\le1, \uppi_{i}\in[0,1]^{\kappa_{i}}, \forall i\in\N.
\end{equation} 
In order to effectively tackle the policy optimization problem~\eqref{equ_Relaxation}, we next investigate the properties of our proposed policy-based continuous extension $F_{t}$.
\subsection{Properties of Policy-based Continuous Extension}\label{sec:properties}
This subsection will focus on characterizing the differentiability, monotonicity and submodularity of our proposed policy-based continuous extension. Specifically, we have the following theorem:
\begin{theorem}[Proof in  \cref{appendix_proof_thm}]\label{thm1} The policy-based continuous extension $F_{t}:\prod_{i=1}^{n}\Delta_{\kappa_{i}}\rightarrow\R_{+}$ defined in Eq.\eqref{equ_extension} satisfies the following properties:

	\textbf{1):} For any point $(\uppi_{1},\dots,\uppi_{n})\in\prod_{i=1}^{n}\Delta_{\kappa_{i}}$, the first-order derivative of $F_{t}$ at variable $\pi_{i,m},\forall i\in\N,\forall m\in[\kappa_{i}],$ can be expressed as follows:
	 \begin{equation*}
		\frac{\partial F_{t}}{\partial \pi_{i,m}}(\uppi_{1},\dots,\uppi_{n})\triangleq\E_{a_{j}\sim\uppi_{j},\forall j\in\N}\Big(f_{t}\big(v_{i,m}\big|\cup_{j\neq i,j\in\N}\{a_{j}\}\big)\Big),
	\end{equation*} where $a_{j}\sim\uppi_{j}$ indicates that action $a_{j}$
	is  randomly selected from $\V_{j}\cup\{\emptyset\}$ based on the policy $\uppi_{j}$;

\textbf{2):} If the set function $f_{t}$ is monotone, then $\frac{\partial F_{t}}{\partial \pi_{i,m}}(\uppi_{1},\dots,\uppi_{n})\ge 0$ for any point $(\uppi_{1},\dots,\uppi_{n})\in\prod_{i=1}^{n}\Delta_{\kappa_{i}}$, $ i\in\N$ and $m\in[\kappa_{i}]$, which means the monotonicity of $f_{t}$ can be inherited by $F_{t}$;

\textbf{3):} If $f_{t}$ is $\alpha$-weakly DR-submodular, then $F_{t}$ is $\alpha$-weakly continuous DR-submodular~\citep{hassani2017gradient,pedramfar2024from} over $\prod_{i=1}^{n}\Delta_{\kappa_{i}}$, that is, for any two point $(\uppi^{a}_{1},\dots,\uppi^{a}_{n})\in\prod_{i=1}^{n}\Delta_{\kappa_{i}}$ and $(\uppi^{b}_{1},\dots,\uppi^{b}_{n})\in\prod_{i=1}^{n}\Delta_{\kappa_{i}}$, if  $\uppi^{a}_{i}\le\uppi^{b}_{i}\ \forall i\in\N$, we have that $\nabla F_{t}(\uppi^{a}_{1},\dots,\uppi^{a}_{n})\ge \alpha \nabla F_{t}(\uppi^{b}_{1},\dots,\uppi^{b}_{n})$;

\textbf{4):}  For any subset $S$ within the constraint of problem~\eqref{equ_problem} and any point $(\uppi_{1},\dots,\uppi_{n})\in\prod_{i=1}^{n}\Delta_{\kappa_{i}}$, when $f_{t}$ is monotone $\alpha$-weakly DR-submodular, the following inequality holds:
\begin{equation}\label{equ:thm1.4.1}
	\alpha\Big(f_{t}(S)-F_{t}(\uppi_{1},\dots,\uppi_{n})\Big)\le\sum_{(i,m): v_{i,m}\in S}\frac{\partial F_{t}}{\partial \pi_{i,m}}(\uppi_{1},\dots,\uppi_{n}),
\end{equation} where $\{(i,m): v_{i,m}\in S\}$ denotes the set of all indices $(i,m)$ such that $ v_{i,m}\in S$. Similarly, when $f_{t}$ is monotone $(\gamma,\beta)$-weakly submodular, we can show that
\begin{equation}\label{equ:thm1.4.2}
		\Big(\gamma^{2}f_{t}(S)-(\beta(1-\gamma)+\gamma^{2})F_{t}(\uppi_{1},\dots,\uppi_{n})\Big)\le\sum_{(i,m): v_{i,m}\in S}\frac{\partial F_{t}}{\partial \pi_{i,m}}(\uppi_{1},\dots,\uppi_{n}).
\end{equation}
	%\item  For any vector $(\uppi_{1},\dots,\uppi_{n})\in\prod_{i=1}^{n}\Delta_{K_{i}}$, the second-order derivative of $F_{t}$ at variables $\pi_{i_{1},m_{1}},\pi_{i_{2},m_{2}},\forall i_{1}\in\N,\forall m_{1}\in[K_{i_{1}}],\forall i_{2}\in\N,\forall m_{2}\in[K_{i_{2}}],$ can be defined as:
	  
	%\textbf{i): When $i_{1}\neq i_{2}$},
	%\begin{equation*}
		%\frac{\partial^{2}F_{t}}{\partial \pi_{i_{1},m_{1}}\partial \pi_{i_{2},m_{2}}}\triangleq\kappa_{i_{1}}\kappa_{i_{2}}\E_{a_{\hat{i},\hat{k}}\sim\uppi_{\hat{i}},\forall \hat{i}\in\N, \hat{k}\in[\kappa_{\hat{i}}]}\Big(f_{t}\big(v_{i_{1},m_{1}}\big|S\cup\{v_{i_{2},m_{2}}\}\big)-f_{t}\big(v_{i_{1},m_{1}}\big|S\big)\Big),
	%\end{equation*}
	%where $S\triangleq\cup_{(\hat{i},\hat{k})\neq\left\{(i_{1},1),(i_{2},1)\right\}}\{a_{\hat{i},\hat{k}}\}$ denotes the union of all actions except $a_{i_{1},1}$ and $a_{i_{2},1}$; \textbf{ii): When $i_{1}= i_{2}$}, if we set $S'\triangleq\cup_{(\hat{i},\hat{k})\neq\left\{(i_{1},1),(i_{1},2)\right\}}\{a_{\hat{i},\hat{k}}\}$,
	%\begin{equation*}
		%\frac{\partial^{2}F_{t}}{\partial \pi_{i_{1},m_{1}}\partial \pi_{i_{1},m_{2}}}\triangleq\kappa_{i_{1}}(\kappa_{i_{1}}-1)\E_{a_{\hat{i},\hat{k}}\sim\uppi_{\hat{i}},\forall \hat{i}\in\N, \hat{k}\in[\kappa_{\hat{i}}]}\Big(f_{t}\big(v_{i_{1},m_{1}}\big|S'\cup\{v_{i_{1},m_{2}}\}\big)-f_{t}\big(v_{i_{1},m_{1}}\big|S'\big)\Big).
	%\end{equation*} 
\end{theorem}
\begin{remark}\label{remark:different}
Part $\textbf{3)}$ indicates that when $f_{t}$ is submodular, namely $\alpha=1$,  our policy-based continuous extension $F_{t}$
	belongs to the well-studied continuous DR-submodular functions. However, it is important to  highlight that these existing results for continuous DR-submodular maximization~\citep{hassani2020stochastic,hassani2017gradient,pedramfar2024unified,pedramfar2023a,wan2023bandit,zhang2022boosting} cannot be directly extended to our policy-based continuous extension $F_{t}$. This is because all of them heavily rely on the  inequality $\langle \mathbf{y} - \mathbf{x}, \nabla G(\mathbf{x}) \rangle\ge G(\mathbf{y} \vee \mathbf{x}) - G(\mathbf{x})$ where  $G$ is a monotone continuous DR-submodular function and $\vee$ denotes coordinate-wise maximum operation. Note that the previous inequality requires that the domain of $G$ must be closed under the maximum operation $\vee$. However, the domain $\prod_{i=1}^{n}\Delta_{\kappa_{i}}$ of our proposed $F_{t}$ does not meet this requirement.
\end{remark}
\section{Multi-Agent Policy Learning}
This section explores how to utilize the policy-based continuous extension introduced in \cref{sec:Continuous-Relaxation} to address our concerned MA-OC problem. Broadly speaking, a notable advantage of continuous-relaxation techniques is that they enable the use of gradient-based methods, such as gradient ascent~\citep{bertsekas2015convex,lan2020first,nesterov2013introductory} and Frank-Wolfe method~\citep{braun2022conditional,lacoste2016convergence}, to tackle discrete optimization problems. Moreover, as is well established in the literature~\citep{allen2018make,drori2020complexity,ghadimi2013stochastic,lacoste2016convergence}, under mild conditions, a wide range of gradient-based algorithms can converge to the stationary points of their target objectives.  Motivated by these findings, we next investigate the stationary points of our proposed policy-based continuous extension $F_{t}$.
\subsection{Stationary Points and Surrogate Functions}\label{sec:surrogate_function}
At first, we recall the definition of stationary points for maximization problem, that is,  
\begin{definition}\label{def:stationary} Given a differentiable function $G:\mathcal{K}\rightarrow\R$ and a domain $\C\subseteq\mathcal{K}$, a point $\x\in\C$ is called as a stationary point for the function $G$ over the domain $\C$ if and only if $\max_{\y\in\C}\langle\y-\x,\nabla G(\x)\rangle\le 0$.
\end{definition}
 Next, we examine the performance of the stationary points of our proposed policy-based continuous extension $F_{t}$ relative to the maximum value of problem~\eqref{equ_problem}. Specifically, we have that
\begin{theorem}[Proof in \cref{appendix_proof_thm0}]\label{thm2}
Given a set function $f_{t}:2^{\V}\rightarrow\R_{+}$, if $(\uppi^{s}_{1},\dots,\uppi^{s}_{n})$ is a stationary point  of its policy-based continuous extension $F_{t}$ over the domain $\prod_{i=1}^{n}\Delta_{\kappa_{i}}$ and $S^{*}$ denotes the optimal solution of the corresponding maximization problem~\eqref{equ_problem}, then the following inequalities hold: 

\textbf{1):} When $f_{t}$ is  monotone submodular with curvature $c\in[0,1]$, $F_{t}(\uppi^{s}_{1},\dots,\uppi^{s}_{n})\ge\Big(\frac{1}{1+c}\Big)f_{t}(S^{*})$;

\textbf{2):} When  $f_{t}$ is  monotone $\alpha$-weakly DR-submodular, $F_{t}(\uppi^{s}_{1},\dots,\uppi^{s}_{n})\ge\Big(\frac{\alpha^{2}}{1+\alpha^{2}}\Big)f_{t}(S^{*})$;

\textbf{3):} When  $f_{t}$ is  monotone $(\gamma,\beta)$-weakly submodular, $F_{t}(\uppi^{s}_{1},\dots,\uppi^{s}_{n})\ge\Big(\frac{\gamma^{2}}{\beta+\beta(1-\gamma)+\gamma^{2}}\Big)f_{t}(S^{*})$.
\end{theorem}
\begin{remark}
It is worth noting that, in a certain sense, the approximation guarantees established in Theorem~\ref{thm2} is \textbf{tight}. In \cref{appendix:special_case}, we will present a simple instance of a submodular function $f_{t}$, i.e., $c=\alpha=\gamma=\beta=1$, whose policy-based continuous extension $F_{t}$ can attain a $\frac{1}{2}$-approximation guarantee at a stationary point that is also a local maximum.
\end{remark}
 Theorem~\ref{thm2} indicates that when the original set function $f_{t}$ is monotone  submodular with curvature $c$, $\alpha$-weakly DR-submodular or $(\gamma,\beta)$-weakly submodular, the direct application of gradient-based methods on our proposed policy-based continuous extension $F_{t}$ can ensure an approximation ratio of $(\frac{1}{1+c})$, $(\frac{\alpha^{2}}{1+\alpha^{2}})$ or $(\frac{\gamma^{2}}{\beta+\beta(1-\gamma)+\gamma^{2}})$ to the problem~\eqref{equ_problem}, respectively. However, as shown in \citep{sviridenko2017optimal,feldman2021guess}, the optimal approximation ratio for maximizing a monotone submodular function with curvature $c$ is $(1-\frac{c}{e})$, which significantly  exceeds the previous $(\frac{1}{1+c})$-approximation provided by stationary points of $F_{t}$. Similarly, \citep{harshaw2019submodular} pointed out that the optimal approximation for $\alpha$-weakly DR-submodular maximization is $(1-e^{-\alpha})$, which is also greater than the  $(\frac{\alpha^{2}}{1+\alpha^{2}})$-approximation established in \cref{thm2}. Motivated by these findings, we naturally
 wonder \emph{whether it is possible to enhance the approximation guarantees of the stationary points of our proposed policy-based continuous extension.}

	 Previously, the \emph{multi-linear extension} of submodular  functions also encountered a similar issue, namely, the stationary points of the \emph{multi-linear extension} only can guarantee a \emph{sub-optimal} $\frac{1}{2}$-approximation~\citep{filmus2012power,hassani2017gradient}. To overcome this drawback, prior studies~\citep{filmus2012tight,filmus2014monotone,zhang2022boosting,zhang2025nearoptimal} constructed a novel surrogate function for the \emph{multi-linear extension} and proved that the stationary points of this surrogate function can achieve the optimal $(1-\frac{1}{e})$-approximation. Inspired by this idea, we also hope to develop a surrogate function for our proposed policy-based continuous extension $F_{t}$ such that it can improve the approximation ratios of the stationary points of $F_{t}$. Specifically, we have that
	 \begin{theorem}[Proof in \cref{appendix_proof_thm0.5}]\label{thm4}Similar to \citep{zhang2022boosting,zhang2024boosting,zhang2025nearoptimal}, for any given policy-based continuous extension $F_{t}$ introduced in \cref{def_extension}, we consider a surrogate function $F_{t}^{s}:\prod_{i=1}^{n}\Delta_{K_{i}}\rightarrow\R_{+}$ whose gradient at each point $\x\in\prod_{i=1}^{n}\Delta_{\kappa_{i}}$ is a weighted average of the gradient $\nabla F_{t}(z*\x)$, i.e., $\nabla F_{t}^{s}(\x)=\int_{0}^{1} w(z)\nabla F_{t}(z*\boldsymbol{x})\mathrm{d}z$ where $w(z)$ is a positive weight function over $[0,1]$. After elaborately designing the weight function $w(z)$, we can show that:
	 	
		\textbf{1):} When $f_{t}$ is $\alpha$-weakly DR-submodular and $w(z)=e^{\alpha(z-1)}$, for any stationary point $(\uppi^{s}_{1},\dots,\uppi^{s}_{n})$ of the surrogate objective $F_{t}^{s}$ over $\prod_{i=1}^{n}\Delta_{\kappa_{i}}$, then  we have $F_{t}(\uppi^{s}_{1},\dots,\uppi^{s}_{n})\ge\big(1-e^{-\alpha}\big)f_{t}(S^{*})$;
	 	
\textbf{2):} When  $f_{t}$ is  $(\gamma,\beta)$-weakly submodular and $w(z)=e^{\phi(\gamma,\beta)(z-1)}$ where $\phi(\gamma,\beta)=\beta(1-\gamma)+\gamma^2$, for any stationary point $(\uppi^{s}_{1},\dots,\uppi^{s}_{n})$ of the surrogate objective $F_{t}^{s}$ over the domain $\prod_{i=1}^{n}\Delta_{\kappa_{i}}$, then we can show that $F_{t}(\uppi^{s}_{1},\dots,\uppi^{s}_{n})\ge\big(\frac{\gamma^{2}(1-e^{-\phi(\gamma,\beta)})}{\phi(\gamma,\beta)}\big)f_{t}(S^{*})$;

	 \textbf{3):} When $f_{t}$ is submodular with curvature $c$ and $w(z)=e^{z-1}$, for any stationary point $(\uppi^{s}_{1},\dots,\uppi^{s}_{n})$ of the objective $(F_{t}^{s}+\frac{G_{t}}{e})$ over $\prod_{i=1}^{n}\Delta_{\kappa_{i}}$, then we have $F_{t}(\uppi^{s}_{1},\dots,\uppi^{s}_{n})\ge\big(1-\frac{c}{e}\big)f_{t}(S^{*})$, where $S^{*}$ is the optimal subset of problem~\eqref{equ_problem} and $G_{t}(\uppi_{1},\dots,\uppi_{n})\triangleq \sum_{i=1}^{n}\sum_{m=1}^{\kappa_{i}}\big(f_{t}\big(v_{i,m}|\V-\{v_{i,m}\}\big)\big)\pi_{i,m}$.
	 \end{theorem}
	 \begin{remark} 
	 	Note that part $\textbf{3)}$ of \cref{thm4} considers the objective $(F_{t}^{s}+\frac{G_{t}}{e})$ instead of the surrogate $F_{t}^{s}$ alone. Here, $G_{t}$ is a special linear linear function related to the minimum margin gains of the corresponding submodular objective $f_{t}$, that is, $f_{t}\big(a|\V-\{a\}\big),\forall a\in\V$. Furthermore, it is important to emphasize that the proof of Theorem~\ref{thm4} is not a parallel copy of the  articles~\citep{zhang2022boosting,zhang2025nearoptimal} regarding the multi-linear extension. This is because the works~\citep{zhang2022boosting,zhang2025nearoptimal} also utilized the same inequality in \cref{remark:different},which requires our domain to be closed under the operation $\vee$. However, the domain of our $F_{t}$ does not satisfy this condition. \textbf{As a result, new techniques are required to verify \cref{thm4}}.
	 \end{remark}

%It is important to emphasize that the proof of Theorem~\ref{thm4} is not a parallel copy of the  articles~\citep{zhang2022boosting,zhang2024boosting,zhang2025nearoptimal} regarding multi-linear extension. This is because \citep{zhang2022boosting,zhang2024boosting,zhang2025nearoptimal} also utilized the inequality in \cref{remark:different},which requires our domain to be closed under the operation $\vee$. However, our proposed policy-based continuous extension $F_{t}$ does not satisfy this condition. As a result, new techniques and inequalities  are required to verify \cref{thm4}.
\subsection{Multi-Agent Policy Learning via Surrogate Functions }
\begin{algorithm}[t]
	\caption{Multi-Agent Surrogate Policy Learning(\texttt{MA-SPL})}\label{alg:SPL}
	\KwIn{Time horizon $T$, weight function $w(z)$, action set $\V_{i}\triangleq\{v_{i,1},\dots,v_{i,\kappa_{i}}\}$,
		weight matrix $\W\triangleq[w_{ij}]_{n\times n}$, parameters $(\alpha,\gamma,\beta)$ and step size $\eta_{t},\forall t\in[T]$}
	\tcp{\textcolor{teal}{Policy Initialization (Lines 1-2)}}
	Initialize a policy vector $(\uppi_{i,1}(1),\dots,\uppi_{i,n}(1))$ for each agent $i\in\N$\; 
	Set the policy $\uppi_{i,i}(1)\triangleq\frac{1}{\kappa_{i}}\one_{\kappa_{i}}$ and $\uppi_{i,j}(1)\triangleq\boldsymbol{0}_{\kappa_{j}}$ when $i\neq j$ for any $i\in\N$\;
	\For{$t=1,\dots,T$}{
		\For{each agent $i\in\N$}{
			\tcp{\textcolor{teal}{Actions Sampling and Information Exchange (Lines 5-8)}}
			Compute the normalized policy $\p_{i}(t)\triangleq\frac{\uppi_{i,i}(t)}{\|\uppi_{i,i}(t)\|_{1}}$\;
			Utilize the normalized policy $\p_{i}(t)$ to sample an action $a_{i}(t)$ from $\V_{i}$\;
			Agent $i$ executes the sampled action  $a_{i}(t)$\;
			Exchange policy vector $\big(\uppi_{i,1}(t),\dots,\uppi_{i,n}(t)\big)$ with the neighboring agent $j\in\mathcal{N}_{i}$\;
			\tcp{\textcolor{teal}{Surrogate Gradient Estimation (Lines 9-16)}}
			Generate a random number $z_{i}(t)$ from r.v. $\mathcal{Z}$ where $\text{Pr}(\mathcal{Z}\le z)\triangleq\frac{\int_{0}^{z}w(a)\mathrm{d}a}{\int_{0}^{1}w(a)\mathrm{d}a},\forall z\in[0,1]$\;
			\For{$j=1,\dots,n$}{
				Utilize the weighted policy $z_{i}(t)*\uppi_{i,j}$ to sample an action $\widetilde{a}_{j}(t)\in\V_{i}\cup\{\emptyset\}$;}
			Compute the random subset of actions $S_{i}(t):=\cup_{j\neq i,j\in\N}\{\widetilde{a}_{j}(t)\}$\;
			Estimate the derivatives of the surrogate function based on \cref{thm1} and \ref{thm4}, i.e.,
			\begin{equation*}
				\widehat{\frac{\partial F^{s}_{t}}{\partial \pi_{i,m}}}(\uppi_{i,1}(t),\dots,\uppi_{i,n}(t)\big)\triangleq d_{i,m}(t)\triangleq(\int_{z=0}^{1}w(z)\mathrm{d}z)f_{t}\big(v_{i,m}\big|S_{i}(t)\big),\forall m\in[\kappa_{i}];
			\end{equation*}\vspace{-0.8em}\\
			\If{$f_{t}$ is submodular}{
				Update $d_{i,m}(t)=\big(d_{i,m}(t)+e^{-1}f_{t}(v_{i,m}|\V-\{v_{i,m}\})\big),\forall m\in[\kappa_{i}]$\;}
			Aggregate the surrogate gradient estimations $\big(d_{i,1}(t),\dots,d_{i,\kappa_{i}}(t)\big)$ as vector $\mathbf{d}_{i}(t)$\;
			\tcp{\textcolor{teal}{Policy Update (Lines 17-19)}}
			Update $\uppi_{i,j}(t+1)\triangleq\sum_{k\in\mathcal{N}_{i}\cup\{i\}}w_{ik}\uppi_{k,j}(t)$ for any $j\neq i$ and $j\in\N$\;
			Compute $\y_{i,i}(t+1)\triangleq\sum_{j\in\mathcal{N}_{i}\cup\{i\}}w_{ij}\uppi_{j,i}(t)+\eta_{t}\mathbf{d}_{i}(t)$\;
			Update $\uppi_{i,i}(t+1)\triangleq\mathop{\arg\min}_{\mathbf{b}\in\Delta_{\kappa_{i}}}\|\mathbf{b}-\y_{i,i}(t+1)\|_{2}$;}}
\end{algorithm}

\cref{thm4} suggests that, when the original set function $f_{t}$ is monotone submodular with curvature $c$, $\alpha$-weakly DR-submodular or  $(\gamma,\beta)$-weakly submodular,  the direct application of gradient-based
methods targeting stationary points to the surrogate function $F^{s}_{t}$ or its variant $(F_{t}^{s}+\frac{G_{t}}{e})$ can achieve a tight approximation ratio of $(1-\frac{c}{e})$, $(1-e^{-\alpha})$ or $\big(\frac{\gamma^{2}(1-e^{-(\beta(1-\gamma)+\gamma^2)})}{\beta(1-\gamma)+\gamma^2}\big)$ to the set function maximization problem~\eqref{equ_problem}, respectively. Furthermore, recent study~\citep{zhang2025nearoptimal} developed an effective online algorithm for the multi-agent \emph{submodular} coordination problem based on the well-studied consensus technique~\citep{nedic2009distributed,shahrampour2017distributed,yuan2024multi,yuan2016convergence} and the surrogate functions of \emph{multi-linear extension}~\citep{zhang2022boosting,zhang2025nearoptimal}. Thus, in order to address the unexplored weakly submodular scenarios and simultaneously achieve the optimal $(1-\frac{c}{e})$-approximation for submodular settings, we naturally consider replacing the surrogate functions of \emph{multi-linear extension} in the algorithms of \citep{zhang2025nearoptimal} with those of our proposed policy-based continuous extension $F_{t}$. Motivated by this idea, we then present a general online algorithm named \texttt{MA-SPL} for MA-OC problem, with details presented in \cref{alg:SPL}.

In \cref{alg:SPL}, at every time step $t\in[T]$, each agent $i \in\N$ will maintain a local policy vector $\big(\uppi_{i,1}(t),\dots,\uppi_{i,n}(t)\big)$. Here, $\uppi_{i,i}(t)$ represents the policy being executed by agent $i$, while $\uppi_{i,j}(t),j\neq i$ reflects agent $i$'s current estimate of the policy $\uppi_{j,j}(t)$ being taken by other agent $j$. After that, each agent $i$ selects an action $a_{i}(t)$ from $\V_{i}$ based on the normalized policy $\p_{i}(t)\triangleq\frac{\uppi_{i,i}(t)}{\|\uppi_{i,i}(t)\|_{1}}$ and shares its local policy vector with the neighboring agent $j\in\N_{i}$, where $\N_{i}$ denotes the neighbors of agent $i$. Then, according to the results of \cref{thm1} and \cref{thm4}, each agent $i$ estimates the first-order partial derivatives of our proposed surrogate function $F_{t}^{s}$ at every coordinate $\pi_{i,m}$. Specifically,  agent $i$ initially samples a random number $z_{i}(t)$ from the random variable $\mathcal{Z}$ with distribution $\text{Pr}(\mathcal{Z}\le z)\triangleq\frac{\int_{0}^{z}w(a)\mathrm{d}a}{\int_{0}^{1}w(a)\mathrm{d}a},\forall z\in[0,1]$ and then approximates each $\frac{\partial F^{s}_{t}}{\partial \pi_{i,m}}(\uppi_{i,1}(t),\dots,\uppi_{i,n}(t)\big)$ by $d_{i,m}(t)\triangleq(\int_{z=0}^{1}w(z)\mathrm{d}z)f_{t}\big(v_{i,m}\big|S_{i}(t)\big)$ where $S_{i}(t)$ is a random set generated from the weighted policy vector $z_{i}(t)*\big(\uppi_{i,1}(t),\dots,\uppi_{i,n}(t)\big)$. Furthermore, when the set objective function $f_{t}$ is \emph{submodular}, \cref{alg:SPL} will further adjust each surrogate gradient estimation $d_{i,m}(t)$ with the minimum marginal contribution $f_{t}(v_{i,m}|\V-\{v_{i,m}\}),\forall m\in[\kappa_{i}]$ (See Lines 14-15). Finally, each agent $i$ updates its policy $\uppi_{i,i}(t)$ throughout a projected ascent along the direction $\mathbf{d}_{i}(t)\triangleq\big(d_{i,1}(t),\dots,d_{i,\kappa_{i}}(t)\big)$.% It is worth noting that in \cref{alg:SPL}, each agent $i$ only needs to evaluate the margin contributions of the actions in $\V_{i}$.

In sharp contrast with the previous MA-OSMA and MA-OSEA algorithms in \citep{zhang2025nearoptimal}, the key innovation of \cref{alg:SPL} lies in the use of our proposed policy-based continuous extension and its surrogate functions to update the policy vector in Lines 9-19, rather than relying on the well-studied \emph{multi-linear extension}~\citep{calinescu2011maximizing}. Moreover, our proposed policy-based continuous extension naturally aligns with the actions sampling process in Lines 5-7, ensuring that, for any set function $f_{t}$, the function value at the executed actions $\cup_{i\in\N}\{a_{i}(t)\}$ is at least as large as the expected function value  $F_{t}(\uppi_{1,1},\dots,\uppi_{n,n})$, namely, $\E\left(f_{t}(\cup_{i\in\N}\{a_{i}(t)\})\right)\ge F_{t}(\uppi_{1,1},\dots,\uppi_{n,n})$. This is a significant advantage over the \emph{multi-linear extension}, as the latter generally requires the \emph{submodular} assumption to guarantee the losslessness of the action sampling of Lines 5-7 (See Lemma 13 in \citep{zhang2025nearoptimal}).

%Due to space limitations, the detailed description of the previously mentioned \texttt{MA-SPL} algorithm are provided in \cref{alg:SPL} of \cref{appendix:alg1}. In contrast with the algorithms in \citep{zhang2025nearoptimal}, the primary innovation of \texttt{MA-SPL} algorithm lies in the use of our proposed policy-based continuous extension and its surrogate functions to update the policy vector in Lines 9-16, rather than relying on the well-established \emph{multi-linear extension}. Furthermore, our proposed policy-based continuous extension naturally aligns with the actions sampling process in Lines 5-7, ensuring that, for any set function $f_{t}$, the function value at the executed actions $\cup_{i\in\N}\{a_{i}(t)\}$ is at least as large as the expected function value  $F_{t}(\uppi_{1,1},\dots,\uppi_{n,n})$, namely, $\E\left(f_{t}(\cup_{i\in\N}\{a_{i}(t)\})\right)\ge F_{t}(\uppi_{1,1},\dots,\uppi_{n,n})$. This is a significant advantage over the \emph{multi-linear extension}, as the latter generally requires the \emph{submodular} assumption to guarantee the losslessness of the action sampling of Lines 5-7 (See Lemma 13 in \citep{zhang2025nearoptimal}). 

Next, we provide the theoretical analysis for the proposed \cref{alg:SPL}. Before that, we introduce some standard assumptions about the communication graph $G(\N,\mathcal{E})$ and the weight matrix $\mathbf{W}$, i.e.,

\begin{assumption}\label{ass:1}
The graph $G(\N,\mathcal{E})$ is connected. Furthermore, the weight matrix $\mathbf{W}\triangleq[w_{ij}]_{n\times n}\in\R_{+}^{n\times n}$ is symmetric and doubly 
	stochastic, namely, $\W^{T}=\W$ and $\W\one_{n}=\one_{n}$. That is to say, $\tau<1$ where $\tau=\max(|\lambda_{2}(\W)|,|\lambda_{n}(\W)|)$ is the second largest magnitude of the eigenvalues of $\W$ and  $\lambda_{i}(\W)$ denotes the $i$-th largest eigenvalue of matrix $\W$.
\end{assumption}
With this \cref{ass:1}, we then can get the following convergence results for our \cref{alg:SPL}, i.e.,
\begin{theorem}[Proof  in \cref{appendix_proof_thm1}]\label{thm:result1}
	Under \cref{ass:1}, when each set objective function $f_{t}$ is monotone submodular with curvature $c$, $\alpha$-weakly DR-submodular or  $(\gamma,\beta)$-weakly submodular, if we set the weight function $w(z)$ according to \cref{thm4} and choose $\eta_{t}=\mathcal{O}\big(\sqrt{\frac{(1-\tau)P_{T}}{T}}\big)$ where $P_{T}\triangleq\sum_{t=2}^{T}|\mathcal{A}_{t}^{*}\triangle\mathcal{A}_{t-1}^{*}|$ is the deviation of maximizer sequence and the symbol $\triangle$ denotes the symmetric difference, namely, $S\triangle T=(S\setminus T)\cup(T\setminus S)$,  then our proposed \cref{alg:SPL} can achieve a dynamic $\rho$-regret  bound of $\mathcal{O}\left(\sqrt{\frac{P_{T}T}{1-\tau}}\right)$, that is, $\E\left(R_{\rho}^{*}(T)\right)\le\mathcal{O}\left(\sqrt{\frac{P_{T}T}{1-\tau}}\right)$, where $\rho=(1-\frac{c}{e})$, $\rho=(1-e^{-\alpha})$ or $\rho=\big(\frac{\gamma^{2}(1-e^{-(\beta(1-\gamma)+\gamma^2)})}{\beta(1-\gamma)+\gamma^2}\big)$, respectively.
\end{theorem}
\begin{remark}
To the best of our knowledge, this is the first result that achieves the tight $(1-\frac{c}{e})$-approximation for the MA-OC problem with submodular objectives and simultaneously can tackle the previously unexplored  $(\gamma,\beta)$-weakly submodular and $\alpha$-weakly DR-submodular scenarios.
\end{remark}
\begin{remark}\label{remark:limitation}
Note that when considering the weakly submodular objectives, Line 9 of \cref{alg:SPL} requires  prior knowledge of the ratios $\alpha,\gamma,\beta$ to set the weight function $w(z)$. However, in general, accurately computing these parameters will incur exponential computations. To overcome this drawback, we further present a \textbf{parameter-free} online algorithm named \texttt{MA-MPL} in \cref{appendix:alg2}.
\end{remark}
\section{Numerical Experiments}
In this section, we validate the effectiveness of our proposed \texttt{MA-SPL} and \texttt{MA-MPL} algorithms via two different multi-target tracking scenarios. Especially in \cref{graph_show_1},\ref{graph_show_2},\ref{graph_show_3}, we consider a \emph{submodular} facility-location objective function~\citep{xu2023online,zhang2025nearoptimal} with different proportions of `\underline{R}andom', `\underline{A}dversarial' and `\underline{P}olyline' targets. According to the results in \cref{graph_show_1}-\ref{graph_show_3}, we can find that the average utility of our \texttt{MA-SPL} can significantly exceed the state-of-the-art MA-OSMA and MA-OSEA algorithms in \citep{zhang2025nearoptimal}, which is consistent with our  \cref{thm:result1}. Note that the suffixes in \cref{graph_show_1}-\ref{graph_show_4} represent two different choices for communication graphs, where `c' stands for a complete graph and `r' denotes an Erdos-Renyi random graph with average degree $4$.  In contrast, \cref{graph_show_4}-\ref{graph_show_6} adopt a bayesian A-optimal criterion for the target tracking task, which will lead to a $\alpha$-weakly DR-submodular utility function~\citep{harshaw2019submodular,hashemi2019submodular,thiery2022two}. Given the unknown DR ratio $\alpha$, in \cref{graph_show_5}-\ref{graph_show_6}, we perform a brute-force $0.1$-network search to find the  optimal parameter settings for \texttt{MA-SPL} algorithm. Subsequently, we report the best $\alpha=0.1$ case and the worse $\alpha=1$ scenario in \cref{graph_show_4}. Similarly, from \cref{graph_show_4}, we also find that our proposed \texttt{MA-MPL} and  \texttt{MA-SPL} can substantially outperform the `RANDOM' baseline, which is in accord with our theoretical findings. Particularly in `RANDOM' baseline, we let  each agent $i$ randomly execute an action from its own action set $\V_{i}$. Due to space limitations, more discussions about experiment setups and results are presented in \cref{appendix:experiments}.
\begin{figure*}[h]
	\subfigure[\tiny`R':`A':`P'=8:1:1\label{graph_show_1}]{\includegraphics[scale=0.0925]{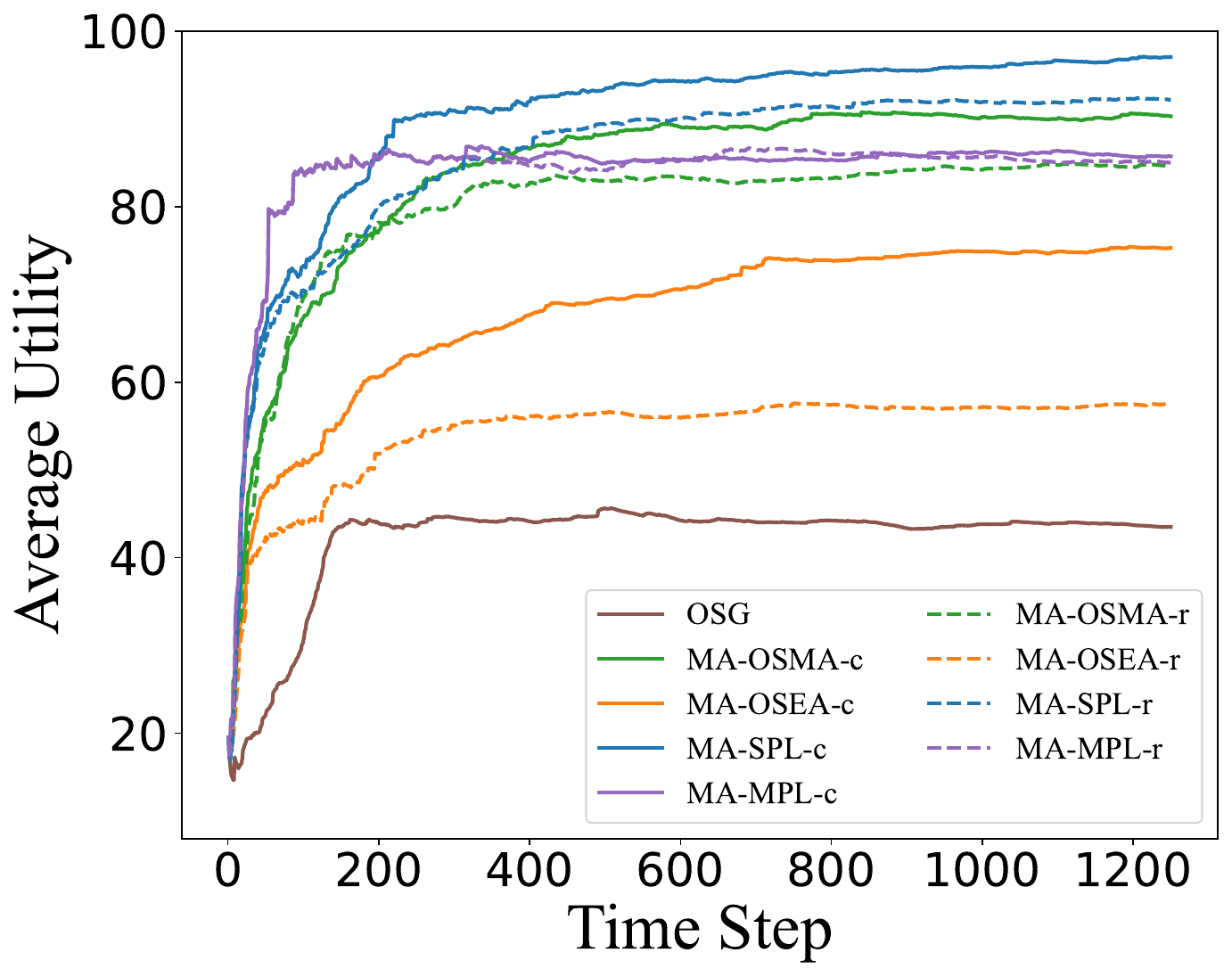}}
	\subfigure[\tiny `R':`A':`P'=6:3:1\label{graph_show_2}]{\includegraphics[scale=0.0925]{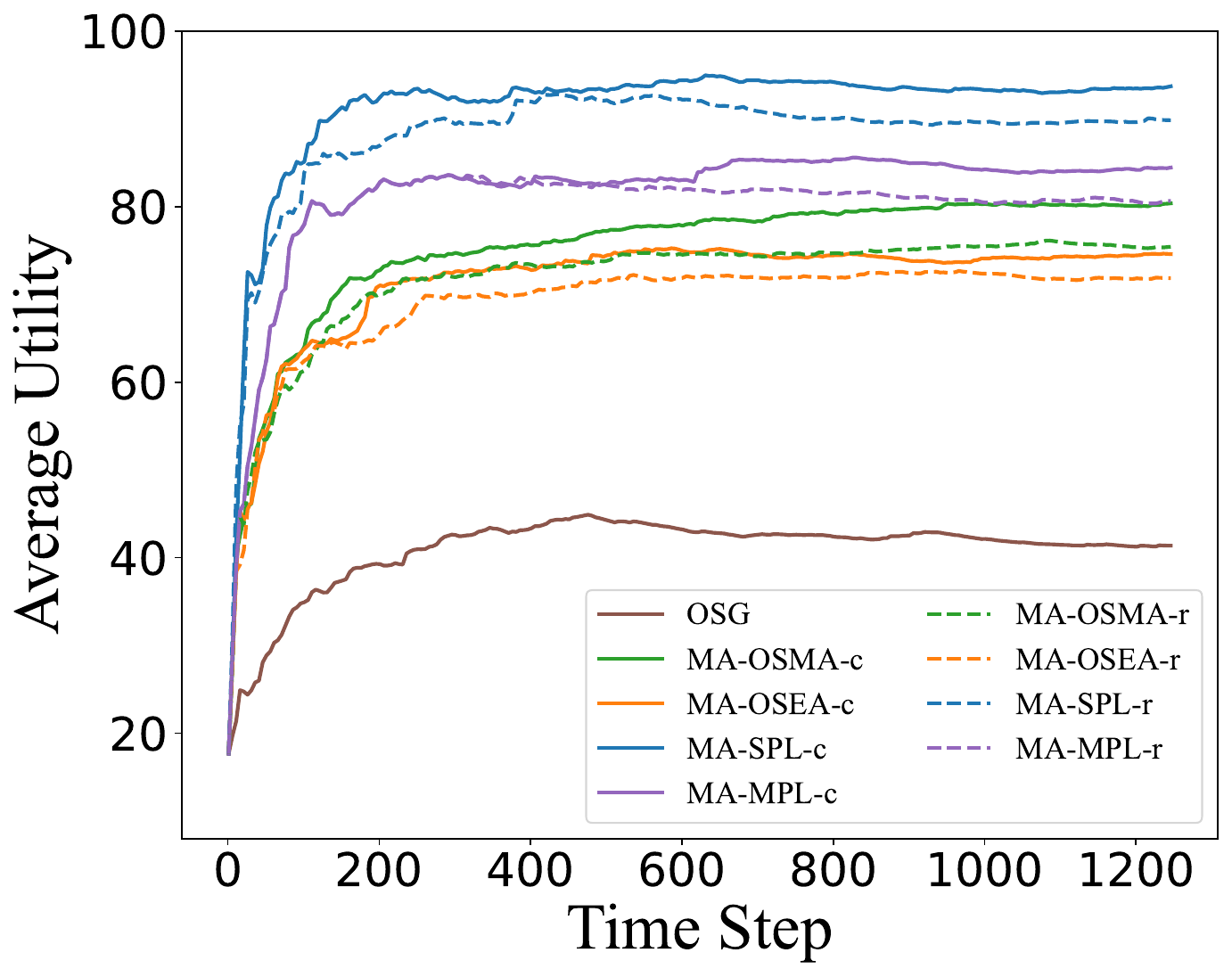}}
	\subfigure[\tiny `R':`A':`P'=4:5:1\label{graph_show_3}]{\includegraphics[scale=0.0925]{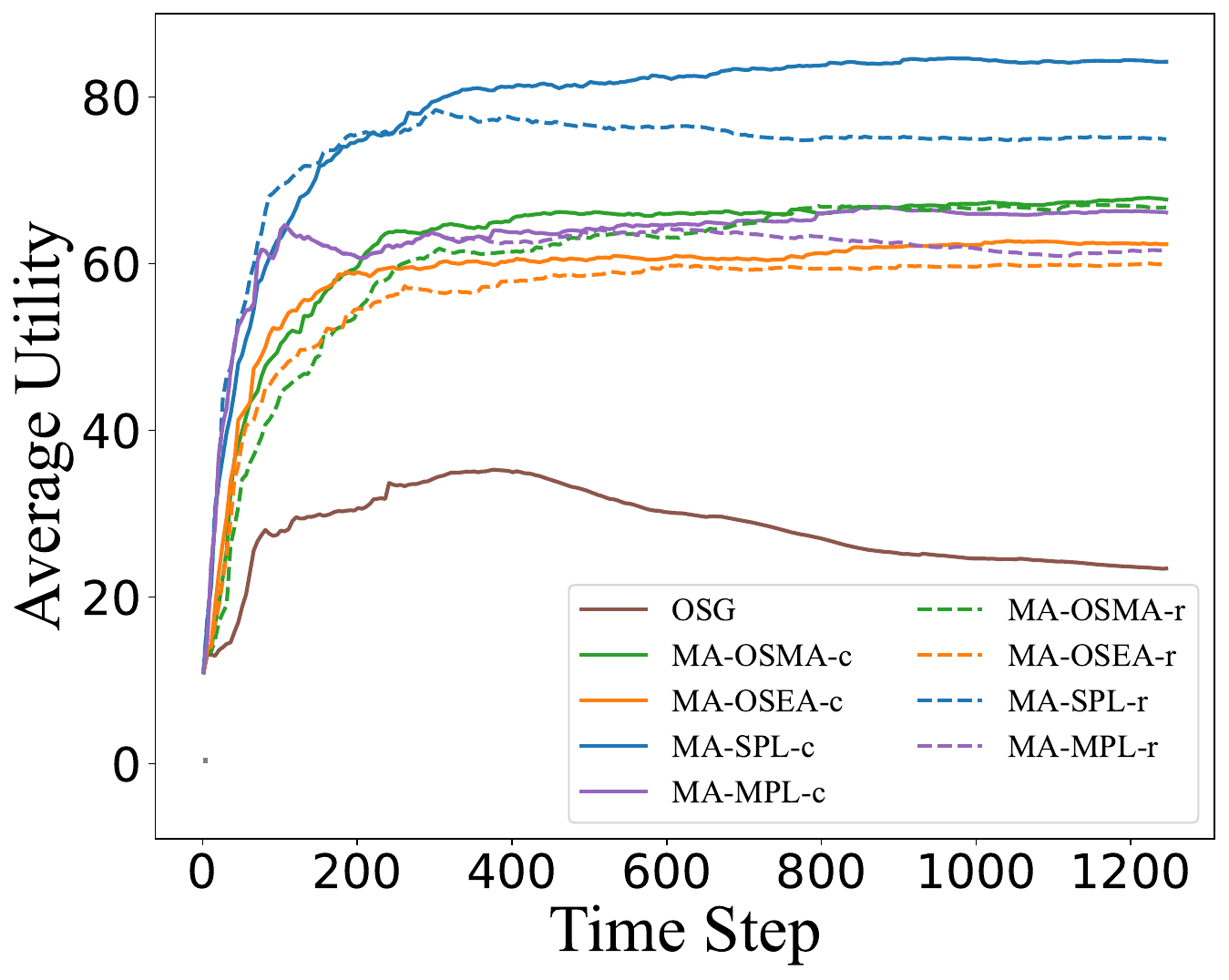}}
	\subfigure[\tiny A-Optimal Design\label{graph_show_4}]{\includegraphics[scale=0.0925]{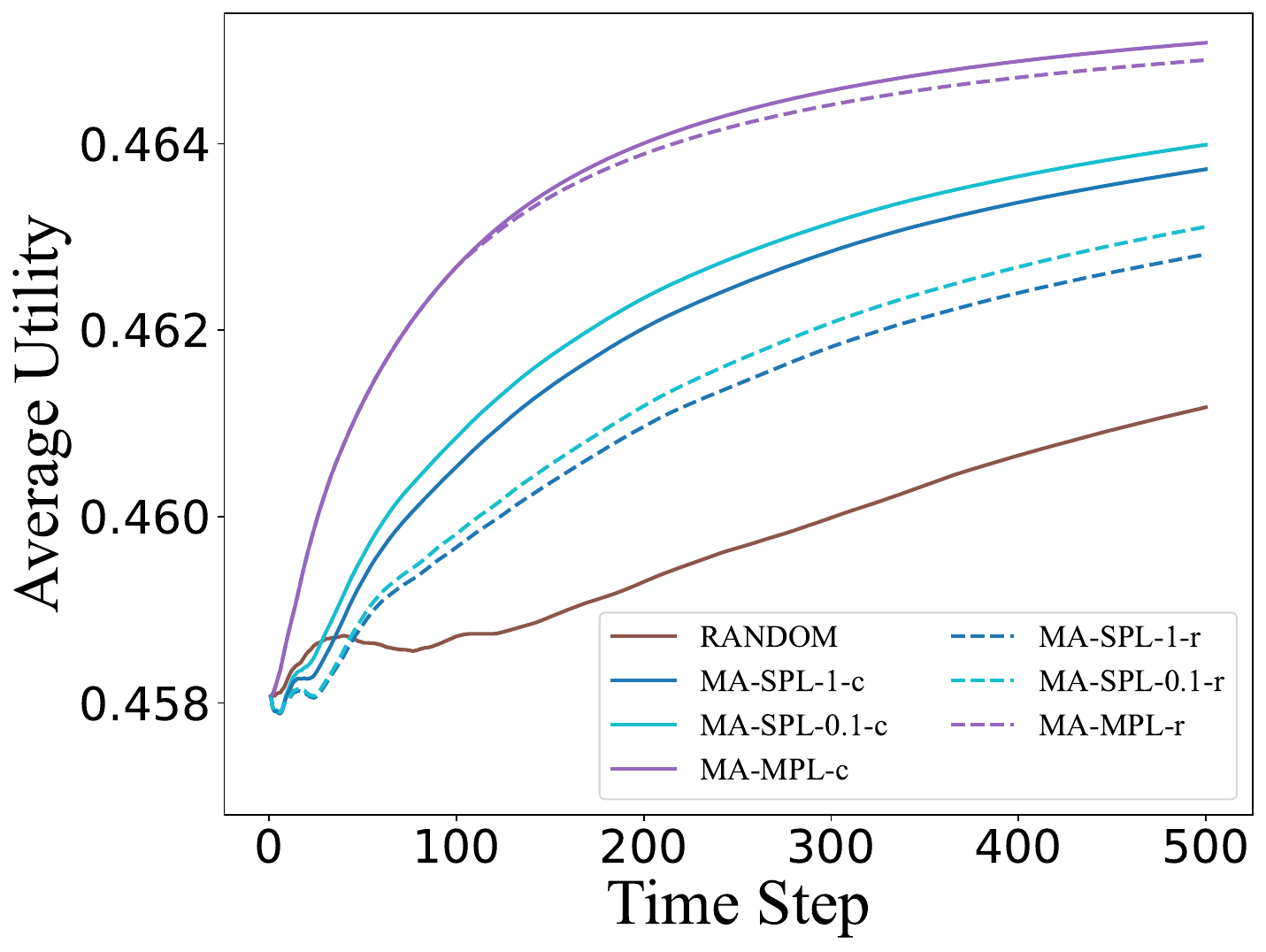}}
	\subfigure[\tiny Complete Graph\label{graph_show_5}]{\includegraphics[scale=0.0925]{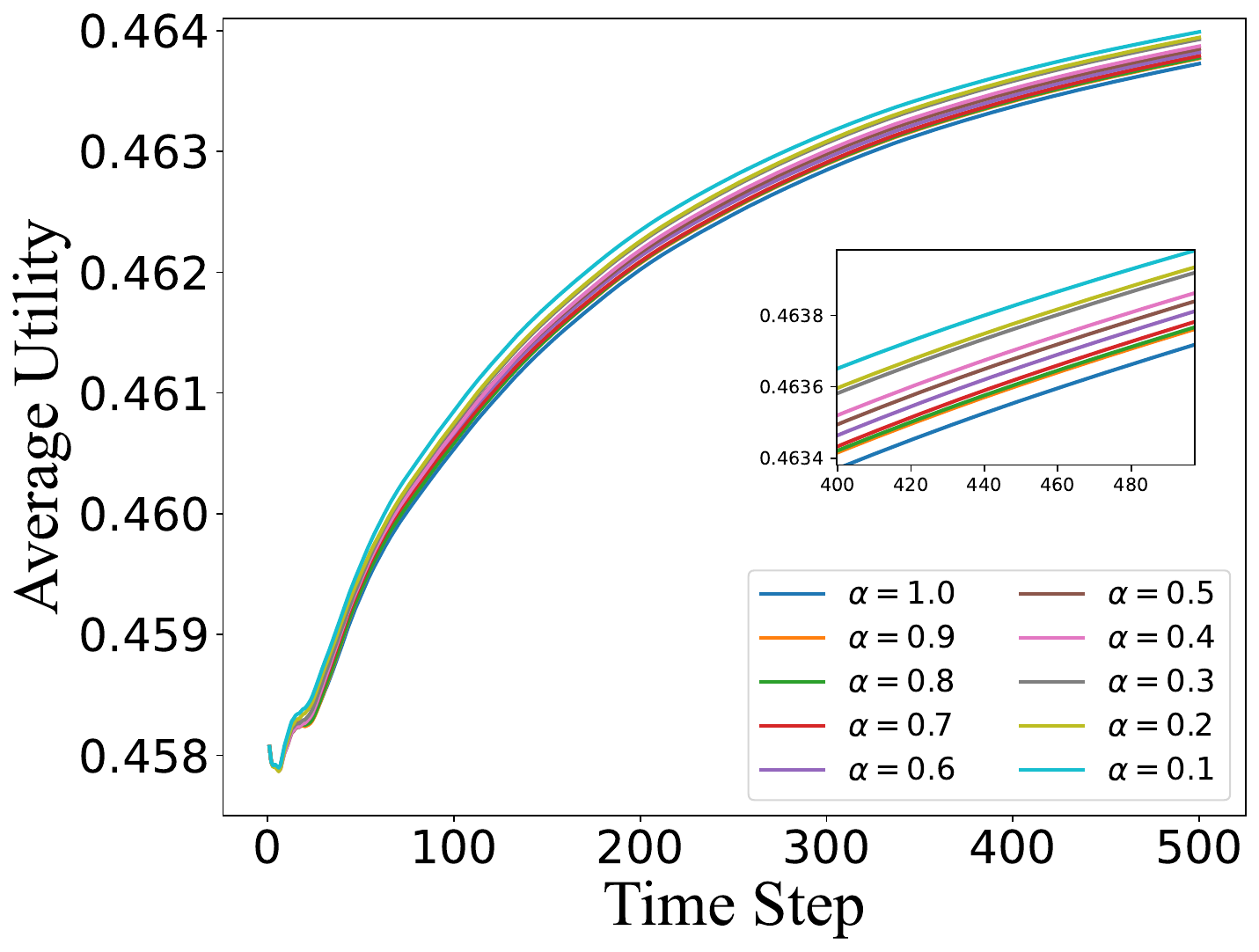}}
	\subfigure[\tiny Erdos-Renyi  Graph\label{graph_show_6}]{\includegraphics[scale=0.0925]{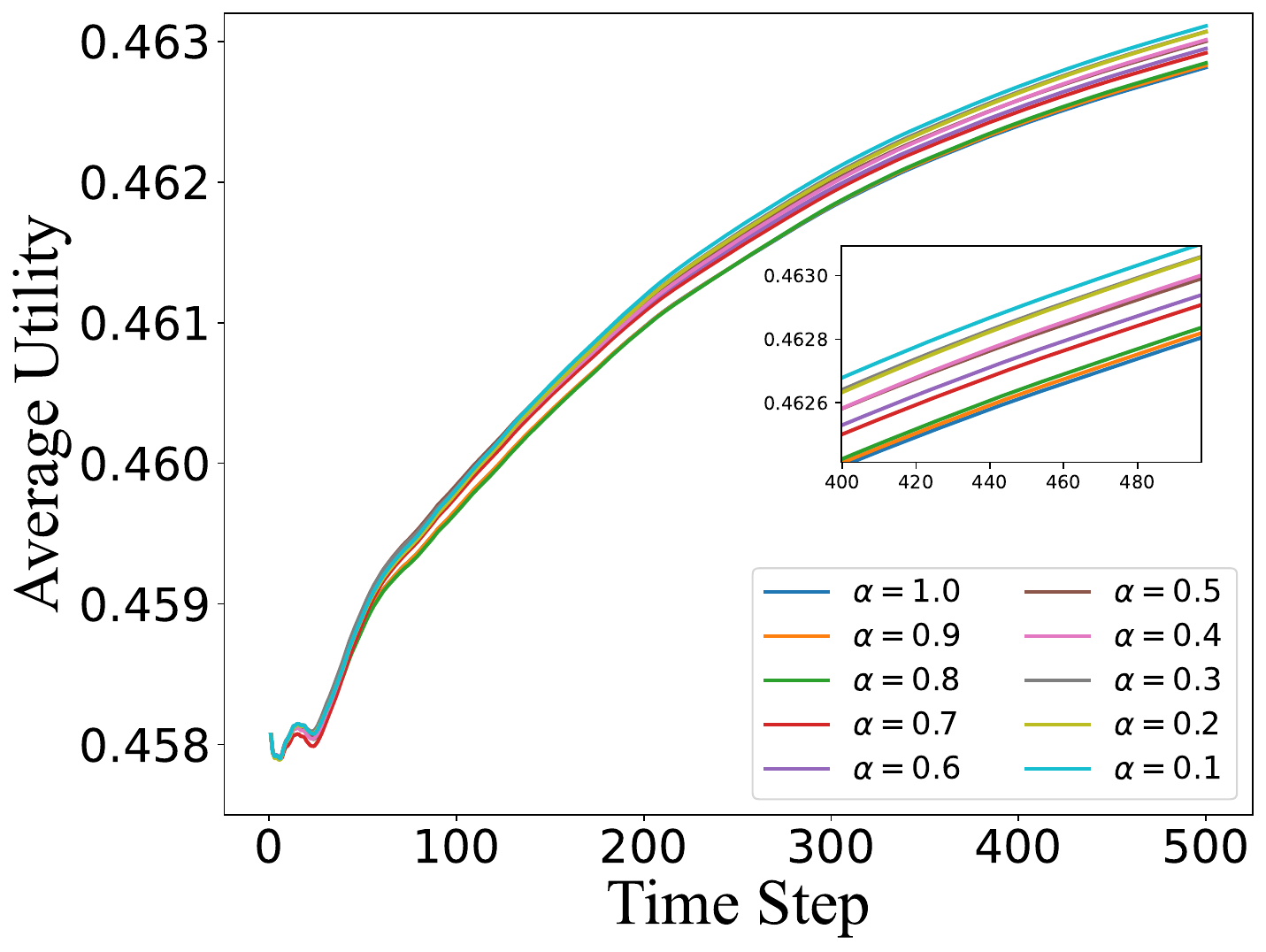}}
	\vspace{-0.5em}
	\caption{A comparison of the running average utility across two distinct target tracking simulations.}\label{graph_show}
\end{figure*}

\section{Conclusion}
In this paper, we primarily introduce an effective online policy learning algorithm named \texttt{MA-SPL} for the concerned MA-OC problem. Compared with the state-of-art MA-OSMA and MA-OSEA algorithms~\citep{zhang2025nearoptimal} , our proposed \texttt{MA-SPL} not only can guarantee a tight $(1-\frac{c}{e})$-approximation for MA-OC problem with submodular objectives but also can address the unexplored  $(\gamma,\beta)$-weakly submodular and $\alpha$-weakly DR-submodular scenarios. Subsequently, to eliminate the dependence on the unknown DR ratio and submodularity ratio in our \texttt{MA-SPL} algorithm, we further present a \emph{parameter-free} online algorithm named \texttt{MA-MPL} for the MA-OC problem. The key cornerstone of our \texttt{MA-SPL} and \texttt{MA-MPL} algorithms is a novel continuous relaxation termed as \emph{policy-based continuous extension}. In sharp contrast with the well-studied \emph{multi-linear extension}~\citep{calinescu2011maximizing}, a notable advantage of this new \emph{policy-based continuous extension} is its ability to provide a lossless rounding scheme for any set function, thereby enabling us to tackle  the challenging  weakly submodular objective functions. 
\section*{Acknowledgment}
This project is supported by the National Research Foundation, Singapore, under its NRF Professorship Award No. NRF-P2024-001.

 \bibliographystyle{plain}
\bibliography{neurips_2025}

\begin{thebibliography}{100}

\bibitem{ageev2004pipage}
Alexander~A Ageev and Maxim~I Sviridenko.
\newblock Pipage rounding: A new method of constructing algorithms with proven performance guarantee.
\newblock {\em Journal of Combinatorial Optimization}, 8:307--328, 2004.

\bibitem{albrecht2024multi}
Stefano~V Albrecht, Filippos Christianos, and Lukas Sch{\"a}fer.
\newblock {\em Multi-Agent Reinforcement Learning: Foundations and Modern Approaches}.
\newblock MIT Press, 2024.

\bibitem{allen2018make}
Zeyuan Allen-Zhu.
\newblock How to make the gradients small stochastically: Even faster convex and nonconvex sgd.
\newblock {\em Advances in Neural Information Processing Systems}, 31, 2018.

\bibitem{arslan2007autonomous}
G{\"u}rdal Arslan, Jason~R Marden, and Jeff~S Shamma.
\newblock Autonomous vehicle-target assignment: A game-theoretical formulation.
\newblock {\em Journal of Dynamic Systems, Measurement and Control, Transactions of the ASME}, 129(5):584--596, 2007.

\bibitem{atanasov2015decentralized}
Nikolay Atanasov, Jerome Le~Ny, Kostas Daniilidis, and George~J Pappas.
\newblock Decentralized active information acquisition: Theory and application to multi-robot slam.
\newblock In {\em 2015 IEEE International Conference on Robotics and Automation (ICRA)}, pages 4775--4782. IEEE, 2015.

\bibitem{bach2013learning}
Francis Bach et~al.
\newblock Learning with submodular functions: A convex optimization perspective.
\newblock {\em Foundations and Trends{\textregistered} in machine learning}, 6(2-3):145--373, 2013.

\bibitem{bertsekas2015convex}
Dimitri Bertsekas.
\newblock {\em Convex optimization algorithms}.
\newblock Athena Scientific, 2015.

\bibitem{bogunovic2018robust}
Ilija Bogunovic, Junyao Zhao, and Volkan Cevher.
\newblock Robust maximization of non-submodular objectives.
\newblock In {\em International Conference on Artificial Intelligence and Statistics}, pages 890--899. PMLR, 2018.

\bibitem{braun2022conditional}
G{\'a}bor Braun, Alejandro Carderera, Cyrille~W Combettes, Hamed Hassani, Amin Karbasi, Aryan Mokhtari, and Sebastian Pokutta.
\newblock Conditional gradient methods.
\newblock {\em arXiv preprint arXiv:2211.14103}, 2022.

\bibitem{buchbinder2019constrained}
Niv Buchbinder and Moran Feldman.
\newblock Constrained submodular maximization via a nonsymmetric technique.
\newblock {\em Mathematics of Operations Research}, 44(3):988--1005, 2019.

\bibitem{buchbinder2024constrained}
Niv Buchbinder and Moran Feldman.
\newblock Constrained submodular maximization via new bounds for dr-submodular functions.
\newblock In {\em Proceedings of the 56th Annual ACM Symposium on Theory of Computing}, pages 1820--1831, 2024.

\bibitem{buchbinder2014submodular}
Niv Buchbinder, Moran Feldman, Joseph Naor, and Roy Schwartz.
\newblock Submodular maximization with cardinality constraints.
\newblock In {\em Proceedings of the twenty-fifth annual ACM-SIAM symposium on Discrete algorithms}, pages 1433--1452. SIAM, 2014.

\bibitem{calinescu2011maximizing}
Gruia Calinescu, Chandra Chekuri, Martin Pal, and Jan Vondr{\'a}k.
\newblock Maximizing a monotone submodular function subject to a matroid constraint.
\newblock {\em SIAM Journal on Computing}, 40(6):1740--1766, 2011.

\bibitem{chaloner1995bayesian}
Kathryn Chaloner and Isabella Verdinelli.
\newblock Bayesian experimental design: A review.
\newblock {\em Statistical science}, pages 273--304, 1995.

\bibitem{chamon2017approximate}
Luiz Chamon and Alejandro Ribeiro.
\newblock Approximate supermodularity bounds for experimental design.
\newblock {\em Advances in Neural Information Processing Systems}, 30, 2017.

\bibitem{chekuri2010submodular}
Chandra Chekuri, Jan Vondrak, and Rico Zenklusen.
\newblock Dependent randomized rounding via exchange properties of combinatorial structures.
\newblock In {\em 2010 IEEE 51st Annual Symposium on Foundations of Computer Science}, pages 575--584, 2010.

\bibitem{chekuri2014submodular}
Chandra Chekuri, Jan Vondr{\'a}k, and Rico Zenklusen.
\newblock Submodular function maximization via the multilinear relaxation and contention resolution schemes.
\newblock {\em SIAM Journal on Computing}, 43(6):1831--1879, 2014.

\bibitem{chen2018weakly}
Lin Chen, Moran Feldman, and Amin Karbasi.
\newblock Weakly submodular maximization beyond cardinality constraints: Does randomization help greedy?
\newblock In {\em International Conference on Machine Learning}, pages 804--813. PMLR, 2018.

\bibitem{chen2018online}
Lin Chen, Hamed Hassani, and Amin Karbasi.
\newblock Online continuous submodular maximization.
\newblock In {\em International Conference on Artificial Intelligence and Statistics}, pages 1896--1905. PMLR, 2018.

\bibitem{chen2024continuousnonmonotonedrsubmodularmaximization}
Shengminjie Chen, Donglei Du, Wenguo Yang, Dachuan Xu, and Suixiang Gao.
\newblock Continuous non-monotone dr-submodular maximization with down-closed convex constraint, 2024.

\bibitem{chenyixin2024}
Yixin Chen, Ankur Nath, Chunli Peng, and Alan Kuhnle.
\newblock Discretely beyond 1/e: Guided combinatorial algortihms for submodular maximization.
\newblock In A.~Globerson, L.~Mackey, D.~Belgrave, A.~Fan, U.~Paquet, J.~Tomczak, and C.~Zhang, editors, {\em Advances in Neural Information Processing Systems}, volume~37, pages 108929--108973. Curran Associates, Inc., 2024.

\bibitem{conforti1984submodular}
Michele Conforti and G{\'e}rard Cornu{\'e}jols.
\newblock Submodular set functions, matroids and the greedy algorithm: tight worst-case bounds and some generalizations of the rado-edmonds theorem.
\newblock {\em Discrete applied mathematics}, 7(3):251--274, 1984.

\bibitem{corah2021scalable}
Micah Corah and Nathan Michael.
\newblock Scalable distributed planning for multi-robot, multi-target tracking.
\newblock In {\em 2021 IEEE/RSJ International Conference on Intelligent Robots and Systems (IROS)}, pages 437--444. IEEE, 2021.

\bibitem{das2011submodular}
Abhimanyu Das and David Kempe.
\newblock Submodular meets spectral: greedy algorithms for subset selection, sparse approximation and dictionary selection.
\newblock In {\em International Conference on Machine Learning}, pages 1057--1064, 2011.

\bibitem{das2018approximate}
Abhimanyu Das and David Kempe.
\newblock Approximate submodularity and its applications: Subset selection, sparse approximation and dictionary selection.
\newblock {\em Journal of Machine Learning Research}, 19(3):1--34, 2018.

\bibitem{drori2020complexity}
Yoel Drori and Ohad Shamir.
\newblock The complexity of finding stationary points with stochastic gradient descent.
\newblock In {\em International Conference on Machine Learning}, pages 2658--2667. PMLR, 2020.

\bibitem{du2022jacobi}
Bin Du, Kun Qian, Christian Claudel, and Dengfeng Sun.
\newblock Jacobi-style iteration for distributed submodular maximization.
\newblock {\em IEEE transactions on automatic control}, 67(9):4687--4702, 2022.

\bibitem{el2022data}
Marwa El~Halabi, Suraj Srinivas, and Simon Lacoste-Julien.
\newblock Data-efficient structured pruning via submodular optimization.
\newblock {\em Advances in Neural Information Processing Systems}, 35:36613--36626, 2022.

\bibitem{elenberg2018restricted}
Ethan~R Elenberg, Rajiv Khanna, Alexandros~G Dimakis, and Sahand Negahban.
\newblock Restricted strong convexity implies weak submodularity.
\newblock {\em The Annals of Statistics}, 46(6B):3539--3568, 2018.

\bibitem{feige1998threshold}
Uriel Feige.
\newblock A threshold of ln n for approximating set cover.
\newblock {\em Journal of the ACM}, 45(4):634--652, 1998.

\bibitem{feldman2021guess}
Moran Feldman.
\newblock Guess free maximization of submodular and linear sums.
\newblock {\em Algorithmica}, 83(3):853--878, 2021.

\bibitem{feldman2011unified}
Moran Feldman, Joseph Naor, and Roy Schwartz.
\newblock A unified continuous greedy algorithm for submodular maximization.
\newblock In {\em 2011 IEEE 52nd Annual Symposium on Foundations of Computer Science}, pages 570--579. IEEE, 2011.

\bibitem{filmus2012power}
Yuval Filmus and Justin Ward.
\newblock The power of local search: Maximum coverage over a matroid.
\newblock In {\em 29th Symposium on Theoretical Aspects of Computer Science}, volume~14, pages 601--612. LIPIcs, 2012.

\bibitem{filmus2012tight}
Yuval Filmus and Justin Ward.
\newblock A tight combinatorial algorithm for submodular maximization subject to a matroid constraint.
\newblock In {\em 2012 IEEE 53rd Annual Symposium on Foundations of Computer Science}, pages 659--668. IEEE, 2012.

\bibitem{filmus2014monotone}
Yuval Filmus and Justin Ward.
\newblock Monotone submodular maximization over a matroid via non-oblivious local search.
\newblock {\em SIAM Journal on Computing}, 43(2):514--542, 2014.

\bibitem{fisher1978analysis}
Marshall~L Fisher, George~L Nemhauser, and Laurence~A Wolsey.
\newblock An analysis of approximations for maximizing submodular set functions—ii.
\newblock In {\em Polyhedral Combinatorics}, pages 73--87. Springer, 1978.

\bibitem{fujishige2005submodular}
Satoru Fujishige.
\newblock {\em Submodular functions and optimization}.
\newblock Elsevier, 2005.

\bibitem{gatmiry2018non}
Khashayar Gatmiry and Manuel Gomez-Rodriguez.
\newblock Non-submodular function maximization subject to a matroid constraint, with applications.
\newblock {\em arXiv preprint arXiv:1811.07863}, 2018.

\bibitem{ghadimi2013stochastic}
Saeed Ghadimi and Guanghui Lan.
\newblock Stochastic first-and zeroth-order methods for nonconvex stochastic programming.
\newblock {\em SIAM journal on optimization}, 23(4):2341--2368, 2013.

\bibitem{gharesifard2017distributed}
Bahman Gharesifard and Stephen~L Smith.
\newblock Distributed submodular maximization with limited information.
\newblock {\em IEEE transactions on control of network systems}, 5(4):1635--1645, 2017.

\bibitem{gong2014diverse}
Boqing Gong, Wei-Lun Chao, Kristen Grauman, and Fei Sha.
\newblock Diverse sequential subset selection for supervised video summarization.
\newblock {\em Advances in neural information processing systems}, 27, 2014.

\bibitem{JOGO-continuous-loss}
Suning Gong, Qingqin Nong, Wenjing Liu, and Qizhi Fang.
\newblock Parametric monotone function maximization with matroid constraints.
\newblock {\em Journal of Global Optimization}, 75(3):833–849, November 2019.

\bibitem{GONG202116}
Suning Gong, Qingqin Nong, Tao Sun, Qizhi Fang, Dingzhu Du, and Xiaoyu Shao.
\newblock Maximize a monotone function with a generic submodularity ratio.
\newblock {\em Theoretical Computer Science}, 853:16--24, 2021.

\bibitem{grant2014cvx}
Michael Grant and Stephen Boyd.
\newblock Cvx: Matlab software for disciplined convex programming, version 2.1, 2014.

\bibitem{grimsman2018impact}
David Grimsman, Mohd~Shabbir Ali, Joao~P Hespanha, and Jason~R Marden.
\newblock The impact of information in distributed submodular maximization.
\newblock {\em IEEE Transactions on Control of Network Systems}, 6(4):1334--1343, 2018.

\bibitem{han2020deterministic}
Kai Han, Shuang Cui, Benwei Wu, et~al.
\newblock Deterministic approximation for submodular maximization over a matroid in nearly linear time.
\newblock {\em Advances in Neural Information Processing Systems}, 33:430--441, 2020.

\bibitem{harshaw2019submodular}
Chris Harshaw, Moran Feldman, Justin Ward, and Amin Karbasi.
\newblock Submodular maximization beyond non-negativity: Guarantees, fast algorithms, and applications.
\newblock In {\em International Conference on Machine Learning}, pages 2634--2643. PMLR, 2019.

\bibitem{harvey2020improved}
Nicholas Harvey, Christopher Liaw, and Tasuku Soma.
\newblock Improved algorithms for online submodular maximization via first-order regret bounds.
\newblock {\em Advances in Neural Information Processing Systems}, 33:123--133, 2020.

\bibitem{hashemi2019submodular}
Abolfazl Hashemi, Mahsa Ghasemi, Haris Vikalo, and Ufuk Topcu.
\newblock Submodular observation selection and information gathering for quadratic models.
\newblock In {\em International Conference on Machine Learning}, pages 2653--2662. PMLR, 2019.

\bibitem{hashemi2020randomized}
Abolfazl Hashemi, Mahsa Ghasemi, Haris Vikalo, and Ufuk Topcu.
\newblock Randomized greedy sensor selection: Leveraging weak submodularity.
\newblock {\em IEEE Transactions on Automatic Control}, 66(1):199--212, 2020.

\bibitem{hassani2020stochastic}
Hamed Hassani, Amin Karbasi, Aryan Mokhtari, and Zebang Shen.
\newblock Stochastic conditional gradient++:(non) convex minimization and continuous submodular maximization.
\newblock {\em SIAM Journal on Optimization}, 30(4):3315--3344, 2020.

\bibitem{hassani2017gradient}
Hamed Hassani, Mahdi Soltanolkotabi, and Amin Karbasi.
\newblock Gradient methods for submodular maximization.
\newblock In {\em Advances in Neural Information Processing Systems}, pages 5841--5851, 2017.

\bibitem{jawaid2015submodularity}
Syed~Talha Jawaid and Stephen~L Smith.
\newblock Submodularity and greedy algorithms in sensor scheduling for linear dynamical systems.
\newblock {\em Automatica}, 61:282--288, 2015.

\bibitem{kakade2007playing}
Sham~M Kakade, Adam~Tauman Kalai, and Katrina Ligett.
\newblock Playing games with approximation algorithms.
\newblock In {\em Proceedings of the thirty-ninth annual ACM symposium on Theory of computing}, pages 546--555, 2007.

\bibitem{kaya2025randomized}
Ege~Can Kaya, Michael Hibbard, Takashi Tanaka, Ufuk Topcu, and Abolfazl Hashemi.
\newblock Randomized greedy methods for weak submodular sensor selection with robustness considerations.
\newblock {\em Automatica}, 171:111984, 2025.

\bibitem{gatmiry2019nonsubmodular}
G.~Khashayar and G.~R. Manuel.
\newblock Non-submodular function maximization subject to a matroid constraint, with applications, 2019.

\bibitem{kia2025submodular}
Solmaz~S Kia.
\newblock Submodular maximization subject to uniform and partition matroids: From theory to practical applications and distributed solutions.
\newblock {\em arXiv preprint arXiv:2501.01071}, 2025.

\bibitem{krause2008efficient}
Andreas Krause, Jure Leskovec, Carlos Guestrin, Jeanne VanBriesen, and Christos Faloutsos.
\newblock Efficient sensor placement optimization for securing large water distribution networks.
\newblock {\em Journal of Water Resources Planning and Management}, 134(6):516--526, 2008.

\bibitem{krause2008robust}
Andreas Krause, H~Brendan McMahan, Carlos Guestrin, and Anupam Gupta.
\newblock Robust submodular observation selection.
\newblock {\em Journal of Machine Learning Research}, 9(12), 2008.

\bibitem{krause2008near}
Andreas Krause, Ajit Singh, and Carlos Guestrin.
\newblock Near-optimal sensor placements in gaussian processes: Theory, efficient algorithms and empirical studies.
\newblock {\em Journal of Machine Learning Research}, 9(2), 2008.

\bibitem{kuhnle2018fast}
Alan Kuhnle, J~David Smith, Victoria Crawford, and My~Thai.
\newblock Fast maximization of non-submodular, monotonic functions on the integer lattice.
\newblock In {\em International Conference on Machine Learning}, pages 2786--2795. PMLR, 2018.

\bibitem{lacoste2016convergence}
Simon Lacoste-Julien.
\newblock Convergence rate of frank-wolfe for non-convex objectives.
\newblock {\em arXiv preprint arXiv:1607.00345}, 2016.

\bibitem{lan2020first}
Guanghui Lan.
\newblock {\em First-order and stochastic optimization methods for machine learning}, volume~1.
\newblock Springer, 2020.

\bibitem{li2023submodularity}
Ruolin Li, Negar Mehr, and Roberto Horowitz.
\newblock Submodularity of optimal sensor placement for traffic networks.
\newblock {\em Transportation research part B: methodological}, 171:29--43, 2023.

\bibitem{liao2023improved}
Yucheng Liao, Yuanyu Wan, Chang Yao, and Mingli Song.
\newblock Improved projection-free online continuous submodular maximization.
\newblock {\em arXiv preprint arXiv:2305.18442}, 2023.

\bibitem{liu2021distributed}
Jun Liu, Lifeng Zhou, Pratap Tokekar, and Ryan~K Williams.
\newblock Distributed resilient submodular action selection in adversarial environments.
\newblock {\em IEEE Robotics and Automation Letters}, 6(3):5832--5839, 2021.

\bibitem{JOGO-Lu}
Cheng Lu, Wenguo Yang, Ruiqi Yang, and Suixiang Gao.
\newblock Maximizing a non-decreasing non-submodular function subject to various types of constraints.
\newblock {\em Journal of Global Optimization}, 83(4):727–751, August 2022.

\bibitem{manupriya2024submodular}
Piyushi Manupriya, Pratik Jawanpuria, Karthik~S. Gurumoorthy, SakethaNath Jagarlapudi, and Bamdev Mishra.
\newblock Submodular framework for structured-sparse optimal transport.
\newblock In {\em Forty-first International Conference on Machine Learning}, 2024.

\bibitem{marden2016role}
Jason~R Marden.
\newblock The role of information in distributed resource allocation.
\newblock {\em IEEE Transactions on Control of Network Systems}, 4(3):654--664, 2016.

\bibitem{mirzasoleiman2018streaming}
Baharan Mirzasoleiman, Stefanie Jegelka, and Andreas Krause.
\newblock Streaming non-monotone submodular maximization: Personalized video summarization on the fly.
\newblock In {\em Proceedings of the AAAI Conference on Artificial Intelligence}, volume~32, 2018.

\bibitem{mirzasoleiman2016distributed}
Baharan Mirzasoleiman, Amin Karbasi, Rik Sarkar, and Andreas Krause.
\newblock Distributed submodular maximization.
\newblock {\em The Journal of Machine Learning Research}, 17(1):8330--8373, 2016.

\bibitem{mokhtari2018decentralized}
Aryan Mokhtari, Hamed Hassani, and Amin Karbasi.
\newblock Decentralized submodular maximization: Bridging discrete and continuous settings.
\newblock In {\em International conference on machine learning}, pages 3616--3625. PMLR, 2018.

\bibitem{mualem2023resolving}
Loay Mualem and Moran Feldman.
\newblock Resolving the approximability of offline and online non-monotone dr-submodular maximization over general convex sets.
\newblock In {\em International Conference on Artificial Intelligence and Statistics}, pages 2542--2564. PMLR, 2023.

\bibitem{natarajan1995sparse}
Balas~Kausik Natarajan.
\newblock Sparse approximate solutions to linear systems.
\newblock {\em SIAM journal on computing}, 24(2):227--234, 1995.

\bibitem{nedic2009distributed}
Angelia Nedic and Asuman Ozdaglar.
\newblock Distributed subgradient methods for multi-agent optimization.
\newblock {\em IEEE Transactions on Automatic Control}, 54(1):48--61, 2009.

\bibitem{nemhauser1978analysis}
George~L Nemhauser, Laurence~A Wolsey, and Marshall~L Fisher.
\newblock An analysis of approximations for maximizing submodular set functions—i.
\newblock {\em Mathematical Programming}, 14(1):265--294, 1978.

\bibitem{nesterov2013introductory}
Y~Nesterov.
\newblock {\em Introductory Lectures on Convex Optimization: A Basic Course}, volume~87.
\newblock Springer Science \& Business Media, 2013.

\bibitem{niazadeh2020optimal}
Rad Niazadeh, Tim Roughgarden, and Joshua~R Wang.
\newblock Optimal algorithms for continuous non-monotone submodular and dr-submodular maximization.
\newblock {\em Journal of Machine Learning Research}, 21(125):1--31, 2020.

\bibitem{pedramfar2024from}
Mohammad Pedramfar and Vaneet Aggarwal.
\newblock From linear to linearizable optimization: A novel framework with applications to stationary and non-stationary {DR}-submodular optimization.
\newblock In {\em The Thirty-eighth Annual Conference on Neural Information Processing Systems}, 2024.

\bibitem{pedramfar2024unified}
Mohammad Pedramfar, Yididiya~Y. Nadew, Christopher~John Quinn, and Vaneet Aggarwal.
\newblock Unified projection-free algorithms for adversarial {DR}-submodular optimization.
\newblock In {\em The Twelfth International Conference on Learning Representations}, 2024.

\bibitem{pedramfar2023a}
Mohammad Pedramfar, Christopher~John Quinn, and Vaneet Aggarwal.
\newblock A unified approach for maximizing continuous {DR}-submodular functions.
\newblock In {\em Thirty-seventh Conference on Neural Information Processing Systems}, 2023.

\bibitem{prajapat2024submodular}
Manish Prajapat, Mojmir Mutny, Melanie Zeilinger, and Andreas Krause.
\newblock Submodular reinforcement learning.
\newblock In {\em The Twelfth International Conference on Learning Representations}, 2024.

\bibitem{qu2019distributed}
Guannan Qu, Dave Brown, and Na~Li.
\newblock Distributed greedy algorithm for multi-agent task assignment problem with submodular utility functions.
\newblock {\em Automatica}, 105:206--215, 2019.

\bibitem{FedSub}
Akbar Rafiey.
\newblock Decomposable submodular maximization in federated setting.
\newblock In {\em Forty-first International Conference on Machine Learning, {ICML} 2024, Vienna, Austria, July 21-27}, 2024.

\bibitem{rezazadeh2023distributed}
Navid Rezazadeh and Solmaz~S Kia.
\newblock Distributed strategy selection: A submodular set function maximization approach.
\newblock {\em Automatica}, 153:111000, 2023.

\bibitem{robey2021optimal}
Alexander Robey, Arman Adibi, Brent Schlotfeldt, Hamed Hassani, and George~J Pappas.
\newblock Optimal algorithms for submodular maximization with distributed constraints.
\newblock In {\em Learning for Dynamics and Control}, pages 150--162. PMLR, 2021.

\bibitem{sahin2020sets}
Aytunc Sahin, Yatao Bian, Joachim Buhmann, and Andreas Krause.
\newblock From sets to multisets: provable variational inference for probabilistic integer submodular models.
\newblock In {\em International Conference on Machine Learning}, pages 8388--8397. PMLR, 2020.

\bibitem{schlotfeldt2021resilient}
Brent Schlotfeldt, Vasileios Tzoumas, and George~J Pappas.
\newblock Resilient active information acquisition with teams of robots.
\newblock {\em IEEE Transactions on Robotics}, 38(1):244--261, 2021.

\bibitem{semsar2009multi}
Elham Semsar-Kazerooni and Khashayar Khorasani.
\newblock Multi-agent team cooperation: A game theory approach.
\newblock {\em Automatica}, 45(10):2205--2213, 2009.

\bibitem{shahrampour2017distributed}
Shahin Shahrampour and Ali Jadbabaie.
\newblock Distributed online optimization in dynamic environments using mirror descent.
\newblock {\em IEEE Transactions on Automatic Control}, 63(3):714--725, 2017.

\bibitem{shi2023robust}
Guangyao Shi, Lifeng Zhou, and Pratap Tokekar.
\newblock Robust multiple-path orienteering problem: Securing against adversarial attacks.
\newblock {\em IEEE Transactions on Robotics}, 39(3):2060--2077, 2023.

\bibitem{singh2007efficient}
Amarjeet Singh, Andreas Krause, Carlos Guestrin, William Kaiser, and Maxim Batalin.
\newblock Efficient planning of informative paths for multiple robots.
\newblock In {\em Proceedings of the 20th international joint conference on Artifical intelligence}, pages 2204--2211, 2007.

\bibitem{singh2009nonmyopic}
Amarjeet Singh, Andreas Krause, and William~J Kaiser.
\newblock Nonmyopic adaptive informative path planning for multiple robots.
\newblock In {\em Proceedings of the 21st International Joint Conference on Artificial Intelligence}, pages 1843--1850, 2009.

\bibitem{streeter2008online}
Matthew Streeter and Daniel Golovin.
\newblock An online algorithm for maximizing submodular functions.
\newblock In {\em Advances in Neural Information Processing Systems}, pages 1577--1584, 2008.

\bibitem{sun2025maskpro}
Yan Sun, Qixin Zhang, Zhiyuan Yu, Xikun Zhang, Li~Shen, and Dacheng Tao.
\newblock Maskpro: Linear-space probabilistic learning for strict (n: M)-sparsity on large language models.
\newblock {\em arXiv preprint arXiv:2506.12876}, 2025.

\bibitem{sviridenko2017optimal}
Maxim Sviridenko, Jan Vondr{\'a}k, and Justin Ward.
\newblock Optimal approximation for submodular and supermodular optimization with bounded curvature.
\newblock {\em Mathematics of Operations Research}, 42(4):1197--1218, 2017.

\bibitem{thiery2022two}
Theophile Thiery and Justin Ward.
\newblock Two-sided weak submodularity for matroid constrained optimization and regression.
\newblock In Po-Ling Loh and Maxim Raginsky, editors, {\em Proceedings of Thirty Fifth Conference on Learning Theory}, volume 178, pages 3605--3634. PMLR, 02--05 Jul 2022.

\bibitem{van2004detection}
Harry~L Van~Trees.
\newblock {\em Detection, estimation, and modulation theory, part I: detection, estimation, and linear modulation theory}.
\newblock John Wiley \& Sons, 2004.

\bibitem{vondrak2010submodularity}
Jan Vondr{\'a}k.
\newblock Submodularity and curvature: The optimal algorithm (combinatorial optimization and discrete algorithms).
\newblock {\em RIMS Kokyuroku Bessatsu B, 23:253–266,}, 23:253--266, 2010.

\bibitem{pmlr-v195-wan23a}
Yuanyu Wan, Lijun Zhang, and Mingli Song.
\newblock Improved dynamic regret for online frank-wolfe.
\newblock In Gergely Neu and Lorenzo Rosasco, editors, {\em Proceedings of Thirty Sixth Conference on Learning Theory}, volume 195 of {\em Proceedings of Machine Learning Research}, pages 3304--3327. PMLR, 12--15 Jul 2023.

\bibitem{wan2023bandit}
Zongqi Wan, Jialin Zhang, Wei Chen, Xiaoming Sun, and Zhijie Zhang.
\newblock Bandit multi-linear dr-submodular maximization and its applications on adversarial submodular bandits.
\newblock In {\em International Conference on Machine Learning}, pages 35491--35524. PMLR, 2023.

\bibitem{wang2022shaq}
Jianhong Wang, Yuan Zhang, Yunjie Gu, and Tae-Kyun Kim.
\newblock Shaq: Incorporating shapley value theory into multi-agent q-learning.
\newblock {\em Advances in Neural Information Processing Systems}, 35:5941--5954, 2022.

\bibitem{wang2025data}
Jun Wang, Zaifu Zhan, Qixin Zhang, Mingquan Lin, Meijia Song, and Rui Zhang.
\newblock Data-efficient biomedical in-context learning: A diversity-enhanced submodular perspective.
\newblock {\em arXiv preprint arXiv:2508.08140}, 2025.

\bibitem{xu2023bandit}
Zirui Xu, Xiaofeng Lin, and Vasileios Tzoumas.
\newblock Bandit submodular maximization for multi-robot coordination in unpredictable and partially observable environments.
\newblock In {\em Robotics: Science and Systems}, 2023.

\bibitem{xu2023online}
Zirui Xu, Hongyu Zhou, and Vasileios Tzoumas.
\newblock Online submodular coordination with bounded tracking regret: Theory, algorithm, and applications to multi-robot coordination.
\newblock {\em IEEE Robotics and Automation Letters}, 8(4):2261--2268, 2023.

\bibitem{yang2016tracking}
Tianbao Yang, Lijun Zhang, Rong Jin, and Jinfeng Yi.
\newblock Tracking slowly moving clairvoyant: Optimal dynamic regret of online learning with true and noisy gradient.
\newblock In {\em International Conference on Machine Learning}, pages 449--457. PMLR, 2016.

\bibitem{ye2023maximization}
Lintao Ye, Zhi-Wei Liu, Ming Chi, and Vijay Gupta.
\newblock Maximization of nonsubmodular functions under multiple constraints with applications.
\newblock {\em Automatica}, 155:111126, 2023.

\bibitem{yuan2024multi}
Deming Yuan, Alexandre Proutiere, Guodong Shi, et~al.
\newblock Multi-agent online optimization.
\newblock {\em Foundations and Trends{\textregistered} in Optimization}, 7(2-3):81--263, 2024.

\bibitem{yuan2016convergence}
Kun Yuan, Qing Ling, and Wotao Yin.
\newblock On the convergence of decentralized gradient descent.
\newblock {\em SIAM Journal on Optimization}, 26(3):1835--1854, 2016.

\bibitem{zhang2021multi}
Kaiqing Zhang, Zhuoran Yang, and Tamer Ba{\c{s}}ar.
\newblock Multi-agent reinforcement learning: A selective overview of theories and algorithms.
\newblock {\em Handbook of reinforcement learning and control}, pages 321--384, 2021.

\bibitem{zhang2019online}
Mingrui Zhang, Lin Chen, Hamed Hassani, and Amin Karbasi.
\newblock Online continuous submodular maximization: From full-information to bandit feedback.
\newblock In {\em Advances in Neural Information Processing Systems}, pages 9206--9217, 2019.

\bibitem{zhang2022boosting}
Qixin Zhang, Zengde Deng, Zaiyi Chen, Haoyuan Hu, and Yu~Yang.
\newblock Stochastic continuous submodular maximization: Boosting via non-oblivious function.
\newblock In {\em International Conference on Machine Learning}, pages 26116--26134. PMLR, 2022.

\bibitem{zhang2023online}
Qixin Zhang, Zengde Deng, Zaiyi Chen, Kuangqi Zhou, Haoyuan Hu, and Yu~Yang.
\newblock Online learning for non-monotone dr-submodular maximization: From full information to bandit feedback.
\newblock In {\em International Conference on Artificial Intelligence and Statistics}, pages 3515--3537. PMLR, 2023.

\bibitem{zhang2023communication}
Qixin Zhang, Zengde Deng, Xiangru Jian, Zaiyi Chen, Haoyuan Hu, and Yu~Yang.
\newblock Communication-efficient decentralized online continuous dr-submodular maximization.
\newblock In {\em Proceedings of the 32nd ACM International Conference on Information and Knowledge Management}, pages 3330--3339, 2023.

\bibitem{zhang2024boosting}
Qixin Zhang, Zongqi Wan, Zengde Deng, Zaiyi Chen, Xiaoming Sun, Jialin Zhang, and Yu~Yang.
\newblock Boosting gradient ascent for continuous dr-submodular maximization.
\newblock {\em arXiv preprint arXiv:2401.08330}, 2024.

\bibitem{zhang2025nearoptimal}
Qixin Zhang, Zongqi Wan, Yu~Yang, Li~Shen, and Dacheng Tao.
\newblock Near-optimal online learning for multi-agent submodular coordination: Tight approximation and communication efficiency.
\newblock In {\em The Thirteenth International Conference on Learning Representations}, 2025.

\bibitem{zhao2021improved}
Peng Zhao and Lijun Zhang.
\newblock Improved analysis for dynamic regret of strongly convex and smooth functions.
\newblock In {\em Learning for Dynamics and Control}, pages 48--59. PMLR, 2021.

\bibitem{zhou2025improved}
Huanjian Zhou, Lingxiao Huang, and Baoxiang Wang.
\newblock Improved approximation algorithms for \$k\$-submodular maximization via multilinear extension.
\newblock In {\em The Thirteenth International Conference on Learning Representations}, 2025.

\bibitem{zhou2023robust}
Lifeng Zhou and Vijay Kumar.
\newblock Robust multi-robot active target tracking against sensing and communication attacks.
\newblock {\em IEEE Transactions on Robotics}, 39(3):1768--1780, 2023.

\bibitem{zhou2019sensor}
Lifeng Zhou and Pratap Tokekar.
\newblock Sensor assignment algorithms to improve observability while tracking targets.
\newblock {\em IEEE Transactions on Robotics}, 35(5):1206--1219, 2019.

\bibitem{zhou2022risk}
Lifeng Zhou and Pratap Tokekar.
\newblock Risk-aware submodular optimization for multirobot coordination.
\newblock {\em IEEE Transactions on Robotics}, 38(5):3064--3084, 2022.

\bibitem{zhou2018resilient}
Lifeng Zhou, Vasileios Tzoumas, George~J Pappas, and Pratap Tokekar.
\newblock Resilient active target tracking with multiple robots.
\newblock {\em IEEE Robotics and Automation Letters}, 4(1):129--136, 2018.

\bibitem{zhu2021projection}
Junlong Zhu, Qingtao Wu, Mingchuan Zhang, Ruijuan Zheng, and Keqin Li.
\newblock Projection-free decentralized online learning for submodular maximization over time-varying networks.
\newblock {\em Journal of Machine Learning Research}, 22(51):1--42, 2021.

\bibitem{zinkevich2003online}
Martin Zinkevich.
\newblock Online convex programming and generalized infinitesimal gradient ascent.
\newblock In {\em International Conference on Machine Learning}, pages 928--936, 2003.

\end{thebibliography}
\clearpage
\appendix
\onecolumn

\begin{center}
	\LARGE \textbf{Appendices for  ``Effective Multi-Agent Online Coordination Beyond Submodular Objectives''} \\
\end{center}
\vspace{-3.0em}
\addcontentsline{toc}{section}{} % Add the appendix text to the document TOC
\part{} % Start the appendix part
\parttoc % Insert the appendix TOC
\clearpage
\section{Literature Review}\label{appendix:related_work}
In this section, we aim to provide a comprehensive review of the related literature.
\subsection{Submodular Maximization}
 A set function $f:2^{\V}\rightarrow\R_{+}$ is called submodular if and only if it satisfies the diminishing-return property, namely, for any two subsets $A\subseteq B\subseteq\V$ and $v\in\V\setminus B$, $f(A\cup\{v\})-f(A)\ge f(B\cup\{v\})-f(B)$. It is widely recognized that the maximization of submodular functions is $\textbf{NP}$-hard, implying that no polynomial-time algorithms can solve it optimally. To overcome this challenge, \citep{nemhauser1978analysis} introduced a greedy algorithm for solving the monotone submodular maximization problem under a cardinality constraint and demonstrated that this greedy algorithm can achieve an approximation ratio of $(1-e^{-1})$. Subsequently, \citep{feige1998threshold} showed that this $(1-e^{-1})$-approximation guarantee is tight for monotone submodular maximization under reasonable complexity-theoretic assumptions. After that, \citep{fisher1978analysis} extended the greedy algorithm to the general matroid constraint.  However,  \citep{fisher1978analysis}  also pointed out that under such a matroid constraint, the approximation ratio achievable by the greedy algorithm diminishes from $(1-e^{-1})$ to $1/2$. To achieve the tight $(1-e^{-1})$-approximation under matroid constraint, then \citep{calinescu2011maximizing} proposed a continuous greedy algorithm for submodular functions. A key innovation of this continuous greedy algorithm is  a novel continuous-relaxation technique termed as the \emph{multi-linear extension}. Furthermore, there has been extensive research dedicated to the non-monotone submodular maximization~\citep{buchbinder2019constrained,buchbinder2024constrained,chekuri2014submodular,chen2024continuousnonmonotonedrsubmodularmaximization,chenyixin2024,feldman2011unified,han2020deterministic,mualem2023resolving,niazadeh2020optimal,zhang2023online}%, online submodular maximization~\citep{chen2018projection,chen2018online,harvey2020improved,streeter2008online,thang2021online,zhang2019online} and decentralized submodular maximization~\citep{gao2023convergence,liao2023improved,lu2025decentralized,mokhtari2018decentralized,xie2019decentralized,zhang2023communication,zhu2021projection}.
\subsection{Non-Submodular Maximization}
 Recently, numerous studies have found that there exist various real-world applications inducing utility functions that are close to submodular, but not strictly submodular. Examples include variable selection~\citep{das2018approximate, elenberg2018restricted}, data summarization~\citep{gong2014diverse,mirzasoleiman2018streaming,wang2025data}, neural network pruning~\citep{el2022data,sun2025maskpro}, target tracking~\citep{hashemi2019submodular,hashemi2020randomized} and sparse optimal transport~\citep{manupriya2024submodular}.
 
 \textbf{Weakly Submodular Maximization.} A important class of close-to-submodular functions is known as $\gamma$-weakly submodular functions. Specifically, for a set function $f:2^{\V}\rightarrow\R_{+}$, it is called $\gamma$-weakly submodular if and only if, for any two subsets $A\subseteq B\subseteq\V$, the following inequality holds: $\sum_{v\in B\setminus A}f(v|A)\ge \gamma\big(f(B)-f(A)\big)$. The $\gamma$-weakly submodular functions were originally introduced by the work~\citep{das2011submodular}. Furthermore, when considering the simple cardinality constraint, \citep{das2011submodular} also show that the standard greedy algorithm can achieve an approximation ratio of $(1-e^{-\gamma})$ for the maximization of a $\gamma$-weakly submodular function. Subsequently, \citep{harshaw2019submodular} proved that, for any $\epsilon>0$, no polynomial-time algorithm can achieve a $(1-e^{-\gamma}+\epsilon)$-approximation for the problem of maximizing a $\gamma$-weakly submodular function subject to a cardinality constraint. As for more complicated matroid constraints,  \citep{chen2018weakly} showed that the residual random greedy method of \citep{buchbinder2014submodular} can achieve an approximation ratio of  $\frac{\gamma^{2}}{(1+\gamma)^{2}}$ for the problem of maximizing a monotone $\gamma$-weakly submodular functions. After that, \citep{gatmiry2019nonsubmodular} examined the approximation performance of the standard greedy algorithm on the $\gamma$-weakly submodular maximization problem over a matroid constraint, which indicated that the standard greedy algorithm only can offer an approximation ratio of $\frac{0.4\gamma^{2}}{\sqrt{r\gamma}+1}$ where $r$ is the rank of the matroid. It is important to note that this approximation ratio $\frac{0.4\gamma^{2}}{\sqrt{r\gamma}+1}$ is  not a constant guarantee and highly depends on the matroid rank $r$. In order to improve the approximation performance of these greedy-based algorithms over matroid constraints, \citep{thiery2022two} introduced the notion of upper submodularity ratio $\beta$ and developed a more powerful distorted local-search algorithm for $(\gamma,\beta)$-weakly submodular maximization problem. More importantly, this distorted local search can guarantee a $\frac{\gamma^{2}(1-e^{-(\beta(1-\gamma)+\gamma^2)})}{\beta(1-\gamma)+\gamma^2}$-approximation for the problem of maximizing a monotone $(\gamma,\beta)$-weakly submodular functions subject to a matroid constraint. Note that, when the $(\gamma,\beta)$-weakly submodular function is closer to being submodular, i.e., $\gamma,\beta\rightarrow 1$, the approximation ratio $\frac{\gamma^{2}(1-e^{-(\beta(1-\gamma)+\gamma^2)})}{\beta(1-\gamma)+\gamma^2}$ will approach the tight $(1-1/e)$. Conversely, when $\gamma,\beta\rightarrow 1$, the approximation ratio $\frac{\gamma^{2}}{(1+\gamma)^{2}}$ of the residual random greedy method~\citep{chen2018weakly}  will trend to the sub-optimal $1/4$. 

\textbf{Weakly DR-Submodular Maximization.} The other  important class of  non-submodular  functions is known as $\alpha$-weakly DR-submodular functions, where $\alpha$ is variously referred to as the diminishing-return(DR) ratio~\citep{kuhnle2018fast}, the generalized curvature~\citep{bogunovic2018robust} or the generic submodularity ratio~\citep{GONG202116}. The work of \citep{gatmiry2019nonsubmodular} is the first to explore the problem of maximizing a $\alpha$-weakly DR-submodular maximization subject to general matroid constraints and proved that the standard greedy algorithm achieves
approximation ratios of $\frac{\alpha}{1+\alpha}$ for the matroid-constrained $\alpha$-weakly DR-submodular maximization. Recently, \citep{JOGO-continuous-loss}  also showed that the continuous greedy combined with the contention resolution scheme~\citep{chekuri2014submodular} can obtain a sub-optimal approaximation ratio of $\big(\alpha(1-1/e)(1-e^{-\alpha})\big)$ for the problem of maximizing a monotone $\alpha$-weakly DR-submodular functions subject to a matroid constraint. To achieve the tight $(1-e^{-\alpha})$-approximation guarantee, \citep{JOGO-Lu} presented a novel distorted local-search method for the problem of maximizing a $\alpha$-weakly DR-submodular maximization subject to al matroid constraint, which is motivated via the non-oblivious search~\citep{filmus2014monotone,filmus2012tight,filmus2012power}.
\subsection{Multi-Agent Submodular Maximization}
\textbf{Multi-Agent Offline Submodular Maximization.} Coordinating multiple agents to collaboratively maximize a submodular function is a critical task with numerous applications in machine learning, robot planning and control. A common solution for multi-agent submodular maximization problem heavily depends on the distributed implementation of the classic sequential greedy method~\citep{fisher1978analysis}, which can ensure a $(\frac{1}{1+c})$-approximation~\citep{conforti1984submodular} when the submodular function possesses a curvature of $c\in[0,1]$. However, this distributed greedy algorithm requires each agent to have full access to the decisions of all previous agents, thereby forming a \emph{complete} directed communication graph. Subsequently, several studies~\citep{grimsman2018impact,gharesifard2017distributed,marden2016role} have investigated how the topology of the communication network affects the performance of the distributed greedy method. Particularly, \citep{grimsman2018impact} pointed out that the worst-case performance of the distributed greedy algorithm will deteriorate in proportion to the size of the largest independent group of agents in the communication graph. In order to overcome these challenges, various studies \citep{rezazadeh2023distributed,robey2021optimal,du2022jacobi} utilized the \emph{multi-linear extension}~\citep{calinescu2011maximizing} to design algorithms for solving mutli-agent submodular maximization problem. Specifically, \citep{du2022jacobi} proposed
a multi-agent variant of gradient ascent for  mutli-agent submodular maximization problem and showed that this multi-agent variant of gradient ascent can attain $\frac{1}{2}OPT-\epsilon$ over any connected communication graph where $OPT$ is the optimal value. After that, to achieve the tight  $(1-1/e)$-approximation, \citep{robey2021optimal} developed a multi-agent variant of continuous greedy method~\citep{calinescu2011maximizing,chekuri2014submodular}. However, this multi-agent continuous greedy~\citep{robey2021optimal} requires the exact knowledge of the \emph{multi-linear extension}, which will lead to the exponential number of function evaluations. To tackle this drawback, \citep{rezazadeh2023distributed} also proposed a stochastic variant of continuous greedy method~\citep{calinescu2011maximizing,chekuri2014submodular}, which can enjoy $(\frac{1-e^{-c}}{c})OPT-\epsilon$. Here, $c$ is the curvature of the investigated submodular objectives.

According to the hardness result in \emph{centralized} submodular maximization~\citep{sviridenko2017optimal}, the optimal achievable approximation guarantee for maximizing a submodular function with curvature $c$ is  $(1-\frac{c}{e})$. However, as previously discussed, the state-of-the-art approximation guarantee for multi-agent offline submodular maximization problems is $(\frac{1-e^{-c}}{c})$, which highly mismatches the best possible guarantee of $(1-\frac{c}{e})$ established in \citep{sviridenko2017optimal}. Thus, if in \cref{alg:SPL} we treat any incoming objective function $f_{t},\forall t\in[T]$ as a fixed submodular set objective $f$, our \cref{alg:SPL} can be naturally transformed into an approximation algorithm with the tight $(1-\frac{c}{e})$  guarantee for multi-agent offline submodular maximization problems. Similar to \citep{zhang2025nearoptimal}, we compare this offline version of our \cref{alg:SPL} with existing algorithms for multi-agent offline submodular maximization problems in \cref{tab:Comparison1}.
\begin{table}[t]
	\centering
	\caption{\small Comparison of the different algorithms for multi-agent offline submodular maximization problems. Note that 	`Approx.' denotes the obtained approximation result, `\#Com.' represents the number of communication, `\#Queries' denotes the number of queries to the set objective functions, `Proj-free' indicates whether the method does not require projection, `Para-free'  indicates whether the method does not require the knowledge of curvature $c$, $OPT$ denotes the optimal function value, $d(G)$ is the diameter of the graph $G$, $\tau$ is the spectral gap of the weight matrix, $\kappa\triangleq\sum_{i=1}^{n}\kappa_{i}$ and $\alpha(G)\ge1$ is the number of nodes in the largest independent set in graph $G$.}\label{tab:Comparison1}
	\vspace{0.5em}
	\resizebox{1.01\textwidth}{!}{
		\setlength{\tabcolsep}{2.5mm}{
			\begin{tabular}{ccccccccc}
				\toprule[1.3pt]
				Method&Type&Para-free&Proj-free&Approx.&Graph($G$)&\#Com.&\#Queries&Reference \\
				\hline
				Greedy Method &det.&\ding{52}&\ding{52}&$(\frac{1}{1+c})OPT$ &\textbf{complete}&$\mathcal{O}(n)$&$\mathcal{O}(\kappa n)$&\citep{conforti1984submodular}\\
				PGA&sto.&\ding{52}&\ding{56}&$\frac{1}{2}OPT-\epsilon$ &connected&$\mathcal{O}(\frac{1}{(1-\tau)\epsilon^{2}})$ &$\mathcal{O}\big(\frac{\kappa}{(1-\tau)\epsilon^{2}}\big)$&\citep{du2022jacobi}\\
				Greedy Method&det.&\ding{52}&\ding{52} &$(\frac{1}{1+\alpha(G)})OPT$ &connected&$\mathcal{O}\big(n\big)$ &$\mathcal{O}(\kappa n)$&\citep{grimsman2018impact,gharesifard2017distributed}\\
				CDCG&det.&\ding{52}&\ding{52}&$(1-\frac{1}{e})OPT-\epsilon$ &connected&$\mathcal{O}\big(\frac{1}{(1-\tau)\epsilon}\big)$ &$\mathcal{O}\big(\frac{\kappa2^{\kappa}}{(1-\tau)\epsilon}\big)$&\citep{robey2021optimal}\\
				Distributed-CG&sto.&\ding{52}&\ding{52}&$(\frac{1-e^{-c}}{c})OPT-\epsilon$ &connected&$\mathcal{O}\big(\frac{d(G)}{\epsilon}\big)$ &$\widetilde{\mathcal{O}}\Big(\frac{\kappa d^{3}(G)}{\epsilon^{3}}\Big)$&\citep{rezazadeh2023distributed}\\
				MA-OSMA&sto.&\ding{56}&\ding{56}&$(\frac{1-e^{-c}}{c})OPT-\epsilon$ &connected&$\mathcal{O}\big(\frac{1}{(1-\tau)\epsilon^{2}}\big)$ &$\mathcal{O}\big(\frac{\kappa}{(1-\tau)\epsilon^{2}}\big)$ &\citep{zhang2025nearoptimal}\\
				MA-OSEA&sto.&\ding{56}&\ding{52}&$(\frac{1-e^{-c}}{c})OPT-\epsilon$ &connected&$\mathcal{O}\big(\frac{1}{(1-\tau)\epsilon^{2}}\big)$ &$\mathcal{O}\Big(\frac{\kappa\log(\frac{1}{\epsilon})}{(1-\tau)\epsilon^{2}}\Big)$ &\citep{zhang2025nearoptimal}\\	\midrule[1.3pt]\rowcolor{gray!25}
		\cref{alg:SPL}&sto.&\ding{52}&\ding{56}&$(1-\frac{c}{e})OPT-\epsilon$ &connected&$\mathcal{O}\big(\frac{1}{(1-\tau)\epsilon^{2}}\big)$ &$\mathcal{O}\big(\frac{\kappa}{(1-\tau)\epsilon^{2}}\big)$&\cref{thm:result1}\\
				\midrule[1.3pt]
					\end{tabular}}}\vspace{-1.0em}
\end{table}
 
\textbf{Multi-Agent Online Submodular Maximization.}
 \citep{xu2023online} is the first one to explore the multi-agent submodular maximization problems in time-varying environments. Moreover,  \citep{xu2023online}  also proposed an \emph{online sequence greedy}(OSG) algorithm for multi-agent online submodular maximization problems and proved this OSG algorithm can achieve a sub-optimal $(\frac{1}{1+c})$-approximation over a \emph{complete} communication graph, where $c$ is the joint curvature of the investigated submodular objectives. Concurrently, \citep{xu2023bandit} extended this OSG algorithm into bandit settings. In order to improve this sub-optimal $(\frac{1}{1+c})$-approximation guarantee and reduce the rigid requirement of a fully connected
communication network of OSG algorithm, \citep{zhang2025nearoptimal} utilized the non-oblivious auxiliary functions presented in \citep{zhang2022boosting} to design two multi-agent variants of online gradient ascent algorithm, namely, MA-OSMA algorithm and MA-OSEA algorithm, for the multi-agent online submodular maximization  problem. Furthermore, \citep{zhang2025nearoptimal} also showed that these two MA-OSMA and MA-OSEA algorithms can attain a regret bound of $\widetilde{O}(\sqrt{\frac{\mathcal{P}_{T}T}{1-\tau}})$ against a  $(\frac{1-e^{-c}}{c})$-approximation to the best comparator in hindsight, where $\mathcal{P}_{T}$ is the deviation of maximizer sequence, $\tau$ is the spectral gap of the network and $c$ is the joint curvature of submodular objectives.
 \subsection{Multilinear Extension}
As almost all state-of-the-art algorithms for multi-agent submodular coordination~\citep{rezazadeh2023distributed,robey2021optimal,zhang2025nearoptimal} rely on the \emph{multi-linear extension} of \citep{calinescu2011maximizing}, thus, in this subsection, we review the concept of \emph{multi-linear extension} and compare it with our proposed policy-based continuous extension in \cref{sec:Continuous-Relaxation}.

At first, we define $\kappa\triangleq\sum_{i=1}^{n}\kappa_{i}$. Furthermore, from the definition of $\V$ in \cref{sec:problem}, we can know that $\kappa=|\V|$ such that we can re-define $\V\triangleq[\kappa]\triangleq\{1,\dots,\kappa\}$. Then, we can show that 
\begin{definition}\label{def1:multi-linear}
	For a set function $f:2^{\V}\rightarrow\R_{+}$, we define its multi-linear extension  as 
	\begin{equation}
		\label{equ:multi-linea}
		G(\x)=\sum_{\mathcal{A}\subseteq\V}\Big(f(\mathcal{A})\prod_{a\in\mathcal{A}}x_{a}\prod_{a\notin\mathcal{A}}(1-x_{a})\Big)=\E_{\mathcal{R}\sim\x}\Big(f(\mathcal{R})\Big),
	\end{equation} where $\x=(x_{1},\dots,x_{\kappa})\in [0,1]^{\kappa}$ and $\mathcal{R}\subseteq\V$ is a random set that contains each element $a\in\V$ independently with probability $x_{a}$ and excludes it with probability $1-x_{a}$. We write $\mathcal{R}\sim\x$ to denote that $\mathcal{R}\subseteq\V$ is a random set sampled according to $\x$. 
\end{definition}
With this \emph{multi-linear extension}, we then can transfer the previous discrete subset selection problem~\eqref{equ_problem} into a continuous maximization which aims at learning the optimal independent probability for each element $a\in\V$, that is,
\begin{equation}\label{equ:multilinear_continuous_max}
	\max_{\x} G(\x),\ \ \text{ s.t.}\ \x\in[0,1]^{\kappa}\ \text{and}\ \sum_{a\in\V_{i}}x_{a}\le 1,\forall i\in\N.
\end{equation}
It is important to note that, if we round any point $\x$ included into the constraint of problem~\eqref{equ:multilinear_continuous_max} by the definition of multi-linear extension, namely, Eq.\eqref{equ:multi-linea}, there is a certain probability that the resulting subset  will violate the partition constraint of the subset selection problem~\eqref{equ_problem}.  Therefore, for \emph{multi-linear extension}, we need to design a specific rounding methods based on the properties of the investigated set objective functions. However, current known lossless rounding schemes for multi-linear extension, such as pipage rounding~\citep{ageev2004pipage}, swap rounding~\citep{chekuri2010submodular} and contention resolution~\citep{chekuri2014submodular}, are heavily dependent on the \emph{submodular} assumption. Currently, how to losslessly round the multi-linear extension of \emph{non-submodular} set functions, e.g. $(\gamma,\beta)$-weakly submodular and $\alpha$-weakly DR-submodular functions, still remains an open question~\citep{thiery2022two}.  Conversely, our proposed policy-based continuous extension in \cref{sec:Continuous-Relaxation}  does not assign probabilities to any subsets that are out of the partition constraint of problem~\eqref{equ_problem}, which means that, for any set function $f:2^{\V}\rightarrow\R_{+}$ and any given policy vector $(\uppi_{1},\dots,\uppi_{n})\in\prod_{i=1}^{n}\Delta_{\kappa_{i}}$, we can, through the \cref{def_extension}, easily produce a subset that conforms to the constraint of problem~\eqref{equ_problem} without any loss in terms of the expected function value $F_{t}(\uppi_{1},\dots,\uppi_{n})$. 
\newpage
\section{Parameter-free Multi-Agent Policy Learning}\label{appendix:alg2}
It is important to note that, when considering the weakly submodular scenarios, the Line 9 of \cref{alg:SPL} requires the prior knowledge of the unknown submodularity ratio $(\gamma,\beta)$ or the diminishing-return(DR) ratio $\alpha$ to set the weight function $w(z)$. However, in general, accurately computing these parameters will incur exponential computations.  To overcome this drawback, in this section, we explore how to design a \emph{parameter-free} online algorithm for the MA-OC problem with $(\gamma,\beta)$-weakly submodular or $\alpha$-weakly DR-submodular objectives.

Note that, in part $\textbf{4)}$ of \cref{thm1}, we establish a novel relationship between our proposed policy-based continuous extension $F_{t}$ and its original set function $f_{t}$. More specifically, when the $f_{t}$ is monotone $\alpha$-weakly DR-submodular or $(\gamma,\beta)$-weakly submodular,  the weighted discrepancy  between any $f_{t}(S)$ and its policy-based continuous extension $F_{t}(\uppi_{1},\dots,\uppi_{n})$ can be bounded by some sum of first-order derivatives, i.e., $\sum_{(i,m): v_{i,m}\in S}\frac{\partial F_{t}}{\partial \pi_{i,m}}(\uppi_{1},\dots,\uppi_{n})$. It is worth noting that the computation of this sum $\sum_{(i,m): v_{i,m}\in S}\frac{\partial F_{t}}{\partial \pi_{i,m}}(\uppi_{1},\dots,\uppi_{n})$ does not rely on the knowledge of $\alpha,\gamma,\beta$. Therefore, we naturally consider whether it is possible to devise a \emph{parameter-free} online algorithm for the MA-OC problem by controlling this aforementioned sum $\sum_{(i,m): v_{i,m}\in S}\frac{\partial F_{t}}{\partial \pi_{i,m}}(\uppi_{1},\dots,\uppi_{n})$  to narrow the gap between $F_{t}(\uppi_{1},\dots,\uppi_{n})$ and $f_{t}(S)$.

 Before that, \citep{chen2018online,pedramfar2024from,zhang2019online,zhu2021projection} utilized the idea of meta-actions~\citep{streeter2008online} to devise a Meta-Frank-Wolfe algorithm for online submodular maximization problems. The core of this Meta-Frank-Wolfe algorithm lies in iteratively optimizing an upper bound of the gap of the corresponding \emph{multi-linear extension} $G_{t}$, i.e., $\langle\y, \nabla G_{t}(\x)\rangle\ge G_{t}(\y)-G_{t}(\x)$, which is very similar to our previously discussed idea. Motivated by this finding, we then leverage the idea of meta-actions~\citep{streeter2008online} and the insight from part $\textbf{4)}$ of \cref{thm1} to design a \emph{parameter-free} \texttt{MA-MPL} algorithm for the  MA-OC problem with $(\gamma,\beta)$-weakly submodular or $\alpha$-weakly DR-submodular objectives,
as shown in \cref{alg:MPL}.

\begin{algorithm}[t]
	\caption{Multi-Agent Meta-Policy Learning(\texttt{MA-MPL})}\label{alg:MPL}
	\KwIn{Time horizon $T$, action set $\V_{i}\triangleq\{v_{i,1},\dots,v_{i,\kappa_{i}}\}$, Communication Graph $G(\N,\mathcal{E})$, number of online linear maximization oracles, namely, $K$ and batch size $L$}
	\tcp{\textcolor{teal}{Policy Vector and Online Linear Oracles Initialization (Lines 1-2)}}
	Initialize a policy vector $(\uppi^{(0)}_{i,1}(t),\dots,\uppi^{(0)}_{i,n}(t))\triangleq\boldsymbol{0},\forall t\in[T]$ for each agent $i\in\N$\; 
	Initialize $K$ online linear oracles over $\Delta_{\kappa_{i}}$, i.e., $\{Q_{i}^{(1)},\dots,Q_{i}^{(K)}\}$, for each agent $i\in\N$\; 
	\For{$t=1,\dots,T$}{
		\For{each agent $i\in\N$}{
			\tcp{\textcolor{teal}{Policy Update and Information Exchange (Lines 5-10)}}
			\For{$k=1,\dots,K$}{
				Set the shared vector $\y^{(k)}_{i,m}(t)\triangleq\uppi^{(k-1)}_{i,m}(t)$ for any $m\neq i$ and $m\in\N$\;
				Obtain the update direction $\mathbf{v}_{i}^{(k)}(t)\in\Delta_{\kappa_{i}}$ from oracle $Q_{i}^{(k)}$\;
				Compute $\y^{(k)}_{i,i}(t)\triangleq\uppi^{(k-1)}_{i,i}(t)+\frac{1}{K}\mathbf{v}_{i}^{(k)}(t)$\;
				Exchange  the vector $\big(\y^{(k)}_{i,1}(t),\dots,\y^{(k)}_{i,n}(t)\big)$ with the neighboring agent $j\in\mathcal{N}_{i}$\;
				Set $\uppi^{(k)}_{i,m}\triangleq\max_{j\in\N_{i}\cup\{i\}}\big(\y_{j,m}^{(k)}(t)\big),\forall m\in\N$\;
			}
			\tcp{\textcolor{teal}{Actions Sampling (Lines 11-13)}}
			Compute the normalized policy $\p_{i}(t) \triangleq\uppi^{(K)}_{i,i}(t)\big/\|\uppi^{(K)}_{i,i}(t)\|_{1}$\;
			Utilize the normalized policy $\p_{i}(t)$ to sample an action $a_{i}(t)$ from $\V_{i}$\;
			Agent $i$ executes the sampled action  $a_{i}(t)$\;
			\tcp{\textcolor{teal}{Batch Gradient Estimation and Linear Oracles Update(Lines 14-21)}}
			\For{$k=1,\dots,K$}{
				%\tcp{\textcolor{teal}{Batch Stochastic Gradient Estimation (Lines 16-20)}}
				Set the gradient estimation $\mathbf{d}_{i}^{(k)}(t):=\mathbf{0}_{\kappa_{i}}$\;
				\For{$l=1,\dots,L$}{Utilize the policy $\uppi^{(k)}_{i,j}(t)$ to sample an action $\widetilde{a}_{j}(t)\in\V_{i}\cup\{\emptyset\}$ for any $j\in\N$\;
					Estimate the derivatives of $F_{t}$ based on \cref{thm1}, i.e.,
					\begin{equation*}
						\widehat{\frac{\partial F_{t}}{\partial \pi_{i,m}}}\big(\uppi^{(k)}_{i,1}(t),\dots,\uppi^{(k)}_{i,n}(t)\big)\triangleq g^{(k)}_{i,m}(t)\triangleq f_{t}\big(v_{i,m}\big|\cup_{j\neq i}\{\widetilde{a}_{j}(t)\}\big),\forall m\in[\kappa_{i}];
					\end{equation*}\vspace{-0.8em}\\
					Aggregate the gradient estimations $\big(g^{(k)}_{i,1}(t),\dots,g^{(k)}_{i,\kappa_{i}}(t)\big)$ as vector $\mathbf{g}^{(k)}_{i}(t)$\;
					Update $\mathbf{d}^{(k)}_{i}(t)\triangleq\mathbf{d}^{(k)}_{i}(t)+\frac{1}{L}\mathbf{g}^{(k)}_{i}(t)$\;}
				Feed back the batch gradient estimation $\mathbf{d}^{(k)}_{i}(t)$ to the linear oracle $Q_{i}^{(k)}$\;}}}
\end{algorithm}

Like the \cref{alg:SPL}, the core of our \cref{alg:MPL} is primarily composed of three interleaved components, namely, Policy Update and Information Exchange (Lines 5-10), Actions Sampling(Lines 11-13), Surrogate Gradient Estimation(Lines 9-16) and Batch Gradient Estimation and Linear Oracles Update (Lines 14-21). Firstly, at every time step $t\in[T]$, each agent $i \in\N$ employs $K$ online linear  oracles $\{Q_{i}^{(1)},\dots,Q_{i}^{(K)}\}$ to mimic the process of decentralized Meta-Frank-Wolfe~\citep{mokhtari2018decentralized,rezazadeh2023distributed} for maximizing our proposed policy-based continuous extension  $F_{t}$. Specifically, during every inner iteration $k\in[K]$, each agent $i$ pushes the $i$-th local policy vector $\uppi_{i,i}^{(k-1)}$ along the direction $\mathbf{v}_{i}^{(k)}\in\Delta_{\kappa_{i}}$ provided by the oracle $Q_{i}^{(k)}$, while maintaining other policies $\uppi_{i,j}^{(k-1)},j\neq i$ unchanged. Then, agent $i$ exchanges the updated policy vector $\big(\y^{(k)}_{i,1}(t),\dots,\y^{(k)}_{i,n}(t)\big)$ with the neighboring agent $j\in\mathcal{N}_{i}$ and simultaneously utilizes these received information to initialize the next policy $\uppi^{(k)}_{i,m}$, i.e., $\uppi^{(k)}_{i,m}\triangleq\max_{j\in\N_{i}\cup\{i\}}\big(\y_{j,m}^{(k)}(t)\big),\forall m\in\N$, where $\y^{(k)}_{i,i}(t)\triangleq\uppi^{(k-1)}_{i,i}(t)+\frac{1}{K}\mathbf{v}_{i}^{(k)}(t)$ and $\y^{(k)}_{i,m}(t)\triangleq\uppi^{(k-1)}_{i,m}(t),\forall m\neq i$. After completing all $K$ iterations, each agent $i$ normalizes the final policy vector $\uppi_{i,i}^{(K)}(t)$ to select an action $a_{i}(t)$ from $\V_{i}$(See Lines 11-13). Next, in Lines 14-20, each agent $i$ uses a $L$-batch stochastic estimation $\mathbf{d}^{(k)}_{i}(t)$  to approximate the first-order partial derivatives of our proposed policy-based continuous extension $F_{t}$ at every policy vector $\big(\uppi^{(k)}_{i,1}(t),\dots,\uppi^{(k)}_{i,n}(t)\big)$ and every coordinate $\pi_{i,m},\forall m\in[\kappa_{i}]$. Finally, each agent $i$ feeds the obtained $L$-batch gradient estimation $\mathbf{d}^{(k)}_{i}(t)$ back to its corresponding linear oracle $Q_{i}^{(k)}$. 

Note that in the process of \cref{alg:MPL}, each agent $i\in\N$ only needs to evaluate the margin contributions of the actions within its own action set $\V_{i}$(See Line 18).

It is important to highlight that, compared with the previous studies~\citep{chen2018online,streeter2008online,zhang2025nearoptimal,zhu2021projection},  the innovations of our proposed \texttt{MA-MPL} algorithm are threefold: First, we utilize a policy-based continuous extension instead of the well-studied \emph{multi-linear extension} to adjust  the linear oracles in Lines 14-21. Second, rather than using a weight matrix to aggregate the received information, we employ the coordinate-wise \emph{maximization operation}  to update the policy vector(See Line 10). Third, to reduce communication complexity, we implement  the batch gradient estimation in Lines 16-20.

Next, we provide the theoretical analysis for the proposed \cref{alg:MPL}. Before that, we introduce some standard assumptions about the linear maximization oracles $Q_{i}^{(k)},\forall i\in\N,k\in[K]$, namely, 

\begin{assumption}\label{ass:2}
	Each linear maximization oracle $Q_{i}^{(k)}$ can achieve a dynamic regret of $\mathcal{O}(\sqrt{V_{T}T})$ where $V_{T}$ is the variation of any feasible path $(\mathbf{u}_1,\dots,\mathbf{u}_T)\in\prod_{t=1}^{T}\Delta_{\kappa_{i}}$, that is to say, $V_{T}\triangleq\sum_{t=2}^{T}\|\mathbf{u}_{t}-\mathbf{u}_{t-1}\|_{2}$ for any path $(\mathbf{u}_1,\dots,\mathbf{u}_T)\in\prod_{t=1}^{T}\Delta_{\kappa_{i}}$
\end{assumption}
\begin{remark}
	It is worth noting that there exist several effective and efficient algorithms that can achieve a regret bound of $\mathcal{O}(\sqrt{V_{T}T})$  for online linear maximization problem, for instance, online Frank-Wolfe~\citep{pmlr-v195-wan23a} and online gradient ascent~\citep{,yang2016tracking,zhao2021improved,zinkevich2003online}.
\end{remark}

\begin{theorem}[Proof provided in \cref{appendix_proof_thm2}]\label{thm:result2}
	Under \cref{ass:2}, when the communication graph $G(\N,\mathcal{E})$ is connected and each set  function $f_{t}$ is monotone $\alpha$-weakly DR-submodular or  $(\gamma,\beta)$-weakly submodular, if we set $L=\mathcal{O}(T)$ and $K=\mathcal{O}(\sqrt{T})$, our proposed \texttt{MA-MPL} algorithm can achieve a dynamic $\rho$-regret bound of $\mathcal{O}\left(d\left(G\right)\sqrt{P_{T}T}\right)$, that is, $\E\left(R_{\rho}^{*}(T)\right)\le\mathcal{O}\left(d\left(G\right)\sqrt{P_{T}T}\right),$ where $\rho=(1-e^{-\alpha})$ or $\rho=\big(\frac{\gamma^{2}(1-e^{-(\beta(1-\gamma)+\gamma^2)})}{\beta(1-\gamma)+\gamma^2}\big)$, respectively.
\end{theorem}
\begin{remark}
	Note that, in \cref{thm:result2}, $d(G)$ represents the diameter of graph $G(\N,\mathcal{E})$ ,i.e., the length of the shortest path between the most distanced nodes. Moreover, $P_{T}\triangleq\sum_{t=2}^{T}|\mathcal{A}_{t}^{*}\triangle\mathcal{A}_{t-1}^{*}|$ is the deviation of maximizer sequence and $\triangle$ denotes the symmetric difference.
\end{remark}

\section{Additional Experimental Details and Results}\label{appendix:experiments}
In this section, we test the effectiveness of our proposed \cref{alg:SPL} and  \cref{alg:MPL}  in two different multi-target tracking scenarios.
\subsection{Target Tracking with Facility-Location Objective Functions}\label{sec:facility-location}
In line with the prior studies~\citep{xu2023online,zhang2025nearoptimal}, we consider a 2-dimensional plane where 20 agents are deployed to track $30$ moving targets over a duration of 25 seconds, subdivided into $T\triangleq1250$ discrete iterations. At every iteration, agents must determine their movement direction from ``up", ``down", ``left", ``right", or ``diagonally". Simultaneously, agents also need to adjust their speeds from a predefined set of $5$, $10$, or $15$ units/s. As for targets, we categorize them into three different types: the unpredictable `Random' and the structured `Polyline' as well as the challenging `Adversarial'. More specifically, at every iteration,  `Random' target will change its movement angle $\theta$ randomly from $[0,2\pi]$ and steers at a random speed between 5 units/s and 10 units/s. In contrast,  the `Polyline' target will maintain its trajectory and only
behaves like the `Random' target at the specific $\{0, \lfloor\frac{T}{k}\rfloor, 2\lfloor\frac{T}{k}\rfloor\dots, (k-1)\lfloor\frac{T}{k}\rfloor\}$-th iteration where $T$ is the predefined total iterations and $k$ is a random number from $\{1,2,4\}$. Regarding the `Adversarial' target, it mimics the `Random' targets when all agents are beyond a 20 units. Nevertheless, upon detecting an agent within the 20-unit range, the `Adversarial' target will evade at a speed of 15 units/s for one second, pointing to the direction that maximizes the average distance of all agents. %Furthermore, in simulations, we initialize the starting positions of all agents and targets randomly within 20-unit radius circle centered at the origin. 

Generally speaking, every motion of any agent can be characterized by three key parameters, namely, its movement angle $\theta$, speed  $s$ and the unique identifier $i$. With all these three parameters, then the action set  $\V_{i}$  available to each agent $i\in[20]$ can be mathematically represented as:
\begin{equation*}
\V_{i}=\{(\theta,s,i):s\in\{5,10,15\}\text{units/s},\theta\in\{\frac{\pi}{4},\frac{\pi}{2},\frac{3\pi}{4},\pi,\dots,2\pi\}\},\forall i\in[20],
\end{equation*} where each tuple $(\theta,s,i)$ encodes a specific action of agent $i$, that is, it will move at a speed of $s$ in the direction of  $\theta$. Furthermore, at each time step $t\in[T]$, we denote the location of each target $j\in[30]$ as $o_{j}(t)$. Similarly, we also utilize the symbol $o_{(\theta,s,i)}(t)$ to represent the new position of agent $i$ after moving from its previous location at time $(t-1)$ with a movement angle $\theta$ and speed $s$.

To enhance the tracking quality, a common strategy is to minimize the distances between agents and targets. Inspired by this idea, many studies~\citep{corah2021scalable,xu2023bandit,xu2023online,zhang2025nearoptimal} naturally consider the following facility-location objective function for agents at every iteration $t\in[T]$, that is,
\begin{equation*}
	f_{t}(S)\triangleq\sum_{j=1}^{30}\max_{(\theta,s,i)\in S}\frac{1}{\|o_{(\theta,s,i)}(t)-o_{j}(t)\|_{2}},
\end{equation*} where $\|o_{(\theta,s,i)}(t)-o_{j}(t)\|_{2}$ represents the Euclidean distance between the location  $o_{j}(t)$ of target $j$ and the new position $o_{(\theta,s,i)}(t)$ after agent $i$ executing the action $(\theta,s,i)$ and $S$ is a subset of the ground action set $\V\triangleq\cup_{i=1}^{n}\V_{i}$. It is important to note that the larger the value of $\frac{1}{\|o_{(\theta,s,i)}(t)-o_{j}(t)\|_{2}}$ becomes, the closer the action  $(\theta,s,i)$ drives agent $i$ to the target $j$. 

Considering this facility-location utility set function $f_{t}$ and the truth that each agent only can execute  one decision from $\V_{i}$ at every time $t\in[T]$,  then we can easily transform the aforementioned multi-target tracking problem as a multi-agent online coordination(MA-OC) problem introduced in \cref{sec:problem}. Particularly, at each time $t\in[T]$, we need to tackle the following facility-location utility set function maximization problem in a multi-agent manner, namely,
\begin{equation}\label{appendix:simulation:problem1}
	\max f_{t}(S),\ \ \text{ s.t.}\ S\subseteq\V\ \text{and}\ |S\cap\V_{i}|=1,\forall i\in\N.
\end{equation}  
Furthermore, numerous studies~\citep{prajapat2024submodular,xu2023online,zhang2025nearoptimal} have verified that this facility-location objective function $f_{t}$ is \emph{monotone submodular}. As a result, the problem~\eqref{appendix:simulation:problem1} can be equivalently reformulated as the problem~\eqref{equ_problem} in \cref{sec:problem}, i.e., 
\begin{equation*}
	\max f_{t}(S),\ \ \text{ s.t.}\ S\subseteq\V\ \text{and}\ |S\cap\V_{i}|\le1,\forall i\in\N.
\end{equation*} 
\newpage
\begin{figure*}[h]
	\centering
	\subfigure[Average Utility\label{graph1}]{\includegraphics[scale=0.195]{New_Coordination-1-8-1.pdf}}
	\subfigure[Average Distance \label{graph2}]{\includegraphics[scale=0.195]{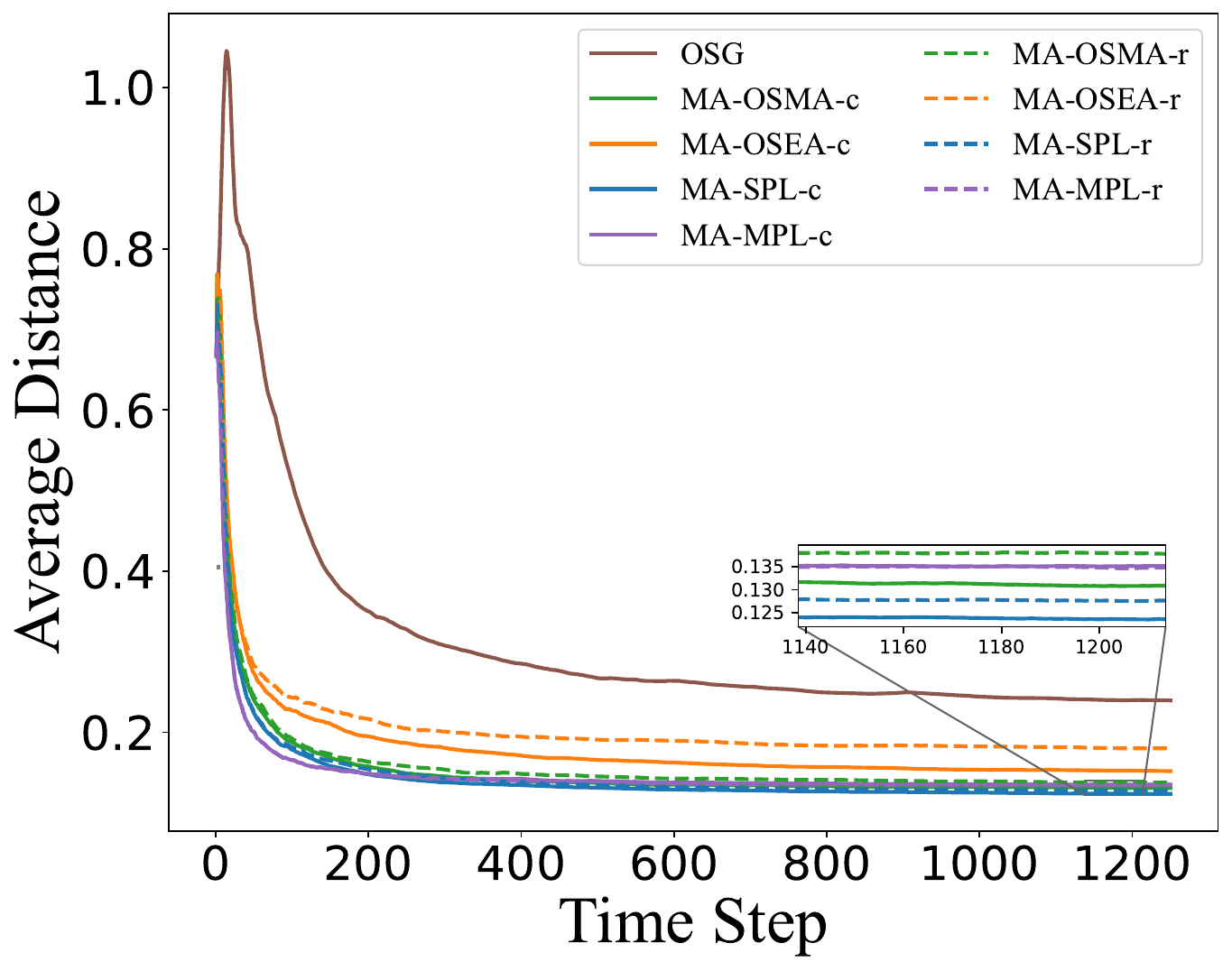}}
	\subfigure[Average Number \label{graph3}]{\includegraphics[scale=0.195]{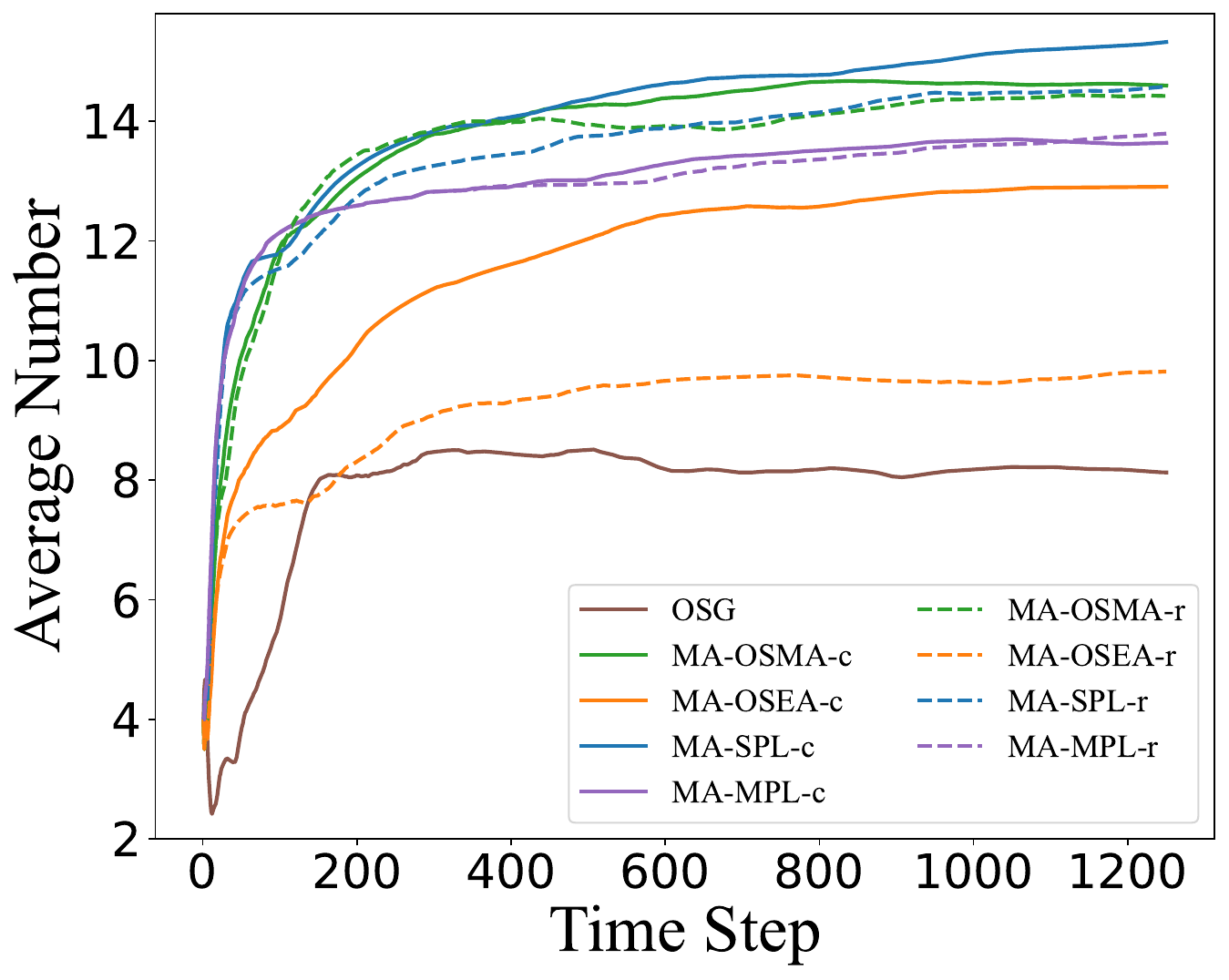}}
	\vspace{-0.5em}
	\caption{Comparison of the running average utility and the runing average number of targets within 1 unit as  well as the running average distance of Top-5 nearest targets of OSG, MA-OSMA, MA-OSEA, MA-OSEA with our proposed \texttt{MA-SPL} and \texttt{MA-MPL} algorithms on the multi-target tracking scenario with `Random':`Adversarial':`Polyline'=$8$:$1$:$1$.}\label{graph:total1}\vspace{-1.0em}
\end{figure*}
\begin{figure*}[h]
	\centering
	\subfigure[Average Utility\label{graph12}]{\includegraphics[scale=0.195]{New_Coordination-1-6-3.pdf}}
	\subfigure[Average Distance \label{graph22}]{\includegraphics[scale=0.195]{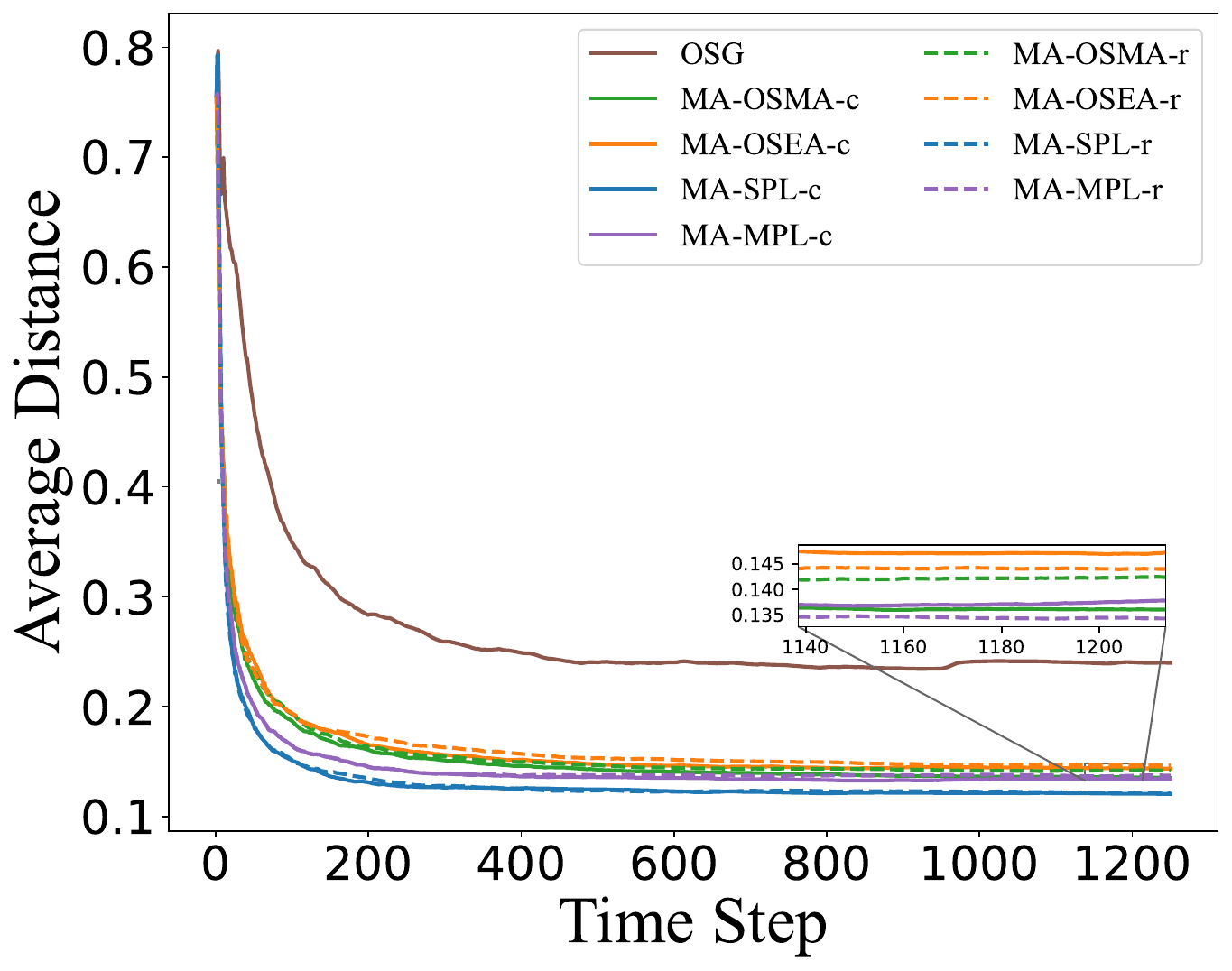}}
	\subfigure[Average Number\label{graph32}]{\includegraphics[scale=0.195]{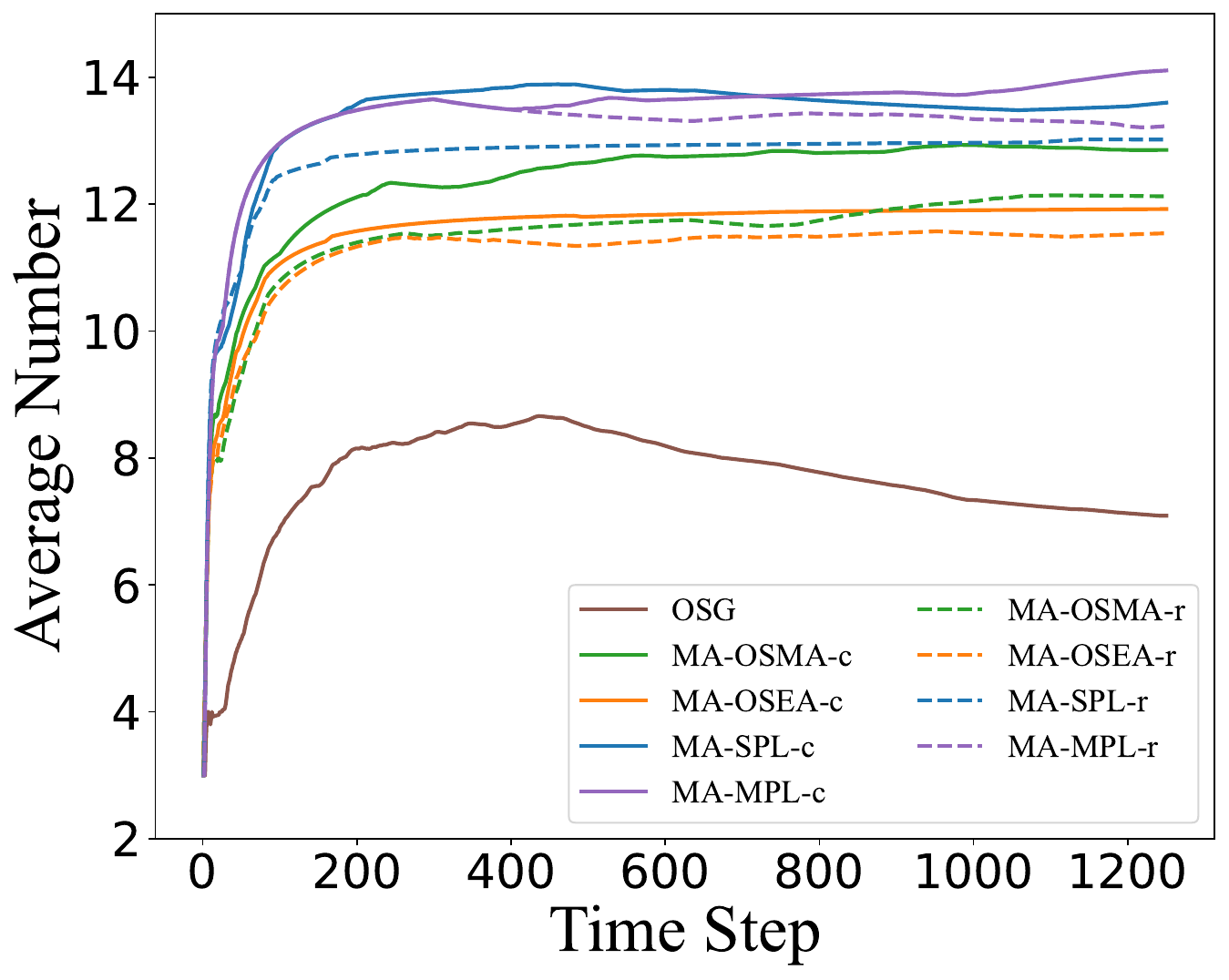}}
	\vspace{-0.5em}
	\caption{Comparison of the running average utility and the runing average number of targets within 1 unit as  well as the running average distance of Top-5 nearest targets of OSG, MA-OSMA, MA-OSEA, MA-OSEA with our proposed \texttt{MA-SPL} and \texttt{MA-MPL} algorithms on the multi-target tracking scenario with `Random':`Adversarial':`Polyline'=$6$:$3$:$1$.}\label{graph:total2}\vspace{-1.0em}
\end{figure*}
\begin{figure*}[h]
	\centering
	\subfigure[Average Utility\label{graph13}]{\includegraphics[scale=0.195]{New_Coordination-1-4-5.pdf}}
	\subfigure[Average Distance \label{graph23}]{\includegraphics[scale=0.195]{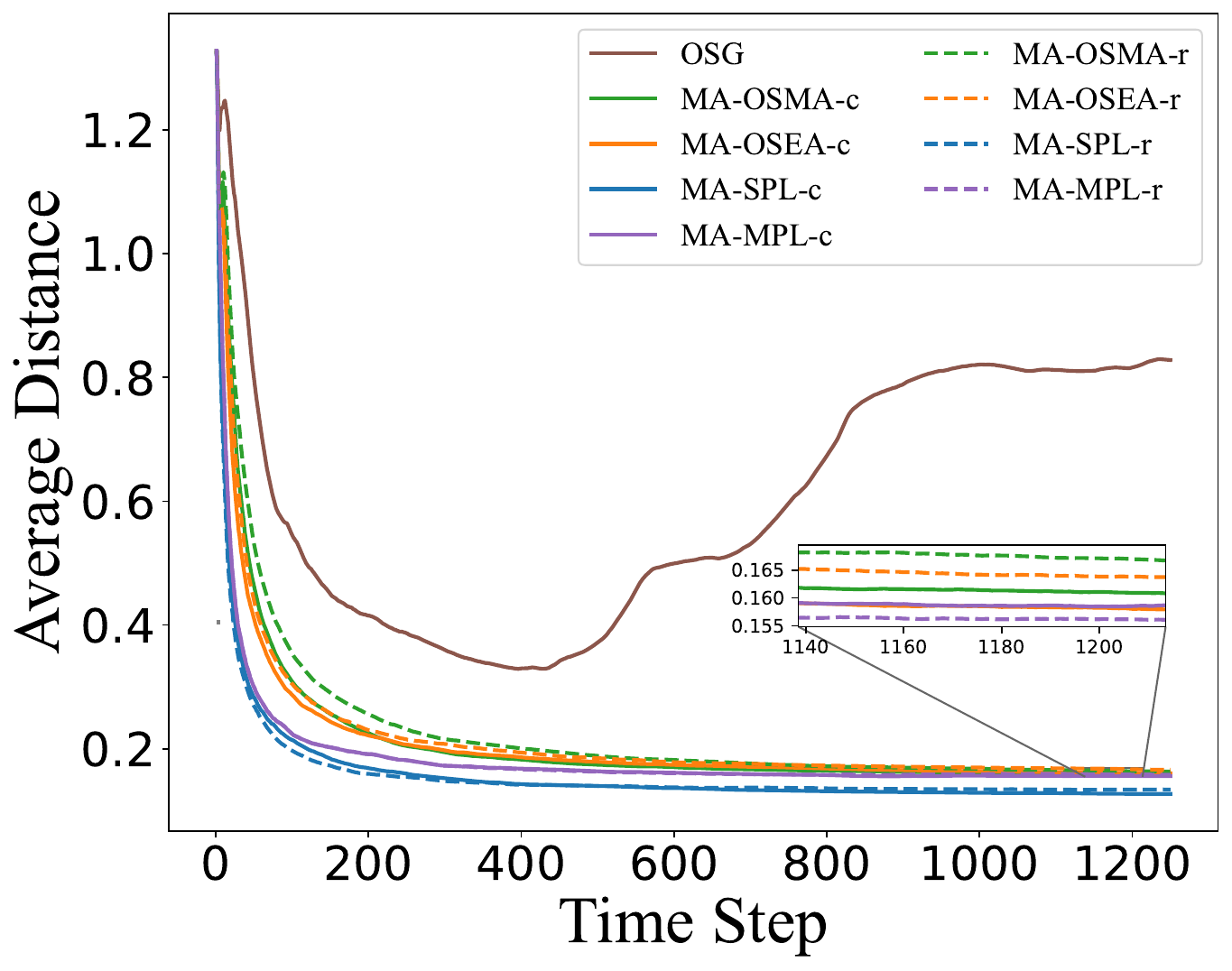}}
	\subfigure[Average Number\label{graph33}]{\includegraphics[scale=0.195]{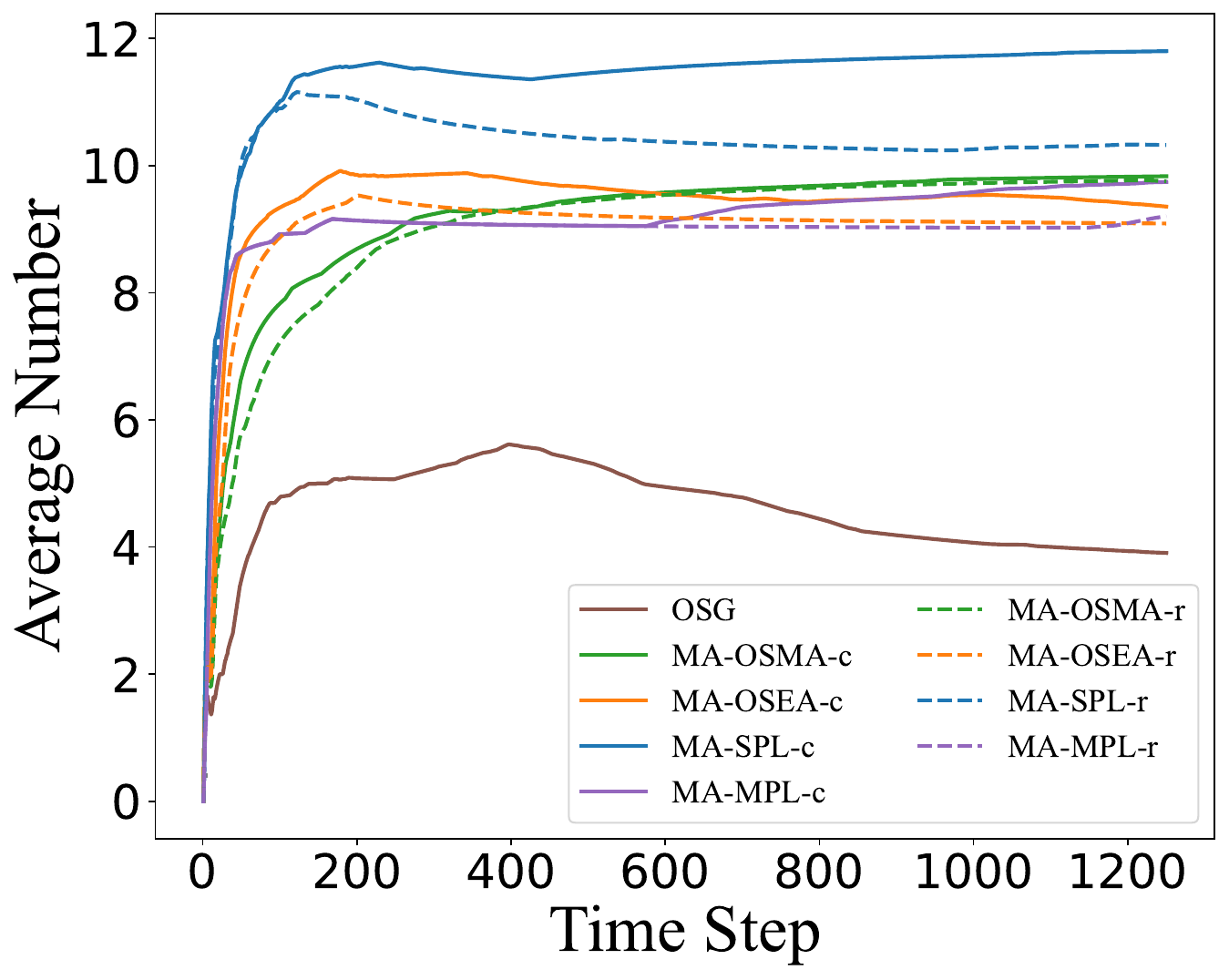}}
	\vspace{-0.5em}
	\caption{Comparison of the running average utility and the runing average number of targets within 1 unit as  well as the running average distance of Top-5 nearest targets of OSG, MA-OSMA, MA-OSEA, MA-OSEA with our proposed \texttt{MA-SPL} and \texttt{MA-MPL} algorithms on the multi-target tracking scenario with `Random':`Adversarial':`Polyline'=$4$:$5$:$1$.}\label{graph:total3}
\end{figure*}
In simulations, we initialize the starting positions of all agents and targets randomly within 20-unit radius circle centered at the origin. Furthermore, we  consider different proportions of `Random', `Polyline', and `Adversarial' targets. Specifically, we set the proportions of targets as `Random':`Adversarial':`Polyline'=$8$:$1$:$1$ in \cref{graph:total1} and $6$:$3$:$1$ in \cref{graph:total2} as well as $4$:$5$:$1$ in \cref{graph:total3}. Like~\citep{zhang2025nearoptimal}, we also use the suffixes to represent two different choices for communication graphs, where `c' stands for a complete graph and `r' denotes an Erdos-Renyi random graph with average degree $4$. According to the results in \cref{graph1}, \cref{graph12} and \cref{graph13}, we can find that the average utility of our proposed \texttt{MA-SPL} algorithm can significantly exceed the state-of-the-art MA-OSMA and MA-OSEA algorithms in \citep{zhang2025nearoptimal}, which is consistent with our  \cref{thm:result1}, that is to say, our proposed \texttt{MA-SPL} algorithm can achieve a tight $(1-\frac{c}{e})$-approximation guarantee for submodular objectives, while the MA-OSMA and MA-OSEA in \citep{zhang2025nearoptimal} only can guarantee a sub-optimal $(\frac{1-e^{-c}}{c})$-approximation. Like \citep{zhang2025nearoptimal}, in \cref{graph2}, \cref{graph22} and \cref{graph23}, we also compare the average distance from agents to their closest five targets of our proposed \texttt{MA-SPL} and \texttt{MA-MPL} algorithms against OSG~\citep{xu2023online}, MA-OSMA and MA-OSEA algorithms. Similarly, we observe that our proposed  \texttt{MA-SPL} can effectively reduce the average distance between agents and targets. Furthermore, from  \cref{graph3}, \cref{graph32}, and \cref{graph33}, we also can infer that, compared with OSG, MA-OSMA, and MA-OSEA algorithms, our proposed \texttt{MA-SPL} can make more targets gather within 1 unit of agents. It is worth noting that as the proportion of `Adversarial' targets increases, both the maximum runing average utility and the maximum running average number of targets within 1 unit exhibit a downward trend.
\subsection{Target Tracking under Extended Kalman-Filter Framework}
In previous \cref{sec:facility-location}, we considered a simplified multi-target tracking model, that is, we assume that the monitoring quality of each moving target $j$ only depends on its nearest agent. However, in many real-world scenarios, due to agents' different sensing capabilities and varying observation angles, relying \emph{solely} on the information collected by the nearest agent for accurately tracking the state of each target is unrealistic. Instead, we generally need to aggregate the data from multiple agents to comprehensively reconstruct the targets' behaviors. 

Recently, in order to obtain an accurate estimation of every target's location, \citep{hashemi2019submodular,hashemi2020randomized} employed an extend Kalman-filter(EKF) framework to process the observations of multiple distinct agents. Specifically, let us consider a target tracking task using a swarm of agents(UAVs) equipped with GPS and radar systems. In this scenario, at each time $t\in[T]$, each agent $i$ can measure the range to each target $j$ throughout its radar system and also can obtain its own position  information from GPS. Furthermore, like the work~\citep{hashemi2019submodular}, we also assume that the range measurements of the radar systems follow a quadratic model, namely,
\begin{equation}\label{equality:EKF}
r_{(\theta,s,i)\rightarrow j}(t)=\frac{1}{2}\|o_{(\theta,s,i)}(t)-o_{j}(t)\|^{2}_{2}+\xi_{(\theta,s,i)\rightarrow j}(t),
\end{equation} where  the symbol $r_{(\theta,s,i)\rightarrow j}(t)$ denotes the range measurement of agent $i$ to target $j$ after agent $i$ executes the action $(\theta,s,i)$, $o_{j}(t)$ is the location of target $j$, $\xi_{(\theta,s,i)\rightarrow j}(t)$ is the additive independent noise and the symbol $o_{(\theta,s,i)}(t)$ represents the new position of agent $i$ after moving from its previous location at time $t-1$ with a movement angle $\theta$ and speed $s$.

Note that when the new position $o_{(\theta,s,i)}(t)$ and the range measurement $r_{(\theta,s,i)\rightarrow j}(t)$ are known, we can view \cref{equality:EKF}  as a random quadratic experiment of the unknown location $o_{j}(t)$. In other words, different action $(\theta,s,i)$ can lead to  distinct random quadratic observation of the unknown parameter $o_{j}(t)$. Inspired by this perspective,  we can formulated the action selection problem in multi-target tracking task as an experimental design problem, which aims at selecting a feasible subset from the whole collection of  experiments $\{r_{(\theta,s,i)\rightarrow j}(t)=\frac{1}{2}\|o_{(\theta,s,i)}(t)-o_{j}(t)\|^{2}_{2}+\xi_{(\theta,s,i)\rightarrow j}(t) \Big|(\theta,s,i)\in\V\}$ such that the measurements of these selected sub-experiments can accurately estimate the location of every target. Particularly, under the classical Van Trees’ inequality~\citep{van2004detection}, the work~\citep{hashemi2019submodular} established a lower bound for the covariance matrix associated with the EKF estimator of the location $o_{j}(t)$. Specifically, at time step $t\in[T]$, if we consider the sub-experiment $S\subseteq \V$,  then the lower-bound matrix $\textbf{B}^{j}_{S}(t)$ in the Van
Trees’ inequality for the EKF estimate of the unknown location $o_{j}(t)$ can be expressed as (Please refer to Theorem 2 in \citep{hashemi2019submodular}):
\begin{equation}\label{EKF_lower}
	\textbf{B}^{j}_{S}(t)=\left(\sum_{(\theta,s,i)\in S}\frac{1}{\sigma^{2}_{(\theta,s,i)\rightarrow j}(t)}\left(\textbf{P}+\textbf{z}_{(\theta,s,i)\rightarrow j}\textbf{z}_{(\theta,s,i)\rightarrow j}^{T}\right)+\textbf{I}_{j}(t)\right)^{-1},
\end{equation} where $\textbf{z}_{(\theta,s,i)\rightarrow j}\triangleq o_{(\theta,s,i)}(t)-o_{j}(t-1)-\E\left(o_{j}(t)-o_{j}(t-1)\right)$, $\sigma^{2}_{(\theta,s,i)\rightarrow j}(t)\triangleq\text{Var}\left(\xi_{(\theta,s,i)\rightarrow j}(t)\right)$, $\textbf{P}\triangleq\text{Cov}\left(o_{j}(t)-o_{j}(t-1)\right)$ and $\textbf{I}_{j}(t)$ is the  Fisher information matrix  associated with the normalized random gap $\big(o_{j}(t)-o_{j}(t-1)-\E\left(o_{j}(t)-o_{j}(t-1)\right)\big)$. Note that in \cref{EKF_lower}, the tuple $(\theta,s,i)$ represents not only a selected action but also an observation experiment.

With this lower-bound covariance matrix $\textbf{B}^{j}_{S}(t)$, we then can utilize various so-called alphabetical criteria~\citep{chaloner1995bayesian,chamon2017approximate} to design a utility function for selecting a superior action set such that the resulting EKF estimation  can accurately approximate the location of each target. One commonly used strategy is to employ the A-optimality, i.e., we consider minimizing the trace of the lower-bound covariance matrix  $\textbf{B}^{j}_{S}(t)$  or equivalently maximize:
\begin{equation}\label{EKF_obj}
	f_{t}(S)=\sum_{j=1}^{30}\left(\text{Tr}\left(\textbf{I}^{-1}_{j}(t)\right)-\text{Tr}\left(\textbf{B}^{j}_{S}(t)\right)\right)
\end{equation} where ``$\text{Tr}$'' is the trace of matrix and we also consider $30$ moving targets in simulation. 

Thus, under the A-optimality criterion, we can model the agents' action selection task as a special instance of the multi-agent online coordination problem with the set objective function $f_{t}$ defined in \cref{EKF_obj}. Furthermore, according to  recent studies~\citep{harshaw2019submodular,hashemi2019submodular,thiery2022two}, we can show that the utility function $f_{t}$ in \cref{EKF_obj} is monotone $\alpha$-weakly DR-submodular and $(\gamma,\beta)$-weakly submodular (See Theorem 6 in \citep{hashemi2019submodular} and
Theorem C.2 in \citep{thiery2022two}).

In our simulations, to simplify the computation of the lower-bound matrix  $\textbf{B}^{j}_{S}(t)$,  we model the movement of each target as a two-dimensional Brownian motion. Specifically, we set $o_{j}(t)\triangleq o_{j}(t-1)+0.02*\N(\textbf{0}_{2},\textbf{I}_{2})$ where $\textbf{I}_{2}$ is the 2-dimensional identity matrix.  Moreover, we assume that the noise $\xi_{(\theta,s,i)\rightarrow j}$ follows an independent normal distribution, i.e., $\xi_{(\theta,s,i)\rightarrow j}\sim\N(0,0.01)$. As for agents, at every iteration $t\in[T]$, we adjust their speeds from a set of $2$, $7$, or $12$ units/s and simultaneously change their movement directions from ``up", ``down", ``left", ``right", or ``diagonally". As a result, the action set  $\V_{i}$  available to each agent $i\in[20]$ can be mathematically formulated as:
\begin{equation*}
	\V_{i}=\{(\theta,s,i):s\in\{2,7,12\}\text{units/s},\theta\in\{\frac{\pi}{4},\frac{\pi}{2},\frac{3\pi}{4},\pi,\dots,2\pi\}\},\forall i\in[20],
\end{equation*} where $\theta$ denotes the movement angle , $s$ is the speed  and  $i$ represents the unique identifier. 

\begin{figure*}[t]
	\centering
	\subfigure[Average Utility\label{graph41}]{\includegraphics[scale=0.18]{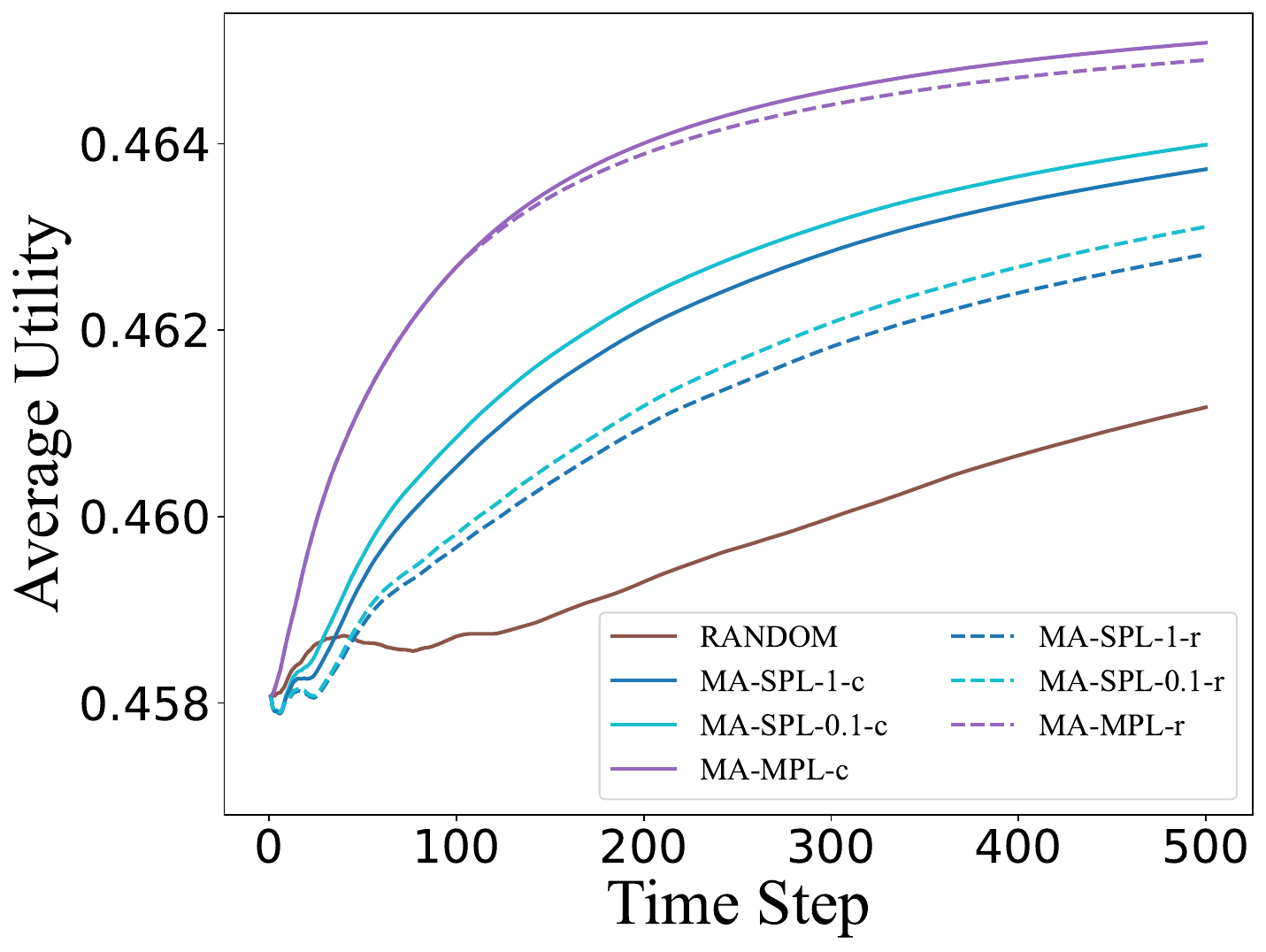}}
	\subfigure[Complete Graph \label{graph42}]{\includegraphics[scale=0.18]{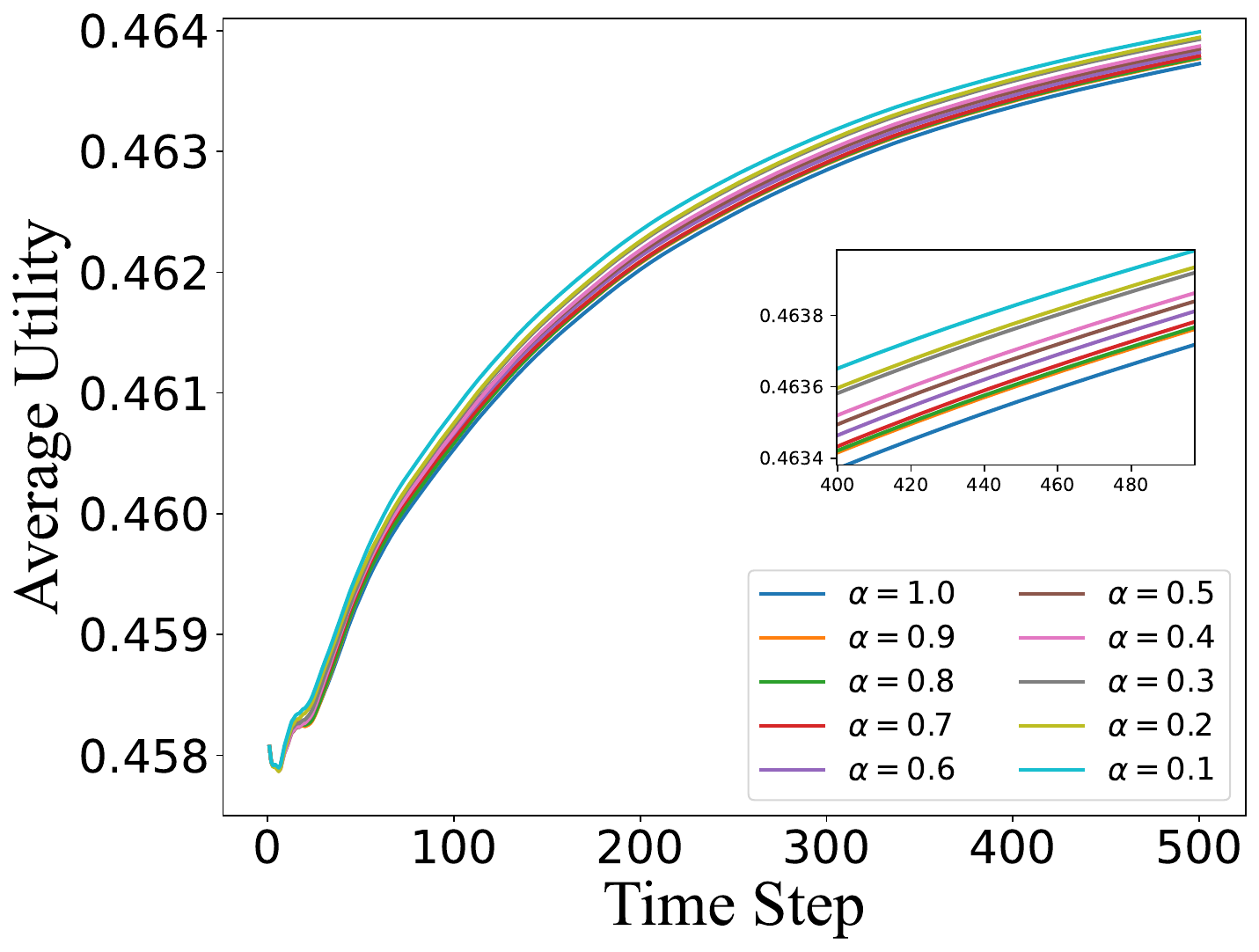}}
	\subfigure[Erdos-Renyi  Graph \label{graph43}]{\includegraphics[scale=0.18]{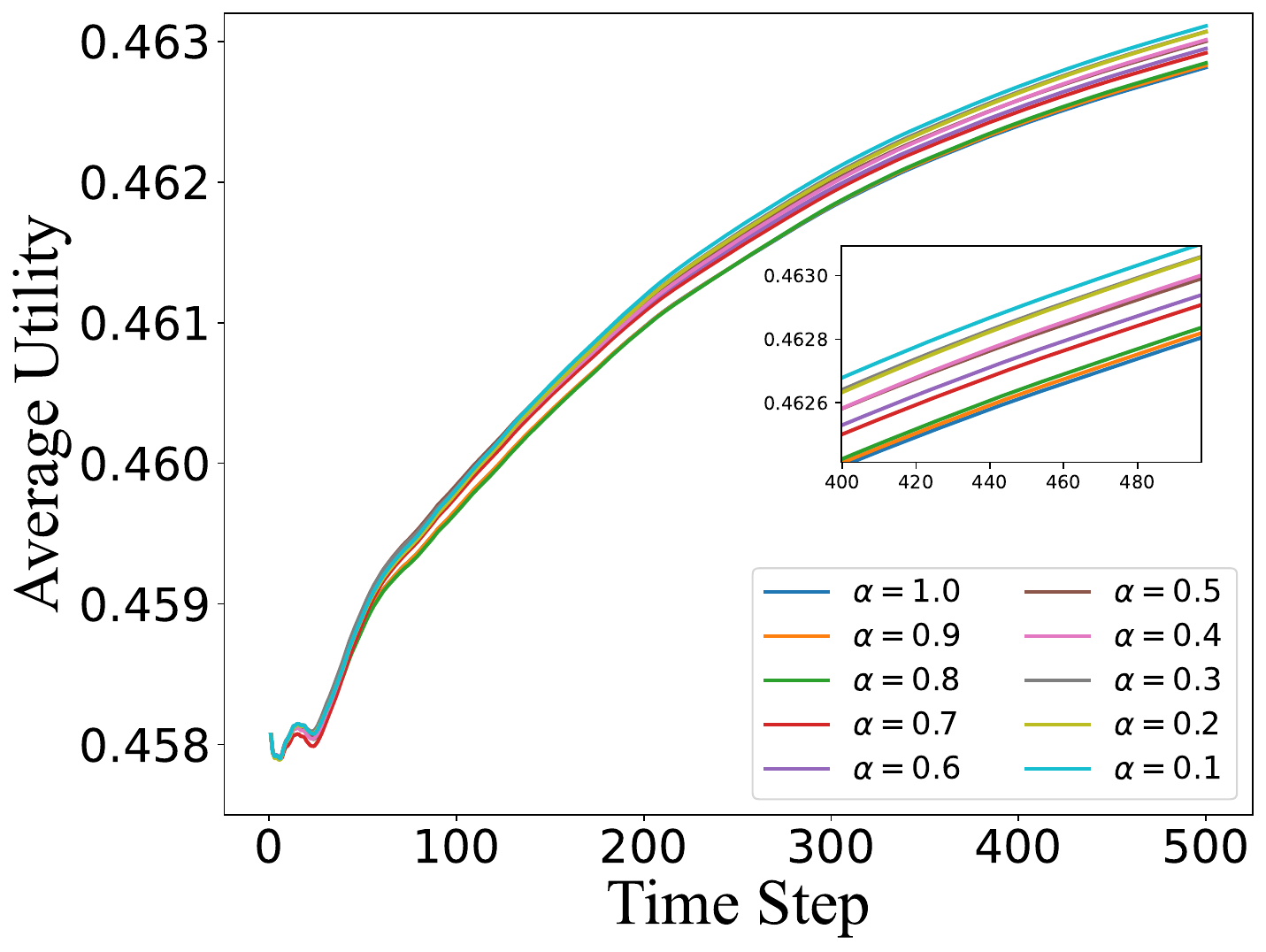}}
	\vspace{-0.5em}
	\caption{ In \cref{graph41}, we compare the running average utility of our proposed \texttt{MA-SPL} and \texttt{MA-MPL} algorithms with that of the `RANDOM'  baseline in a multi-target tracking simulation with A-optimal objective functions. In `RANDOM' baseline, each agent $i$ randomly selects an action from its own action set $\V_{i}$ at every iteration. As for \cref{graph42} and \cref{graph43}, we illustrate the impact of different DR ratios $\alpha\in\{0.1,0.2,\dots,1\}$ on \texttt{MA-SPL} algorithm.   Particularly in \cref{graph42}, we considers a complete communication graph among agents, whereas we employs an Erdos-Renyi graph with an average degree of $4$ in \cref{graph43}. Note that \cref{graph41} only shows  the results for the best scenario ($\alpha=0.1$) and the worst-case scenario ($\alpha=1$) for our proposed \texttt{MA-SPL} algorithm.}\label{graph:total4}\vspace{-1.0em}
\end{figure*}

 Given the unknown diminishing-return(DR) ratio $\alpha$ of our investigated set function $f_{t}$, we test $10$ different configurations from $\alpha=0.1$ to $\alpha=1$ for our proposed \texttt{MA-SPL} and then present the results  in \cref{graph42} and \cref{graph43}. Particularly in \cref{graph42}, we considers a complete communication graph among agents, whereas we employs an Erdos-Renyi graph with an average degree of $4$ in \cref{graph43}. Subsequently, we compare the results of the best-case scenario ($\alpha=0.1$) and the worst-case scenario ($\alpha=1$) of our proposed \texttt{MA-SPL} algorithm with the \emph{parameter-free} \texttt{MA-MPL} algorithm and `RANDOM' baseline in \cref{graph41}. According to the curves in \cref{graph41}, we can find that the running average utility of our proposed \texttt{MA-MPL} and \texttt{MA-SPL} algorithms can significantly exceed the baseline `RANDOM' algorithm, which is in accord with our  \cref{thm:result1} and \cref{thm:result2}. Moreover, we also find that the \emph{parameter-free} \texttt{MA-MPL} algorithm can effectively outperform the \texttt{MA-SPL} algorithm associated with a $0.1$-network search.
It is worth noting that no previous works explore the MA-OC problem with weakly submodular objectives. Thus, we adopt the `RANDOM' algorithm as a baseline in \cref{graph41}.  In `RANDOM', each agent $i$ randomly selects an action from its own action set $\V_{i}$.  Like the previous \cref{sec:facility-location}, we also use the suffixes to represent two different choices for communication graphs in \cref{graph41}, where `c' stands for a complete graph and `r' denotes an Erdos-Renyi random graph with average degree $4$.
 
\subsection{More Details on Experimental Setups}
This subsection discusses some additional details about our experiments. 

At first, we describe the parameter configurations about our proposed `\texttt{MA-SPL}', `\texttt{MA-MPL}', `MA-OSMA' and `MA-OSEA'. Specifically, we make the following setups:
\begin{itemize}
	\item In `\texttt{MA-SPL}', namely, the \cref{alg:SPL}, we set the step size $\eta_{t}=\frac{1}{\sqrt{T}}$ and employ a $10$ batch of stochastic estimation to approximate the surrogate gradient of our proposed policy-based continuous extension in \cref{sec:surrogate_function}. As for the projection in Line 19 of \cref{alg:SPL}, we utilize the CVX optimization solver~\citep{grant2014cvx}.
	\item In `\texttt{MA-MPL}', namely, the \cref{alg:MPL} of \cref{appendix:alg2}, we set the batch size $L=10$ and the number of oracles $K=15$. As for the online linear maximization oracle, we utilize the online gradient ascent algorithm~\citep{yang2016tracking,zinkevich2003online} with step size $\eta=\mathcal{O}(\frac{1}{\sqrt{T}})$.
	\item In `MA-OSMA', namely Algorithm 1 of \citep{zhang2025nearoptimal}, we consider the Euclidean distance with $\phi(\x)=\frac{\|\x\|_{2}^{2}}{2}$, set the step size $\eta_{t}=\frac{1}{\sqrt{T}}$ and implement a $10$ batch of stochastic estimation to approximate the surrogate gradient of \emph{multi-linear extension}. Similarly, we also use the CVX solver for projection operations.
	\item In `MA-OSEA', namely Algorithm 2 of \citep{zhang2025nearoptimal}, we also set the step size $\eta_{t}=\frac{1}{\sqrt{T}}$ and consider the mixing parameter $\gamma=1/T^{1.5}$.
\end{itemize}
As for the communication graph $G$,  we consider two different setups:
\begin{itemize}
	\item `Complete graph' where we set the weight $w_{ij}=\frac{1}{n},\forall i,j\in[n]$ where $n=|\N|$ is the number of agents. 
	\item `Erdos-Renyi random graph with average degree $4$ where if the edge $(i,j)$ is an edge of the graph, we set $w_{ij}\triangleq1/(1+\max(d_{i},d_{j}))$ where $d_{i}$ is the degree of agent $i\in\N$ and when $(i,j)$ is not an edge of the graph and $i\neq j$, we consider $w_{ij}=0$. Finally, we set $w_{ii}\triangleq1-\sum_{j\in\mathcal{N}_{i}}w_{ij}$ where $\N_{i}$ is the neighboring nodes of agent $i$. 
\end{itemize}
Furthermore, for all curves related to \texttt{MA-SPL},  \texttt{MA-MPL}, MA-OSMA, MA-OSEA and OSG, we repeat these algorithms \textbf{five runs} and then report the average result. Note that, in \cref{graph1}, \cref{graph12}, \cref{graph13} and \cref{graph41}, the running average utility at any time $t$ is defined as $\big(\sum_{t_{1}\in[t]}\frac{f_{t}(\cup_{i\in[n]}\{a_{i}(t_{1})\})}{t}\big)$ where $a_{i}(t_{1})$ is the action chosen by agent $i\in[n]$ at time $t_{1}$.

\section{Proof of \texorpdfstring{\cref{thm1}}{}}\label{appendix_proof_thm}
In this section, we prove the \cref{thm1}.

$\textbf{1):}$ From \cref{def_extension}, we have that
\begin{equation}\label{Appendix_A.1}
	\begin{aligned}
&\frac{\partial F_{t}}{\partial \pi_{i,m}} (\uppi_{1},\dots,\uppi_{n})\\&=\sum_{a_{j}\in\V_{j}\cup\{\emptyset\},\forall j\in\N}\Big(f_{t}\big(\cup_{j=1}^{n}\{a_{j}\}\big)\frac{\partial\left(\prod_{j=1}^{n}p(a_{j}|\uppi_{j})\right)}{\partial \pi_{i,m}}\Big)\\
&=\sum_{a_{j}\in\V_{j}\cup\{\emptyset\},\forall j\in\N}\Big(f_{t}\big(\cup_{j=1}^{n}\{a_{j}\}\big)\frac{\partial p(a_{i}|\uppi_{i})}{\partial \pi_{i,m}}\prod_{j\neq i,j\in\N}p(a_{j}|\uppi_{j})\Big).
	\end{aligned}
\end{equation}

Note that $p(v_{i,m}|\uppi_{i}) =\pi_{i,m},\forall i\in[n],\forall m\in[\kappa_{i}]$ and $p(\emptyset | \uppi_{i}) = 1 - \sum_{m=1}^{K_{i}} \pi_{i,m},\forall i\in\N$. Therefore, we have $\frac{\partial p(a_{i}|\uppi_{i})}{\partial \pi_{i,m}}=1$ when $a_{i}=v_{i,m}$, $\frac{\partial p(a_{i}|\uppi_{i})}{\partial \pi_{i,m}}=-1$ when $a_{i}=\emptyset$ and  $\frac{\partial p(a_{i}|\uppi_{i})}{\partial \pi_{i,m}}=0$ when $a_{i}\notin\{v_{i,m},\emptyset\}$. Then, according to Eq.\eqref{Appendix_A.1}, we can get the following equality:
\begin{equation}\label{Appendix_A.2}
	\begin{aligned}
		&\frac{\partial F_{t}}{\partial \pi_{i,m}} (\uppi_{1},\dots,\uppi_{n})\\
		&=\sum_{a_{j}\in\V_{j}\cup\{\emptyset\},\forall j\in\N}\Big(f_{t}\big(\cup_{j=1}^{n}\{a_{j}\}\big)\frac{\partial p(a_{i}|\uppi_{i})}{\partial \pi_{i,m}}\prod_{j\neq i,j\in\N}p(a_{j}|\uppi_{j})\Big)\\
		&=\sum_{a_{j}\in\V_{j}\cup\{\emptyset\},\forall j\neq i}\left(\sum_{a_{i\in\V_{i}}}\Big(f_{t}\big(\cup_{j=1}^{n}\{a_{j}\}\big)\frac{\partial p(a_{i}|\uppi_{i})}{\partial \pi_{i,m}}\prod_{j\neq i,j\in\N}p(a_{j}|\uppi_{j})\Big)\right)\\
		&=\sum_{a_{j}\in\V_{j}\cup\{\emptyset\},\forall j\neq i}\left(\left(f_{t}\big(\{v_{i,m}\}\bigcup\left(\cup_{j\neq i}\{a_{j}\}\right)\big)-f_{t}\big(\cup_{j\neq i,j\in\N}\{a_{j}\}\big)\right)\prod_{j\neq i,j\in\N}p(a_{j}|\uppi_{j})\right)\\
		&=\sum_{a_{j}\in\V_{j}\cup\{\emptyset\},\forall j\neq i}\left(f_{t}\big(v_{i,m}\big|\cup_{j\neq i,j\in\N}\{a_{j}\}\big)\prod_{j\neq i,j\in\N}p(a_{j}|\uppi_{j})\right)\\
		&=\E_{a_{j}\sim\uppi_{j},\forall j\in\N}\Big(f_{t}\big(v_{i,m}\big|\cup_{j\neq i,j\in\N}\{a_{j}\}\big)\Big),
	\end{aligned}
\end{equation} where $a_{j}\sim\uppi_{j}$ indicates that action $a_{j}$
is  randomly selected from $\V_{j}\cup\{\emptyset\}$ based on the policy $\uppi_{j}$.

$\textbf{2):}$ When $f_{t}$ is monotone, we can know that $f_{t}\big(v_{i,m}\big|\cup_{j\neq i,j\in\N}\{a_{j}\}\big)\ge0$ such that $\frac{\partial F_{t}}{\partial \pi_{i,m}} (\uppi_{1},\dots,\uppi_{n})\ge0$. In other words, for any $(\uppi_{1},\dots,\uppi_{n})\in\prod_{i=1}^{n}\Delta_{\kappa_{i}}$, we have $\nabla F_{t}(\uppi_{1},\dots,\uppi_{n})\ge\mathbf{0}$, namely, for any two point $(\uppi^{a}_{1},\dots,\uppi^{a}_{n})\in\prod_{i=1}^{n}\Delta_{\kappa_{i}}$ and $(\uppi^{b}_{1},\dots,\uppi^{b}_{n})\in\prod_{i=1}^{n}\Delta_{\kappa_{i}}$, if  $\uppi^{a}_{i}\le\uppi^{b}_{i}$ for any $i\in\N$, we have $ F_{t}(\uppi^{a}_{1},\dots,\uppi^{a}_{n})\le F_{t}(\uppi^{b}_{1},\dots,\uppi^{b}_{n})$. So $F_{t}$ is monotone.

\textbf{3):} For any fixed two policy vector $(\uppi^{a}_{1},\dots,\uppi^{a}_{n})\in\prod_{i=1}^{n}\Delta_{\kappa_{i}}$ and $(\uppi^{b}_{1},\dots,\uppi^{b}_{n})\in\prod_{i=1}^{n}\Delta_{\kappa_{i}}$ where  $\uppi^{a}_{i}\le\uppi^{b}_{i},\forall i\in\N$, we consider an unified sampling strategy to generate actions. Before that, we set each $\uppi^{a}_{i}=(\pi_{i,1}^{a},\dots,\pi_{i,\kappa_{i}}^{a})$ and $\uppi^{b}_{i}=(\pi_{i,1}^{b},\dots,\pi_{i,\kappa_{i}}^{b}),\forall i\in\N$. At first, we transfer the sampling process according to policy $\uppi^{a}_{i}$ to a uniform random variable $X_{i}\in[0,1]$, namely,
\begin{equation}\label{def1}
	a(X_{i},\uppi^{a}_{i})\triangleq\left\{\begin{aligned}
		&v_{i,1}\ \ \ \text{If}\ X_{i}\in[0,\pi_{i,1}^{a})\\
		&v_{i,m}\ \ \text{If}\ X_{i}\in[\sum_{k=1}^{m-1}\pi_{i,k}^{a},\sum_{k=1}^{m}\pi_{i,k}^{a})\ \text{for any integer}\ m\in [2,\kappa_{i}]\\
		&\emptyset\ \ \ \ \ \text{If}\ X_{i}\ge\sum_{k=1}^{\kappa_{i}}\pi_{i,k}^{a}
	\end{aligned}\right.
\end{equation}
When $X_{i}$ is an uniform random variable over range $[0,1]$, it is easy to check that the $a(X_{i},\uppi^{a}_{i})$ follows the same law as the policy $\uppi^{a}_{i}$, that is, $\text{Pr}(a(X_{i},\uppi^{a}_{i})=v_{i,m})=\pi^{a}_{i,m}$ and $\text{Pr}(a(X_{i},\uppi^{a}_{i})=\emptyset)=1-\sum_{m=1}^{\kappa_{i}}\pi^{a}_{i,m}$ where the symbol `Pr' denotes the probability.

Similarly, we also can transfer the sampling process according to policy $\uppi^{b}_{i}$ to two independent uniform random variables $X_{i},Y\in[0,1]$, namely,
\begin{equation}\label{def2}
	a(X_{i},Y,\uppi^{a}_{i},\uppi^{b}_{i})\triangleq\left\{\begin{aligned}
		&a(X_{i},\uppi^{a}_{i})\ \ \ \ \ \ \text{If}\ X_{i}<\sum_{k=1}^{\kappa_{i}}\pi_{i,k}^{a}\\
		&v_{i,1}\ \ \ \text{If}\ X_{i}\ge\sum_{k=1}^{\kappa_{i}}\pi_{i,k}^{a}\ \text{and}\ Y\in[0,\frac{\pi_{i,1}^{b}-\pi_{i,1}^{a}}{1-\sum_{k=1}^{\kappa_{i}}\pi_{i,k}^{a}})\\
		&v_{i,m}\ \ \text{If}\ X_{i}\ge\sum_{k=1}^{\kappa_{i}}\pi_{i,k}^{a}\ \text{and}\ Y\in\big[\frac{\sum_{k=1}^{m-1}(\pi_{i,k}^{b}-\pi_{i,k}^{a})}{1-\sum_{k=1}^{\kappa_{i}}\pi_{i,k}^{a}},\frac{\sum_{k=1}^{m}(\pi_{i,k}^{b}-\pi_{i,k}^{a})}{1-\sum_{k=1}^{\kappa_{i}}\pi_{i,k}^{a}}\big),\forall m\in [2,\kappa_{i}]\\
		&\emptyset\ \ \ \ \ \ \text{If}\ X_{i}\ge\sum_{k=1}^{\kappa_{i}}\pi_{i,k}^{a}\ \text{and}\ Y\ge\sum_{k=1}^{\kappa_{i}}(\pi_{i,k}^{b}-\pi_{i,k}^{a})
	\end{aligned}\right.
\end{equation} 

From part \textbf{1)} and fixed a sequence of independent uniform variables $X_{i},Y_{i},\forall i\in\N$, we can have that
\begin{equation*}
\frac{\partial F_{t}}{\partial \pi_{i,m}} (\uppi^{a}_{1},\dots,\uppi^{a}_{n})=\E_{X_{i}}\Big(f_{t}\big(v_{i,m}\big|\cup_{j\neq i,j\in\N}\{a(X_{i},\uppi_{i}^{a})\}\big)\Big),
\end{equation*}
and
\begin{equation*}
	\frac{\partial F_{t}}{\partial \pi_{i,m}} (\uppi^{b}_{1},\dots,\uppi^{b}_{n})=\E_{(X_{i},Y_{i})}\Big(f_{t}\big(v_{i,m}\big|\cup_{j\neq i,j\in\N}\{a(X_{i},Y_{i},\uppi_{i}^{a},\uppi_{i}^{b})\}\big)\Big).
\end{equation*}
From Eq.\eqref{def1} and Eq.\eqref{def2}, we can know that $\{a(X_{i},\uppi^{a}_{i})\}\subseteq\{a(X_{i},Y_{i},\uppi^{a}_{i},\uppi^{b}_{i})\}$ such that $\cup_{j\neq i,j\in\N}\{a(X_{i},\uppi_{i}^{a})\}\subseteq\cup_{j\neq i,j\in\N}\{a(X_{i},Y_{i},\uppi_{i}^{a},\uppi_{i}^{b})\}$. Then, from the definition of $\alpha$-DR submodularity, we have that $f_{t}\big(v_{i,m}\big|\cup_{j\neq i,j\in\N}\{a(X_{i},\uppi_{i}^{a})\}\big)\ge\alpha f_{t}\big(v_{i,m}\big|\cup_{j\neq i,j\in\N}\{a(X_{i},Y_{i},\uppi_{i}^{a},\uppi_{i}^{b})\}\big)$, so $\nabla F_{t}(\uppi^{a}_{1},\dots,\uppi^{a}_{n})\ge \alpha \nabla F_{t}(\uppi^{b}_{1},\dots,\uppi^{b}_{n})$.

\textbf{4):} For any subset $S$ within the constraint of problem~\eqref{equ_problem} and any point $(\uppi_{1},\dots,\uppi_{n})\in\prod_{i=1}^{n}\Delta_{\kappa_{i}}$, when $f_{t}$ is monotone $(\gamma,\beta)$-weakly monotone submodular,  we firstly can show that for any two actions $a_{1},a_{2}\in\V$ and the subset $B\subseteq\V$, we have that  
\begin{equation*}
	\gamma\big(f_{t}(a_1| B\cup\{a_{2}\})+f_{t}(a_{2}|B)\big)=\gamma\big(f_{t}(B\cup\{a_{1},a_{2}\})-f(B)\big)\le f_{t}(a_{1}|B)+ f_{t}(a_{2}|B),
\end{equation*}such that 
\begin{equation}\label{equ_appendix2_3}
	f_{t}(a_{1}|B)\ge \gamma f_{t}(a_1| B\cup\{a_{2}\})-(1-\gamma)f_{t}(a_{2}|B).
\end{equation}
Therefore, we can show that
\begin{equation}\label{equ_appendix2_4}
	\begin{aligned}
		&\sum_{(i,m): v_{i,m}\in S}\frac{\partial F_{t}}{\partial \pi_{i,m}}(\uppi_{1},\dots,\uppi_{n})\\&=\sum_{(i,m): v_{i,m}\in S}\E_{a_{j}\sim\uppi_{j},\forall j\in\N}\Big( f_{t}\big(v_{i,m}\big|\cup_{j\neq i,j\in\N}\{a_{j}\}\big)\Big)\\
		&\ge\sum_{(i,m): v_{i,m}\in S}\E_{a_{j}\sim\uppi_{j},\forall j\in\N}\Big( \gamma f_{t}\big(v_{i,m}\big|\cup_{j\in\N}\{a_{j}\}\big)-(1-\gamma)f_{t}\big(a_{i}\big|\cup_{j\neq i,j\in\N}\{a_{j}\}\big)\Big),
	\end{aligned}
\end{equation} where the first equality follows from part $\textbf{1)}$ of \cref{thm1} and the second inequality comes from the Eq.\eqref{equ_appendix2_3}. Then, from the $\gamma$-weakly submodularity,  
\begin{equation}\label{equ_appendix2_5}
	\begin{aligned}
		&\sum_{(i,m): v_{i,m}\in S}\E_{a_{j}\sim\uppi_{j},\forall j\in\N}\Big( f_{t}\big(v_{i,m}\big|\cup_{j\in\N}\{a_{j}\}\big)\Big)\\
		&\ge\gamma\E_{a_{j}\sim\uppi_{j},\forall j\in\N}\Big( \gamma f_{t}\big(S\big|\cup_{j\in\N}\{a_{j}\}\big)\Big)\\
		&\ge\gamma\Big(f_t(S)-F_{t}(\uppi_{1},\dots,\uppi_{n})\Big),
	\end{aligned}
\end{equation} where the final inequality follows from the monotonicity.

Furthermore,  from the $\beta$-weakly  upper submodularity,we also have, 
\begin{equation}\label{equ_appendix2_6}
	\begin{aligned}
		&\sum_{(i,m): v_{i,m}\in S}\E_{a_{j}\sim\uppi_{j},\forall j\in\N}\Big(f_{t}\big(a_{i}\big|\cup_{j\neq i,j\in\N}\{a_{j}\}\big)\Big)\\
		&\le\sum_{i\in\N}\E_{a_{j}\sim\uppi_{j},\forall j\in\N}\Big(f_{t}\big(a_{i}\big|\cup_{j\neq i,j\in\N}\{a_{j}\}\big)\Big)\\
		&\le\beta \E_{a_{j}\sim\uppi_{j},\forall j\in\N}\Big(f_{t}\big(\cup_{j\in\N}\{a_{j}\}\big)-f_{t}(\emptyset)\Big)\\
		&=\beta F_{t}(\uppi_{1},\dots,\uppi_{n}),
	\end{aligned}
\end{equation}where the second inequality follows from the definition of  $\beta$-upper submodularity.

Merging Eq.\eqref{equ_appendix2_6} and Eq.\eqref{equ_appendix2_5} into Eq.\eqref{equ_appendix2_4}, we have that
\begin{equation*}
	\begin{aligned}
		&\sum_{(i,m): v_{i,m}\in S}\frac{\partial F_{t}}{\partial \pi_{i,m}}(\uppi_{1},\dots,\uppi_{n})\\
		&\ge\sum_{(i,m): v_{i,m}\in S}\E_{a_{j}\sim\uppi_{j},\forall j\in\N}\Big( \gamma f_{t}\big(v_{i,m}\big|\cup_{j\in\N}\{a_{j}\}\big)-(1-\gamma)f_{t}\big(a_{i}\big|\cup_{j\neq i,j\in\N}\{a_{j}\}\big)\Big)\\
		&\ge(\gamma^{2}\Big(f_t(S)-F_{t}(\uppi_{1},\dots,\uppi_{n})\Big)-(1-\gamma)\beta F_{t}(\uppi_{1},\dots,\uppi_{n}))\\
		&=\gamma^{2}f_t(S)-\left((1-\gamma)\beta+\gamma^{2}\right)F_{t}(\uppi_{1},\dots,\uppi_{n})).
	\end{aligned}
\end{equation*}

As for $f_{t}$ is monotone $\alpha$-weakly DR-submodular, we also can show that 
\begin{equation*}
	\begin{aligned}
		\sum_{(i,m): v_{i,m}\in S}\frac{\partial F_{t}}{\partial \pi_{i,m}}(\uppi_{1},\dots,\uppi_{n})&=\E_{a_{j}\sim\uppi_{j},\forall j\in\N}\Big(\sum_{(i,m): v_{i,m}\in S} f_{t}\big(v_{i,m}\big|\cup_{j\neq i,j\in\N}\{a_{j}\}\big)\Big)\\
		&\ge \alpha\E_{a_{j}\sim\uppi_{j},\forall j\in\N}\Big(f_{t}\left(S| \cup_{j\in\N}\{a_{j}\}\right)\Big)\\
		&\ge\alpha\left(f_{t}(S)-F_{t}(\uppi_{1},\dots,\uppi_{n})\right),
	\end{aligned} 
\end{equation*} where the first inequality follows from the $\alpha$-DR submodularity and the final inequality comes from the monotonicity of $f_{t}$.

\section{Proof of \texorpdfstring{\cref{thm2}}{}}\label{appendix_proof_thm0}
In this section, we verify the \cref{thm2}.

Before that, we firstly suppose that the symbol $\one_{S}$ represents the indicator function over the subset $S\subseteq\V$, namely, for any action $v_{i,m}\in S$, the vector $\one_{S}$ sets the corresponding probability for this action $v_{i,m}$ as $1$. Note that, when $S$ satisfies the constraints of problem \eqref{equ_problem}, namely, $|S\cap\V_{i}|\le1$, we can infer that $\one_{S}\in\prod_{i=1}^{n}\Delta_{\kappa_{i}}$.

\subsection{Proof of Part \textbf{2)} and Part \textbf{3)} in \texorpdfstring{\cref{thm2}}{}}
With the previously defined symbol $\one_{S}$, we can rewrite  the $\textbf{4)}$ of \cref{thm1} as:

 \textbf{i)}: when $f_{t}$ is monotone $\alpha$-weakly DR-submodular, the following inequality holds:
\begin{equation*}
	\alpha\Big(f_{t}(S)-F_{t}(\uppi_{1},\dots,\uppi_{n})\Big)\le\langle\one_{S},\nabla F_{t}(\uppi_{1},\dots,\uppi_{n})\rangle;
\end{equation*} 

\textbf{ii):} when $f_{t}$ is monotone $(\gamma,\beta)$-weakly submodular, we can show that
\begin{equation*}
	\Big(\gamma^{2}f_{t}(S)-(\beta(1-\gamma)+\gamma^{2})F_{t}(\uppi_{1},\dots,\uppi_{n})\Big)\le\langle\one_{S},\nabla F_{t}(\uppi_{1},\dots,\uppi_{n})\rangle.
\end{equation*}

In order to prove \cref{thm2}, we next show the relationship between $\langle(\uppi_{1},\dots,\uppi_{n}),\nabla F_{t}(\uppi_{1},\dots,\uppi_{n})\rangle$ and  $F_{t}(\uppi_{1},\dots,\uppi_{n})$.

From Eq.\eqref{Appendix_A.2}, if $\uppi_{i}=(\pi_{i,1},\dots,\pi_{i,\kappa_{i}}),\forall i\in\N$, we can show
\begin{equation}\label{Appendix_B.2}
	\begin{aligned}
		&\langle(\uppi_{1},\dots,\uppi_{n}),\nabla F_{t}(\uppi_{1},\dots,\uppi_{n})\rangle\\
		&=\sum_{i=1}^{n}\sum_{m=1}^{\kappa_{i}}\pi_{i,m}\frac{\partial F_{t}}{\partial \pi_{i,m}} (\uppi_{1},\dots,\uppi_{n})\\
		&=\sum_{i=1}^{n}\sum_{m=1}^{\kappa_{i}}\pi_{i,m}\E_{a_{j}\sim\uppi_{j},\forall j\in\N}\Big(f_{t}\big(v_{i,m}\big|\cup_{j\neq i,j\in\N}\{a_{j}\}\big)\Big)\\
		&=\sum_{i=1}^{n}\sum_{m=1}^{\kappa_{i}}\left(\pi_{i,m}\sum_{a_{j}\in\V_{j}\cup\{\emptyset\},\forall j\neq i}\left(f_{t}\big(v_{i,m}\big|\cup_{j\neq i,j\in\N}\{a_{j}\}\big)\prod_{j\neq i,j\in\N}p(a_{j}|\uppi_{j})\right)\right)\\
		&=\sum_{i=1}^{n}\sum_{m=1}^{\kappa_{i}}\left(p(v_{i,m}|\uppi_{i})\sum_{a_{j}\in\V_{j}\cup\{\emptyset\},\forall j\neq i}\left(f_{t}\big(v_{i,m}\big|\cup_{j\neq i,j\in\N}\{a_{j}\}\big)\prod_{j\neq i,j\in\N}p(a_{j}|\uppi_{j})\right)\right)\\
		&=\sum_{i=1}^{n}\sum_{m=1}^{\kappa_{i}}\left(\sum_{a_{j}\in\V_{j}\cup\{\emptyset\},\forall j\neq i}\left(f_{t}\big(v_{i,m}\big|\cup_{j\neq i,j\in\N}\{a_{j}\}\big)p(v_{i,m}|\uppi_{i})*\prod_{j\neq i,j\in\N}p(a_{j}|\uppi_{j})\right)\right)\\
		&=\sum_{i=1}^{n}\sum_{a_{i}\in\V_{i}\cup\{\emptyset\}}\left(\sum_{a_{j}\in\V_{j}\cup\{\emptyset\},\forall j\neq i}\left(f_{t}\big(a_{i}\big|\cup_{j\neq i,j\in\N}\{a_{j}\}\big)\prod_{j\in\N}p(a_{j}|\uppi_{j})\right)\right)\\
		&=\sum_{i=1}^{n}\E_{a_{j}\sim\uppi_{j},\forall j\in\N}\Big(f_{t}\big(a_{i}\big|\cup_{j\neq i,j\in\N}\{a_{j}\}\big)\Big),
	\end{aligned}
\end{equation} where the second equality follows from the $\textbf{1)}$ of \cref{thm1}, the third equality comes from the final equality of the Eq.\eqref{Appendix_A.2}, the fourth equality follows from  $p(v_{i,m}|\uppi_{i})=\pi_{i,m}$ and the final equality from the fact that the random element $a_{i}$ is drawn from the policy $\uppi_{i}$.

Therefore, if the original set function $f_{t}$ is $\beta$-weakly submodular from above, we can show that

\begin{equation}\label{Appendix_B.1}
	\begin{aligned}
		&\langle(\uppi_{1},\dots,\uppi_{n}),\nabla F_{t}(\uppi_{1},\dots,\uppi_{n})\rangle\\
		&=\sum_{i=1}^{n}\E_{a_{j}\sim\uppi_{j},\forall j\in\N}\Big(f_{t}\big(a_{i}\big|\cup_{j\neq i,j\in\N}\{a_{j}\}\big)\Big)\\
		&=\E_{a_{j}\sim\uppi_{j},\forall j\in\N}\Big(\sum_{i=1}^{n}f_{t}\big(a_{i}\big|\cup_{j\neq i,j\in\N}\{a_{j}\}\big)\Big)\\
		&\le\beta\E_{a_{j}\sim\uppi_{j},\forall j\in\N}\Big(f_{t}(\cup_{j\in\N}\{a_{j}\})-f_{t}(\emptyset)\Big)\\
		&=\beta F_{t}(\uppi_{1},\dots,\uppi_{n}).
	\end{aligned}
\end{equation}

Therefore, by merging Eq.\eqref{Appendix_B.1} and the \textbf{4)} of \cref{thm1}, when the set function $f_{t}$ is monotone and $(\gamma,\beta)$-weakly submodular, for any subset $S$ within the constraint of problem~\eqref{equ_problem} and $(\uppi_{1},\dots,\uppi_{n})\in\prod_{i=1}^{n}\Delta_{\kappa_{i}}$, we have that
\begin{equation*}
	\left\langle\one_{S}-(\uppi_{1},\dots,\uppi_{n}),\nabla F_{t}(\uppi_{1},\dots,\uppi_{n})\right\rangle\ge\gamma^{2} f_{t}(S)-(\beta+\beta(1-\gamma)+\gamma^{2})F_{t}(\uppi_{1},\dots,\uppi_{n}).
\end{equation*}

Similarly, from the definition of $\alpha$-weakly DR-submodular function, we know that $\alpha$-weakly DR-submodular function automatically satisfies the conditions for being $\frac{1}{\alpha}$-weakly submodular from above, we also have that, when $f_{t}$ is monotone $\alpha$-weakly DR-submodular, 
\begin{equation*}
		\left\langle\one_{S}-(\uppi_{1},\dots,\uppi_{n}),\nabla F_{t}(\uppi_{1},\dots,\uppi_{n})\right\rangle\ge\alpha f_{t}(S)-(\alpha+\frac{1}{\alpha})F_{t}(\uppi_{1},\dots,\uppi_{n}).
\end{equation*} 

Therefore, when $(\uppi^{s}_{1},\dots,\uppi^{s}_{n})$ is a stationary point  of $F_{t}$ over the domain $\prod_{i=1}^{n}\Delta_{\kappa_{i}}$, we have

\textbf{i):} $f_{t}$ is monotone $(\gamma,\beta)$-weakly submodular, for any $S$ within the constraint of problem~\eqref{equ_problem},
\begin{equation*}
\gamma^{2} f_{t}(S)-(\beta+\beta(1-\gamma)+\gamma^{2})F_{t}(\uppi^{s}_{1},\dots,\uppi^{s}_{n})\le	\left\langle\one_{S}-(\uppi^{s}_{1},\dots,\uppi^{s}_{n}),\nabla F_{t}(\uppi^{s}_{1},\dots,\uppi^{s}_{n})\right\rangle\le0.
\end{equation*}
In other words,
\begin{equation*}
F_{t}(\uppi^{s}_{1},\dots,\uppi^{s}_{n})\ge\frac{\gamma^{2}}{\beta+\beta(1-\gamma)+\gamma^{2}} f_{t}(S^{*}),
\end{equation*} where $S^{*}$ is the optimal subset of problem~\eqref{equ_problem}.

We get the \textbf{3)} in \cref{thm2}.

Similarly, \textbf{ii):} when $f_{t}$ is monotone $\alpha$-weakly DR-submodular, for any $S$ within the constraint of problem~\eqref{equ_problem},
\begin{equation*}
	\alpha f_{t}(S)-(\alpha+\frac{1}{\alpha})F_{t}(\uppi^{s}_{1},\dots,\uppi^{s}_{n})\le	\left\langle\one_{S}-(\uppi^{s}_{1},\dots,\uppi^{s}_{n}),\nabla F_{t}(\uppi^{s}_{1},\dots,\uppi^{s}_{n})\right\rangle\le0.
\end{equation*}
In other words,
\begin{equation*}
	F_{t}(\uppi^{s}_{1},\dots,\uppi^{s}_{n})\ge\frac{\alpha^{2}}{1+\alpha} f_{t}(S^{*}).
\end{equation*} where $S^{*}$ is the optimal subset of problem~\eqref{equ_problem}.

We get the \textbf{2)} in \cref{thm2}.
\subsection{Proof of Part \textbf{1)}  in \texorpdfstring{\cref{thm2}}{}}
In this subsection, we prove the part \textbf{1)} in \cref{thm2}.

From Eq.\eqref{equ_appendix2_4}, we have

\begin{equation*}
	\langle\one_{S},\nabla F_{t}(\uppi_{1},\dots,\uppi_{n})\rangle=\sum_{(i,m): v_{i,m}\in S}\E_{a_{j}\sim\uppi_{j},\forall j\in\N}\Big( f_{t}\big(v_{i,m}\big|\cup_{j\neq i,j\in\N}\{a_{j}\}\big)\Big).
\end{equation*}

From Eq.\eqref{Appendix_B.2}, we have 

\begin{equation*}
\langle(\uppi_{1},\dots,\uppi_{n}),\nabla F_{t}(\uppi_{1},\dots,\uppi_{n})\rangle=\sum_{i=1}^{n}\E_{a_{j}\sim\uppi_{j},\forall j\in\N}\Big(f_{t}\big(a_{i}\big|\cup_{j\neq i,j\in\N}\{a_{j}\}\big)\Big).
\end{equation*}

Therefore, we have that
\begin{equation}\label{Appendix.B.2.0}
\begin{aligned}
	&\langle\one_{S}-(\uppi_{1},\dots,\uppi_{n}),\nabla F_{t}(\uppi_{1},\dots,\uppi_{n})\rangle\\
	&=\sum_{(i,m): v_{i,m}\in S}\E_{a_{j}\sim\uppi_{j},\forall j\in\N}\Big( f_{t}\big(v_{i,m}\big|\cup_{j\neq i,j\in\N}\{a_{j}\}\big)\Big)-\sum_{i=1}^{n}\E_{a_{j}\sim\uppi_{j},\forall j\in\N}\Big(f_{t}\big(a_{i}\big|\cup_{j\neq i,j\in\N}\{a_{j}\}\big)\Big)\\
	&=\E_{a_{j}\sim\uppi_{j},\forall j\in\N}\Big(\sum_{(i,m): v_{i,m}\in S}f_{t}\big(v_{i,m}\big|\cup_{j\neq i,j\in\N}\{a_{j}\}\big)-\sum_{i=1}^{n}f_{t}\big(a_{i}\big|\cup_{j\neq i,j\in\N}\{a_{j}\}\big)\Big)\\
	&=\E_{a_{j}\sim\uppi_{j},\forall j\in\N}\Big(\sum_{v_{i,m}\in\left( S\setminus\cup_{j\in\N}\{a_{j}\}\right)}f_{t}\big(v_{i,m}\big|\cup_{j\neq i,j\in\N}\{a_{j}\}\big)-\sum_{a_{i}\in\left( \cup_{j\in\N}\{a_{j}\}\setminus S\right)}f_{t}\big(a_{i}\big|\cup_{j\neq i,j\in\N}\{a_{j}\}\big)\Big).
\end{aligned}
\end{equation}

When $f_{t}$ is monotone submodular, we have 
\begin{equation}\label{Appendix.B.2.1}
\sum_{v_{i,m}\in\left( S\setminus\cup_{j\in\N}\{a_{j}\}\right)}f_{t}\big(v_{i,m}\big|\cup_{j\neq i,j\in\N}\{a_{j}\}\big)\ge f_{t}\big(\cup_{j\in\N}\{a_{j}\})\cup S\big)-f_{t}(\cup_{j\in\N}\{a_{j}\}).
\end{equation}

Furthermore, from the definition of curvature $c$, we also have
\begin{equation}\label{Appendix.B.2.2}
\begin{aligned}
f_{t}\big(\cup_{j\in\N}\{a_{j}\})\cup S\big)-f_{t}(S)\ge(1-c)\Big(f_{t}(\cup_{j\in\N}\{a_{j}\})-f_{t}(\cup_{j\in\N}\{a_{j}\}\cap S)\Big),
\end{aligned}
\end{equation} where this inequality follows from $\big(\cup_{j\in\N}\{a_{j}\})\cup S\big)\setminus S=\big(\cup_{j\in\N}\{a_{j}\})\big)\setminus\big(\cup_{j\in\N}\{a_{j}\}\cap S\big)$.

Similarly, from the submodularity, we have 
\begin{equation}\label{Appendix.B.2.3}
\sum_{a_{i}\in\left(\cup_{j\in\N}\{a_{j}\}\setminus S\right)}f_{t}\big(a_{i}\big|\cup_{j\neq i,j\in\N}\{a_{j}\}\big)\le\Big(f_{t}(\cup_{j\in\N}\{a_{j}\})-f_{t}(\cup_{j\in\N}\{a_{j}\}\cap S)\Big),
\end{equation} where the inequality comes from that $\big(\cup_{j\in\N}\{a_{j}\}\cap S\big)\subseteq\big(\cup_{j\neq i,j\in\N}\{a_{j}\}\big)$ if $a_{i}\in \cup_{j\in\N}\{a_{j}\}\setminus S$.

Merging Eq.\eqref{Appendix.B.2.1},Eq.\eqref{Appendix.B.2.2} and Eq.\eqref{Appendix.B.2.3} into Eq.\eqref{Appendix.B.2.0}, we can show that
\begin{equation*}
\begin{aligned}
	&\langle\one_{S}-(\uppi_{1},\dots,\uppi_{n}),\nabla F_{t}(\uppi_{1},\dots,\uppi_{n})\rangle\\
	&=\E_{a_{j}\sim\uppi_{j},\forall j\in\N}\Big(\sum_{v_{i,m}\in S\setminus\cup_{j\in\N}\{a_{j}\}}f_{t}\big(v_{i,m}\big|\cup_{j\neq i,j\in\N}\{a_{j}\}\big)-\sum_{a_{i}\in \cup_{j\in\N}\{a_{j}\}\setminus S}f_{t}\big(a_{i}\big|\cup_{j\neq i,j\in\N}\{a_{j}\}\big)\Big)\\
	&\ge\E_{a_{j}\sim\uppi_{j},\forall j\in\N}\Bigg( f_{t}(S)-f_{t}(\cup_{j\in\N}\{a_{j}\})-c\Big(f_{t}(\cup_{j\in\N}\{a_{j}\})-f_{t}(\cup_{j\in\N}\{a_{j}\}\cap S)\Big)\Bigg)\\
&\ge\E_{a_{j}\sim\uppi_{j},\forall j\in\N}\Big( f_{t}(S)-(1+c)f_{t}(\cup_{j\in\N}\{a_{j}\})\Big)\\
&=f_{t}(S)-(1+c)F_{t}(\uppi_{1},\dots,\uppi_{n}).
\end{aligned}
\end{equation*}

Thus, if $(\uppi^{s}_{1},\dots,\uppi^{s}_{n})$ is a stationary point  of $F_{t}$ over the domain $\prod_{i=1}^{n}\Delta_{\kappa_{i}}$, when $f_{t}$ is monotone submodular with curvature $c$, we have, for any $S$ in the constraint of problem \eqref{equ_problem}, we have
\begin{equation*}
	f_{t}(S)-(1+c)F_{t}(\uppi^{s}_{1},\dots,\uppi^{s}_{n})\le 	\langle\one_{S}-(\uppi^{s}_{1},\dots,\uppi^{s}_{n}),\nabla F_{t}(\uppi^{s}_{1},\dots,\uppi^{s}_{n})\rangle\le0.
\end{equation*}
So, \begin{equation*}
F_{t}(\uppi^{s}_{1},\dots,\uppi^{s}_{n})\ge\frac{1}{1+c}f_{t}(S^{*}),
\end{equation*} where $S^{*}$ is the optimal subset of problem~\eqref{equ_problem}.
\subsection{A Policy-based Continuous Extension with \texorpdfstring{$1/2$-Approximation Stationary Point}{}}\label{appendix:special_case}
In this subsection  we consider a special set function. 

At first, let  the universe set $U$ consist of $n-1$ elements $\{x_{1},\dots,x_{n-1}\}$ and $n-k$ elements $\{y_{1},\dots,y_{n-k}\}$, all of weight $1$, and $n-1$ elements $\{\epsilon_{1},\dots,\epsilon_{n-1}\}$ of arbitrarily
small weight $\epsilon>0$. Then, we define two different types of sets namely, $A_{i}$ and $A_{i+n}$ for any $i\in[n]$, that is to say ,
\begin{equation*}
	\begin{aligned}
		&A_{i}\triangleq\{\epsilon_{i}\}\text{\ for $1\le i\le n-1$,}\ \ \ \ \ \ \ \ \ \ \ \ \ \ \ \ \ \ \ A_{n}\triangleq\{x_{1},\dots,x_{n-1}\},\\
	&A_{n+i}\triangleq\{x_{i}\}\text{\ for $1\le i\le n-1$,}\ \ \ \ \ \ \ \ \ \ \ \ \ \ A_{2n}\triangleq\{y_{1},\dots,y_{n-k}\}.\\
	\end{aligned}
\end{equation*}
After that, we define a coverage set function $f:2^{\V}\rightarrow\R_{+}$ over these $2n$ distinct set $\{A_{1},\dots,A_{n},A_{n+1},\dots,A_{2n}\}$ where $\V=[2n]$. Specifically,for any subset $S\subseteq\V$,
\begin{equation}\label{equ:cover_set}
	f(S)\triangleq\sum_{v\in\bigcup_{i\in S}A_{i}}w(v),
\end{equation} where $w(v)$ is the weight of element $v$.

Moreover, we consider a partition constraint that contains at most one of $\{A_{i},A_{n+i}\}$ for any $i\in[n]$. If we set $\V_{i}\triangleq\{i,i+n\}$ and $\V\triangleq\bigcup_{i\in[n]}\V_{i}\triangleq[2n]$, we naturally obtain the following coverage maximization problem:
\begin{equation}\label{equ:coverage}
	\max_{S\subseteq\V} f(S)\ \ \text{s.t.}\ |S\cap\V_{i}|\le 1\ \ \forall i\in[n].	
\end{equation}
Note that this problem~\eqref{equ:coverage} is a special case of the concern problem~\eqref{equ_problem}. From the result of \citep{filmus2012power}, we know that the coverage function $f$ in \cref{equ:cover_set} is a submodular set function, namely, $\alpha=\beta=\gamma=1$.

A key feature of the coverage maximization problem~\eqref{equ:coverage} is that \citep{filmus2012power} found that the standard greedy~\citep{nemhauser1978analysis} will be stuck at a local maximum subset $\{A_{1},\dots,A_{n}\}$ where $S=[n]$ and $f(S)=(1+\epsilon)n$. In contrast, when $\epsilon$ is very small,  the optimal subset for the  problem~\eqref{equ:coverage}  is $\{A_{n+1},\dots,A_{2n}\}$ where $S=\{n+1,\dots,2n\}$ and $f(S)=2n-k-1$. Note that $\lim_{n\rightarrow\infty}\lim_{\epsilon\rightarrow0\text{\ and\ }k\rightarrow0}\frac{(1+\epsilon)n}{2n-k-1}=\frac{1}{2}$.
Motivated by this finding of \citep{filmus2012power}, we also can show that the 
point $\one_{[n]}\triangleq(\underbrace{1,\dots,1}_{n},0,\dots,0)$ is a local stationary point of the policy-based continuous extension of the set function $f$ in \cref{equ:cover_set}. More specifically, we have the following theorem:
\begin{theorem}
The point $\one_{[n]}$ is a stationary point of the policy-based continuous extension $F$ of the set function $f$ in \cref{equ:cover_set}. Moreover, we can show $\frac{F(\one_{[n]})}{f(\{n+1,\dots,2n\})}=\frac{(1+\epsilon)n}{2n-k-1}\rightarrow \frac{1}{2}$.
\end{theorem}
\begin{remark}
This theorem indicates that when the objective set function is submodular, namely, $c=\alpha=\gamma=\beta=1$, the approximation guarantees established in Theorem~\ref{thm2} is \textbf{tight}.
\end{remark}
\begin{proof}
At first, for any $i\in[n]$, we assume $\uppi_{i}=(\pi_{i,1},\pi_{i,2})$. Then, according to \cref{def_extension}, we have that the policy-based continuous extension $F$ of the set function $f$ in \cref{equ:cover_set} can be formulated as:
	\begin{equation}
		F(\uppi_{1},\dots,\uppi_{n})\triangleq\sum_{i=1}^{n}\sum_{a_{i}\in\V_{i}\cup\{\emptyset\}}\Big(f\big(\cup_{i=1}^{n}\{a_{i}\}\big)\prod_{i=1}^{n}p(a_{i}|\uppi_{i})\Big),
	\end{equation} where $p(i| \uppi_i) = \pi_{i,1}$, $p(n+i| \uppi_i) = \pi_{i,2}$ and $p(\emptyset | \uppi_{i}) = 1-\pi_{i,1}-\pi_{i,2}$.
	
	From part \textbf{1)} of \cref{thm1}, we also can show that,
	\begin{equation*}
		\begin{aligned}
			&\frac{\partial F}{\partial\pi_{i,1}}(\one_{[n]})=\epsilon\text{\ for $1\le i\le n-1$,}\ \ \ \ \ \ \ \ \ \ \ \ \ \ \ \ \ \ \ \frac{\partial F}{\partial\pi_{n,1}}(\one_{[n]})=n-1,\\
			&\frac{\partial F}{\partial\pi_{i,2}}(\one_{[n]})=0\text{\ for $1\le i\le n-1$,}\ \ \ \ \ \ \ \ \ \ \ \ \ \ \ \ \ \  \frac{\partial F}{\partial\pi_{n,2}}(\one_{[n]})=n-k.\\
		\end{aligned}
	\end{equation*}
	As a result, for any $(\uppi_{1},\dots,\uppi_{n})\in\prod_{i=1}^{n}\Delta_{2}$, we have 
	\begin{equation*}
		\begin{aligned}
			&\left\langle(\uppi_{1},\dots,\uppi_{n})-\one_{[n]},\nabla F(\one_{[n]})\right\rangle\\
			&=\sum_{i=1}^{n}\Big(\pi_{i,1}\frac{\partial F}{\partial\pi_{i,1}}(\one_{[n]})+\pi_{i,2}\frac{\partial F}{\partial\pi_{i,2}}(\one_{[n]})\Big)-\sum_{i=1}^{n}\frac{\partial F}{\partial\pi_{i,1}}(\one_{[n]})\\
			&=\sum_{i=1}^{n}\Big(\pi_{i,1}\frac{\partial F}{\partial\pi_{i,1}}(\one_{[n]})+\pi_{i,2}\frac{\partial F}{\partial\pi_{i,2}}(\one_{[n]})\Big)-(1+\epsilon)(n-1)\\
			&=\epsilon\sum_{i=1}^{n-1}\pi_{i,1}+(n-1)\pi_{n,1}+(n-k)\pi_{n,2}-(1+\epsilon)(n-1)\\
			&=\epsilon\sum_{i=1}^{n-1}\pi_{i,1}+(n-1)(\pi_{n,1}+\pi_{n,2})+\big((n-k)-(n-1)\big)\pi_{n,2}-(1+\epsilon)(n-1)\\
			&=\epsilon\sum_{i=1}^{n-1}(\pi_{i,1}-1)+\big((n-k)-(n-1)\big)\pi_{i,2}-(n-1)(1-\pi_{n,1}-\pi_{n,2})\\
		&=\epsilon\sum_{i=1}^{n-1}(\pi_{i,1}-1)+(1-k)\pi_{i,2}-(n-1)(1-\pi_{n,1}-\pi_{n,2})\le 0,
		\end{aligned}
	\end{equation*} where the final inequality follows from $1-k\le0$ and $1-\pi_{n,1}-\pi_{n,2}\ge0$.
	As a result, from \cref{def:stationary}, we can know that
	the point $\one_{[n]}$ is a stationary point of the policy-based continuous extension $F$. Note that $\F(\one_{[n]})=f([n])=(1+\epsilon)(n-1)$ such that $\frac{F(\one_{[n]})}{f(\{n+1,\dots,2n\})}=\frac{(1+\epsilon)n}{2n-k-1}$. Particularly when $k\rightarrow0$, $\epsilon\rightarrow 0 $ and $n\rightarrow\infty$, $\frac{F(\one_{[n]})}{f(\{n+1,\dots,2n\})}=\frac{(1+\epsilon)n}{2n-k-1}\rightarrow\frac{1}{2}$.
\end{proof}

\section{Proof of \texorpdfstring{\cref{thm4}}{}}\label{appendix_proof_thm0.5}
In this section, we cut the proof of \cref{thm4} in two distinct subsections. Specifically, \cref{Sec:C1} provides the proof for parts \textbf{1)} and \textbf{2)}, while  \cref{Sec:C2} addresses the part \textbf{3)}.
\subsection{Proof of Part \textbf{1)} and Part \textbf{2)}  in \texorpdfstring{\cref{thm4}}{}}\label{Sec:C1}
Before verifying the parts \textbf{1)} and \textbf{2)}, we firstly prove the following theorem:
\begin{theorem}\label{thm:C1}
	Given a monotone set function $f_{t}$, for its policy-based continuous extension $F_{t}$ introduced in \cref{def_extension}, if we consider a surrogate function  $F_{t}^{s}:\prod_{i=1}^{n}\Delta_{\kappa_{i}}\rightarrow\R_{+}$ whose gradient at each point $\x\in\prod_{i=1}^{n}\Delta_{\kappa_{i}}$ is a weighted average of the gradient $\nabla F_{t}(z*\x)$, namely, $\nabla F_{t}^{s}(\x)\triangleq\int_{0}^{1} w(z)\nabla F_{t}(z*\boldsymbol{x})\mathrm{d}z$ where $w(z)$ is a positive weight function over $[0,1]$, we have that:
	
	\textbf{i):} When$f_{t}$ is $\alpha$-weakly DR-submodular and $w(z)=e^{\alpha(z-1)}$, for any policy vector $(\uppi_{1},\dots,\uppi_{n})\in\prod_{i=1}^{n}\Delta_{\kappa_{i}}$ and any subset $S$ within the constraint of problem~\eqref{equ_problem}, then the following inequality holds:
	\begin{equation}\label{equ:boosting1}
		\begin{aligned}
		&\Big\langle\one_{S}-(\uppi_{1},\dots,\uppi_{n}),\nabla F_{t}^{s}(\uppi_{1},\dots,\uppi_{n})\Big\rangle\\&
		=\left\langle\one_{S}-(\uppi_{1},\dots,\uppi_{n}),\int_{0}^{1}e^{\alpha(z-1)}\nabla F_{t}\left(z*(\uppi_{1},\dots,\uppi_{n})\right)\mathrm{d}z\right\rangle\\&
		\ge\left(1-e^{-\alpha}\right)f_{t}(S)-F_{t}(\uppi_{1},\dots,\uppi_{n});
		\end{aligned}
	\end{equation}
	
\textbf{ii):} When  $f_{t}$ is  $(\gamma,\beta)$-weakly submodular and $w(z)=e^{\phi(\gamma,\beta)(z-1)}$ where $\phi(\gamma,\beta)=\beta(1-\gamma)+\gamma^2$, for any policy vector $(\uppi_{1},\dots,\uppi_{n})\in\prod_{i=1}^{n}\Delta_{\kappa_{i}}$ and any subset $S$ within the constraint of problem~\eqref{equ_problem},  then the following inequality holds:
\begin{equation}\label{equ:boosting2}
	\begin{aligned}
		&\Big\langle\one_{S}-(\uppi_{1},\dots,\uppi_{n}),\nabla F_{t}^{s}(\uppi_{1},\dots,\uppi_{n})\Big\rangle\\&
		=\left\langle\one_{S}-(\uppi_{1},\dots,\uppi_{n}),\int_{0}^{1}e^{\phi(\gamma,\beta)(z-1)}\nabla F_{t}\left(z*(\uppi_{1},\dots,\uppi_{n})\right)\mathrm{d}z\right\rangle\\&
		\ge\left(\frac{\gamma^{2}(1-e^{-\phi(\gamma,\beta)})}{\phi(\gamma,\beta)}\right)f_{t}(S)-F_{t}(\uppi_{1},\dots,\uppi_{n}).
	\end{aligned}
\end{equation}
\end{theorem}
\begin{proof}
For \textbf{i)}: At first, we verify the following relationship:
\begin{equation*}
	\begin{aligned}
		&\left\langle (\uppi_{1},\dots,\uppi_{n}),\int_{0}^{1}e^{\alpha(z-1)}\nabla F_{t}\left(z*(\uppi_{1},\dots,\uppi_{n})\right)\mathrm{d}z\right\rangle\\
		&=\int_{0}^{1}e^{\alpha(z-1)}\Big\langle (\uppi_{1},\dots,\uppi_{n}),\nabla F_{t}\left(z*(\uppi_{1},\dots,\uppi_{n})\right)\Big\rangle\mathrm{d}z\\
		&=\int_{0}^{1}e^{\alpha(z-1)}\frac{\mathrm{d}F_{t}\left(z*(\uppi_{1},\dots,\uppi_{n})\right)}{\mathrm{d}z}\mathrm{d}z\\
		&=\int_{0}^{1}e^{\alpha(z-1)}\mathrm{d}F_{t}\left(z*(\uppi_{1},\dots,\uppi_{n})\right)\\
		&=e^{\alpha(z-1)}F_{t}\left(z*(\uppi_{1},\dots,\uppi_{n})\right)|_{z=0}^{z=1}-\int_{0}^{1}F_{t}\left(z*(\uppi_{1},\dots,\uppi_{n})\right)\mathrm{d}\left(e^{\alpha(z-1)}\right)\\
		&=F_{t}(\uppi_{1},\dots,\uppi_{n})-\alpha\int_{0}^{1}e^{\alpha(z-1)}F_{t}\left(z*(\uppi_{1},\dots,\uppi_{n})\right)\mathrm{d}z,
	\end{aligned}
\end{equation*} where the final equality follows from $F_{t}(\mathbf{0})=0$.

Then, we have the following inequality:
\begin{equation*}
\begin{aligned}
	&\left\langle \one_{S},\int_{0}^{1}e^{\alpha(z-1)}\nabla F_{t}\left(z*(\uppi_{1},\dots,\uppi_{n})\right)\mathrm{d}z\right\rangle\\
	&=\int_{0}^{1}e^{\alpha(z-1)}\Big\langle \one_{S},\nabla F_{t}\left(z*(\uppi_{1},\dots,\uppi_{n})\right)\Big\rangle\mathrm{d}z\\
	&\ge\alpha\int_{0}^{1}e^{\alpha(z-1)}\Big(f_{t}(S)-F_{t}\big(z*(\uppi_{1},\dots,\uppi_{n})\big)\Big)\mathrm{d}z\\
	&=\Big(\alpha\int_{0}^{1}e^{\alpha(z-1)}\mathrm{d}z\Big)f_{t}(S)-\alpha\int_{0}^{1}e^{\alpha(z-1)}F_{t}\left(z*(\uppi_{1},\dots,\uppi_{n})\right)\mathrm{d}z\\
	&=\Big(1-e^{-\alpha}\Big)f_{t}(S)-\alpha\int_{0}^{1}e^{\alpha(z-1)}F_{t}\left(z*(\uppi_{1},\dots,\uppi_{n})\right)\mathrm{d}z,
\end{aligned}
\end{equation*}where the first inequality follows from the part \textbf{4)} of \cref{thm1}.

As a result, we have 
	\begin{equation*}
		\left\langle\one_{S}-(\uppi_{1},\dots,\uppi_{n}),\int_{0}^{1}e^{\alpha(z-1)}\nabla F_{t}\left(z*(\uppi_{1},\dots,\uppi_{n})\right)\mathrm{d}z\right\rangle
		\ge\left(1-e^{-\alpha}\right)f_{t}(S)-F_{t}(\uppi_{1},\dots,\uppi_{n}).
\end{equation*}

As for \textbf{ii):} Similarly, when  $f_{t}$ is  monotone $(\gamma,\beta)$-weakly submodular and $w(z)=e^{\phi(\gamma,\beta)(z-1)}$ where $\phi(\gamma,\beta)=\beta(1-\gamma)+\gamma^2$, we have 

\begin{equation*}
	\begin{aligned}
		&\left\langle (\uppi_{1},\dots,\uppi_{n}),\int_{0}^{1}e^{\phi(\gamma,\beta)(z-1)}\nabla F_{t}\left(z*(\uppi_{1},\dots,\uppi_{n})\right)\mathrm{d}z\right\rangle\\
		&=\int_{0}^{1}e^{\phi(\gamma,\beta)(z-1)}\mathrm{d}F_{t}\left(z*(\uppi_{1},\dots,\uppi_{n})\right)\\
		&=e^{\phi(\gamma,\beta)(z-1)}F_{t}\left(z*(\uppi_{1},\dots,\uppi_{n})\right)|_{z=0}^{z=1}-\phi(\gamma,\beta)\int_{0}^{1}e^{\phi(\gamma,\beta)(z-1)}F_{t}\left(z*(\uppi_{1},\dots,\uppi_{n})\right)\mathrm{d}z\\
		&=F_{t}(\uppi_{1},\dots,\uppi_{n})-\phi(\gamma,\beta)\int_{0}^{1}e^{\phi(\gamma,\beta)(z-1)}F_{t}\left(z*(\uppi_{1},\dots,\uppi_{n})\right)\mathrm{d}z.
	\end{aligned}
\end{equation*} 

Then, we also have 
\begin{equation*}
	\begin{aligned}
		&\left\langle \one_{S},\int_{0}^{1}e^{\phi(\gamma,\beta)(z-1)}\nabla F_{t}\left(z*(\uppi_{1},\dots,\uppi_{n})\right)\mathrm{d}z\right\rangle\\
		&=\int_{0}^{1}e^{\phi(\gamma,\beta)(z-1)}\Big\langle \one_{S},\nabla F_{t}\left(z*(\uppi_{1},\dots,\uppi_{n})\right)\Big\rangle\mathrm{d}z\\
		&\ge\int_{0}^{1}e^{\phi(\gamma,\beta)(z-1)}\Big(\gamma^{2}f_{t}(S)-\phi(\gamma,\beta)F_{t}\big(z*(\uppi_{1},\dots,\uppi_{n})\big)\Big)\mathrm{d}z\\
		&=\Big(\gamma^{2}\int_{0}^{1}e^{\phi(\gamma,\beta)(z-1)}\mathrm{d}z\Big)f_{t}(S)-\phi(\gamma,\beta)\int_{0}^{1}e^{\phi(\gamma,\beta)(z-1)}F_{t}\left(z*(\uppi_{1},\dots,\uppi_{n})\right)\mathrm{d}z\\
		&=\left(\frac{\gamma^{2}(1-e^{-\phi(\gamma,\beta)})}{\phi(\gamma,\beta)}\right)f_{t}(S)-\phi(\gamma,\beta)\int_{0}^{1}e^{\phi(\gamma,\beta)(z-1)}F_{t}\left(z*(\uppi_{1},\dots,\uppi_{n})\right)\mathrm{d}z,
	\end{aligned}
\end{equation*}where the first inequality follows from the part \textbf{4)} of \cref{thm1}.

As a result, we have 
\begin{equation*}
	\left\langle\one_{S}-(\uppi_{1},\dots,\uppi_{n}),\int_{0}^{1}e^{\phi(\gamma,\beta)(z-1)}\nabla F_{t}\left(z*(\uppi_{1},\dots,\uppi_{n})\right)\mathrm{d}z\right\rangle
	\ge\left(\frac{\gamma^{2}(1-e^{-\phi(\gamma,\beta)})}{\phi(\gamma,\beta)}\right)f_{t}(S)-F_{t}(\uppi_{1},\dots,\uppi_{n}).
\end{equation*}
\end{proof}

From \cref{def:stationary}, we can show that, when the policy vector $(\uppi^{s}_{1},\dots,\uppi^{s}_{n})$ is the stationary point of the surrogate objective $F_{t}^{s}$ over domain $\prod_{i=1}^{n}\Delta_{\kappa_{i}}$ and $S$ satisfies the constraint of problem~\eqref{equ_problem}, due to $\one_{S}\in\prod_{i=1}^{n}\Delta_{\kappa_{i}}$ , the following inequality holds:
\begin{equation*}
	\left\langle\one_{S}-(\uppi^{s}_{1},\dots,\uppi^{s}_{n}),\int_{0}^{1}w(z)\nabla F_{t}\left(z*(\uppi^{s}_{1},\dots,\uppi^{s}_{n})\right)\mathrm{d}z\right\rangle\le0,
\end{equation*} where $w(z)$ is the corresponding weight function.

Then, following \cref{thm:C1}, we can show that when$f_{t}$ is monotone $\alpha$-weakly DR-submodular and $w(z)=e^{\alpha(z-1)}$, for any policy vector $(\uppi_{1},\dots,\uppi_{n})\in\prod_{i=1}^{n}\Delta_{\kappa_{i}}$ and any subset $S$ within the constraint of problem~\eqref{equ_problem},
\begin{equation*}
F_{t}(\uppi^{s}_{1},\dots,\uppi^{s}_{n})\ge\big(1-e^{-\alpha}\big)f_{t}(S).
\end{equation*}

So we get the result of  part \textbf{1)} of \cref{thm4}. Similarly, from \textbf{ii)} of \cref{thm:C1}, we also can achieve the part \textbf{2)} of \cref{thm4}.
\subsection{Proof of Part \textbf{3)}  in \texorpdfstring{\cref{thm4}}{}}\label{Sec:C2}
In this subsection, we prove \textbf{3)}  in \cref{thm4}.

At first, we define a modular function $g_{t}$ over $\V$, namely, for any subset $S\subseteq\V$, $g_{t}(S)\triangleq\sum_{a\in S}f_{t}\left(a|\V\setminus\{a\}\right)$. Then, when $f_{t}$ is monotone submodular, we can show that the set function $\left(f_{t}-g_{t}\right)$ is also a monotone submodular function. That is,
\begin{theorem}\label{thm:C2}
For a monotone submodular function $f_{t}:\V\rightarrow\R_{+}$, we define a modular function $g_{t}(S)\triangleq\sum_{a\in S}f_{t}\left(a|\V\setminus\{a\}\right),\forall S\subseteq\V$, then we also can show that the set function $\left(f_{t}-g_{t}\right)$ is monotone submodular and non-negative.
\end{theorem}
\begin{proof}
 Because the set function $\left(f_{t}-g_{t}\right)$ is the sum of  a submodular function  $f_{t}$ and a modular function $-g_{t}$,  $\left(f_{t}-g_{t}\right)$  must be submodular~\cite{sviridenko2017optimal}. Then, we show the monotonicity of the set function $\left(f_{t}-g_{t}\right)$. At first,
  for any subset $S\subseteq\V$ and any action $a\in\V\setminus S$, we have that
\begin{equation}\label{equ:thm:C2:1}
\left(f_{t}-g_{t}\right)\left(a|S\right)=f_{t}\left(a|S\right)-g_{t}\left(a|S\right)=f_{t}\left(a|S\right)-f_{t}\left(a|\V\setminus\{a\}\right)\ge0,
\end{equation} where the final inequality follows from the submodularity of $f_{t}$.

From Eq.\eqref{equ:thm:C2:1}, we can know that $\left(f_{t}-g_{t}\right)$ is monotone. Moreover, due to $f_{t}(\emptyset)=g_{t}(\emptyset)=0$,  $\left(f_{t}-g_{t}\right)$ is non-negative.
\end{proof}

Next, we show the  policy-based continuous extension of our  aforementioned modular function $g_{t}$.
\begin{theorem}\label{thm:C3}
	Given a monotone submodular function $f_{t}:\V\rightarrow\R_{+}$, we define its related modular function as $g_{t}(S)\triangleq\sum_{a\in S}f_{t}\left(a|\V\setminus\{a\}\right),\forall S\subseteq\V$. Then we can show that the policy-based continuous extension of $g_{t}$ is as follows: $G_{t}(\uppi_{1},\dots,\uppi_{n})\triangleq \sum_{i=1}^{n}\sum_{m=1}^{\kappa_{i}}\big(f_{t}\big(v_{i,k}|\V-\{v_{i,m}\}\big)\big)\pi_{i,m}$ if $\uppi_{i}=(\pi_{i,1},\dots,\pi_{i,\kappa_{i}})\in\Delta_{\kappa_{i}},\forall i\in\N$.
\end{theorem}
\begin{proof}
By the definition of $g_{t}$ and \cref{def_extension}, we can show that
\begin{equation*}
\begin{aligned}
	&G_{t}(\uppi_{1},\dots,\uppi_{n})=\sum_{a_{i}\in\V_{i}\cup\{\emptyset\},\forall i\in\N}\Big(g_{t}\big(\cup_{i=1}^{n}\{a_{i}\}\big)\prod_{i=1}^{n}p(a_{i}|\uppi_{i})\Big)\\
	&=\sum_{a_{i}\in\V_{i}\cup\{\emptyset\},\forall i\in\N}\Big(\Big(\sum_{i=1}^{n}f_{t}(a_{i}|\V\setminus\{a_{i}\})\Big)\prod_{i=1}^{n}p(a_{i}|\uppi_{i})\Big)\\
	&=\sum_{i=1}^{n}\sum_{a_{i}\in\V_{i}\cup\{\emptyset\}}f_{t}(a_{i}|\V\subseteq\{a_{i}\})p(a_{i}|\uppi_{i})\\
	&=\sum_{i=1}^{n}\sum_{m=1}^{\kappa_{i}}\big(f_{t}\big(v_{i,k}|\V-\{v_{i,m}\}\big)\big)\pi_{i,m},
\end{aligned}
\end{equation*} where the final equality follows from $\V_{i}=\{v_{i,1},\dots,v_{i,\kappa_{i}}\},\forall i\in\N$ .
\end{proof}

From \cref{thm:C2}, we know that, when $f_{t}$ is monotone submodular, so is $\left(f_{t}-g_{t}\right)$. Then, if we apply  \cref{thm:C1} to the monotone submodular function $\left(f_{t}-g_{t}\right)$, we can have the following relationship:
\begin{lemma}\label{lemma:C3}
Given a monotone submodular set function $f_{t}$, we define a modular function $g_{t}(S)\triangleq\sum_{a\in S}f_{t}\left(a|\V\setminus\{a\}\right),\forall S\subseteq\V$.  then, for any policy vector $(\uppi_{1},\dots,\uppi_{n})\in\prod_{i=1}^{n}\Delta_{\kappa_{i}}$ and any subset $S$ within the constraint of problem~\eqref{equ_problem}, the following inequality holds:
\begin{equation}\label{equ:boosting3}
	\begin{aligned}
		&\left\langle\one_{S}-(\uppi_{1},\dots,\uppi_{n}),\int_{0}^{1}e^{z-1}\nabla \left(F_{t}-G_{t}\right)\left(z*(\uppi_{1},\dots,\uppi_{n})\right)\mathrm{d}z\right\rangle\\
	&\ge(1-1/e)\left(f_{t}-g_{t}\right)(S)-\left(F_{t}-G_{t}\right)(\uppi_{1},\dots,\uppi_{n}),
	\end{aligned}
		\end{equation} where $F_{t}$ and $G_{t}$ is the corresponding policy-based continuous extension of $f_{t}$ and $g_{t}$, respectively.\end{lemma}
\begin{remark}
Note that the submodular objective function is also $1$-weakly DR-submodular fucntion. Thus, the results of \cref{lemma:C3} come from \cref{thm:C1} and \cref{thm:C2}.
\end{remark}
Moreover, from the definition of $G_{t}$ in \cref{thm:C3}, we know that $G_{t}$ is a linear function over $\prod_{i=1}\Delta_{\kappa_{i}}$ such that 
\begin{equation}\label{equ:boosting4}
\Big\langle\one_{S}-(\uppi_{1},\dots,\uppi_{n}),\nabla G_{t}(\uppi_{1},\dots,\uppi_{n})\Big\rangle=G_{t}(\one_{S})-G_{t}(\uppi_{1},\dots,\uppi_{n})=g_{t}(S)-G_{t}(\uppi_{1},\dots,\uppi_{n}).
\end{equation}

Merging Eq.\eqref{equ:boosting3} and  Eq.\eqref{equ:boosting4}, we get the following theorem:
\begin{theorem}\label{thm:C4}
Given a monotone submodular set function $f_{t}$ with curvature $c$, for its policy-based continuous extension $F_{t}$ introduced in \cref{def_extension}, if we consider a surrogate function  $F_{t}^{s}:\prod_{i=1}^{n}\Delta_{\kappa_{i}}\rightarrow\R_{+}$ whose gradient at each point $\x\in\prod_{i=1}^{n}\Delta_{\kappa_{i}}$ is a weighted average of the gradient $\nabla F_{t}(z*\x)$, namely, $\nabla F_{t}^{s}(\x)\triangleq\int_{0}^{1} e^{(z-1)}\nabla F_{t}(z*\boldsymbol{x})\mathrm{d}z$, we have that, for any policy vector $(\uppi_{1},\dots,\uppi_{n})\in\prod_{i=1}^{n}\Delta_{\kappa_{i}}$ and any subset $S$ within the constraint of problem~\eqref{equ_problem},  the following inequality holds:
\begin{equation}\label{equ:boosting_final}
	\begin{aligned}
		&\Big\langle\one_{S}-(\uppi_{1},\dots,\uppi_{n}),\nabla\left(F_{t}^{s}+\frac{G_{t}}{e}\right)(\uppi_{1},\dots,\uppi_{n})\Big\rangle\\&
		=\left\langle\one_{S}-(\uppi_{1},\dots,\uppi_{n}),\frac{\nabla G_{t}(\uppi_{1},\dots,\uppi_{n})}{e}+\int_{0}^{1}e^{\alpha(z-1)}\nabla F_{t}\left(z*(\uppi_{1},\dots,\uppi_{n})\right)\mathrm{d}z\right\rangle\\&
		\ge\left(1-\frac{c}{e}\right)f_{t}(S)-F_{t}(\uppi_{1},\dots,\uppi_{n}).
	\end{aligned}
\end{equation}
\end{theorem}
\begin{proof}
According to Eq.\eqref{equ:boosting3} and  Eq.\eqref{equ:boosting4}, we have that
\begin{equation}\label{thm:C4:1}
\begin{aligned}
	&\left\langle\one_{S}-(\uppi_{1},\dots,\uppi_{n}),\nabla G_{t}(\uppi_{1},\dots,\uppi_{n})+\int_{0}^{1}e^{z-1}\nabla \left(F_{t}-G_{t}\right)\left(z*(\uppi_{1},\dots,\uppi_{n})\right)\mathrm{d}z\right\rangle\\
	&\ge(1-1/e)f_{t}(S)+\frac{g_{t}(S)}{e}-F_{t}((\uppi_{1},\dots,\uppi_{n})\\
	&\ge(1-1/e)f_{t}(S)+\frac{(1-c)f_{t}(S)}{e}-F_{t}((\uppi_{1},\dots,\uppi_{n})\\
	&=(1-c/e)f_{t}(S)-F_{t}((\uppi_{1},\dots,\uppi_{n}),
\end{aligned}
\end{equation} where the final inequality follows from the definition of curvature $c$ (See Lemma 2.1. in \cite{sviridenko2017optimal}).

Moreover, due to that $G_{t}$ is a linear function, we have 
\begin{equation}\label{thm:C4:2}
	\begin{aligned}
		&\int_{0}^{1}e^{z-1}\nabla\left(F_{t}-G_{t}\right)\left(z*(\uppi_{1},\dots,\uppi_{n})\right)\mathrm{d}z\\&=\int_{0}^{1}e^{z-1}\nabla F_{t}\left(z*(\uppi_{1},\dots,\uppi_{n})\right)\mathrm{d}z-	\int_{0}^{1}e^{z-1}\mathrm{d}z\nabla G_{t}(\uppi_{1},\dots,\uppi_{n})\\
		&=\int_{0}^{1}e^{z-1}\nabla F_{t}\left(z*(\uppi_{1},\dots,\uppi_{n})\right)\mathrm{d}z-	(1-1/e)\nabla G_{t}(\uppi_{1},\dots,\uppi_{n}),
	\end{aligned}
\end{equation}  where the first equality follows from $\nabla G_{t}\left(z*(\uppi_{1},\dots,\uppi_{n})\right)=\nabla G_{t}(\uppi_{1},\dots,\uppi_{n})$.

Merging Eq.\eqref{thm:C4:1} and Eq.\eqref{thm:C4:2}, we get the Eq.\eqref{equ:boosting_final}.
\end{proof}

From the result of \cref{thm:C4},  for any stationary point $(\uppi^{s}_{1},\dots,\uppi^{s}_{n})$ of the objective $(F_{t}^{s}+\frac{G_{t}}{e})$ over $\prod_{i=1}^{n}\Delta_{\kappa_{i}}$ and any subset $S$ within the constraint of problem~\eqref{equ_problem}, then we have the following inequality:
\begin{equation*}
	\begin{aligned}
		0\ge\Big\langle\one_{S}-(\uppi^{s}_{1},\dots,\uppi^{s}_{n}),\nabla\left(F_{t}^{s}+\frac{G_{t}}{e}\right)(\uppi^{s}_{1},\dots,\uppi^{s}_{n})\Big\rangle\ge\left(1-\frac{c}{e}\right)f_{t}(S)-F_{t}(\uppi^{s}_{1},\dots,\uppi^{s}_{n}).
	\end{aligned}
\end{equation*}
In other words, $F_{t}(\uppi^{s}_{1},\dots,\uppi^{s}_{n})\ge\big(1-\frac{c}{e}\big)f_{t}(S^{*})$ where $S^{*}$ is the optimal subset of problem~\eqref{equ_problem}.
\section{Convergence  Analysis of \texorpdfstring{\texttt{MA-SPL}}{} Algorithm}\label{appendix_proof_thm1}
In this section, we verify \cref{thm:result1}.

Before that, we firstly show that  the surrogate gradient estimations in Line 16 of \cref{alg:SPL}, namely, $\mathbf{d}_{i}(t)$ is bounded and our proposed policy-based continuous extension $F_{t}$ is smooth.
\begin{lemma}\label{lemma:D1} 
	If we set $M$ as the maximum marginal contribution of each monotone set function $f_{t}$, namely, $M\triangleq\max_{S\subseteq\V,e\in\V\setminus S,t\in[T]}\left(f_{t}(e|S)\right)$,then we have $\|\mathbf{d}_{i}(t)\|_{2}\le\sqrt{\kappa_{i}}(\int_{0}^{1}w(z)\mathrm{d}z+\frac{1}{e})M$ where $\mathbf{d}_{i}(t)$ is the surrogate gradient estimations in Line 16 of \cref{alg:SPL}.
\end{lemma}
\begin{proof}
	From Lines 13-15, we know that $d_{i,m}(t)\triangleq(\int_{z=0}^{1}w(z)\mathrm{d}z)f_{t}\big(v_{i,m}\big|S_{i}(t)\big)$ or $d_{i,m}(t)\triangleq\left((\int_{z=0}^{1}w(z)\mathrm{d}z)f_{t}\big(v_{i,m}\big|S_{i}(t)\big)+e^{-1}f_{t}(v_{i,m}|\V-\{v_{i,m}\})\right)$. Thus, we have $|d_{i,m}(t)|\le(\int_{0}^{1}w(z)\mathrm{d}z+\frac{1}{e})M$ for any $i\in\N$ and $m\in[\kappa_{i}]$ such that $\|\mathbf{d}_{i}(t)\|_{2}=\sqrt{\sum_{m=1}^{2}|d_{i,m}(t)|^{2}}\le\sqrt{\kappa_{i}}(\int_{0}^{1}w(z)\mathrm{d}z+\frac{1}{e})M$.
\end{proof}
Before investigating the smoothness of our proposed policy-based continuous extension $F_{t}$, we firstly show its second-order partial derivative, i.e.,
\begin{lemma}\label{lemma:D2}
For any vector $(\uppi_{1},\dots,\uppi_{n})\in\prod_{i=1}^{n}\Delta_{\kappa_{i}}$, the second-order derivative of $F_{t}$ at variables $\pi_{i_{1},m_{1}},\pi_{i_{2},m_{2}},\forall i_{1}\in\N,\forall m_{1}\in[\kappa_{i_{1}}],\forall i_{2}\in\N,\forall m_{2}\in[\kappa_{i_{2}}],$ can be defined as:

\textbf{i): When $i_{1}= i_{2}$}, $\frac{\partial^{2}F_{t}}{\partial \pi_{i_{1},m_{1}}\partial \pi_{i_{1},m_{2}}}(\uppi_{1},\dots,\uppi_{n})=0$;

\textbf{ii): When $i_{1}\neq i_{2}$},
\begin{equation*}
\begin{aligned}
&\frac{\partial^{2}F_{t}}{\partial \pi_{i_{1},m_{1}}\partial \pi_{i_{2},m_{2}}}(\uppi_{i,1}(t),\dots,\uppi_{i,n}(t))\\&=\E_{a_{i}\sim\uppi_{i},\forall i\in\N}\Bigg(f_{t}\Big(v_{i_{1},m_{1}}\big|\big(\cup_{i\notin\left\{i_{1},i_{2}\right\}}\{a_{i}\}\big)\cup\{v_{i_{2},m_{2}}\}\Big)-f_{t}\Big(v_{i_{1},m_{1}}\big|\cup_{i\notin\left\{i_{1},i_{2}\right\}}\{a_{i}\}\Big)\Bigg).
\end{aligned}	
\end{equation*}
\end{lemma}
\begin{remark}\label{remak:D2}
From the \cref{lemma:D2}, if we set $M$ as the maximum marginal contribution of each monotone set function $f_{t}$, namely, $M\triangleq\max_{S\subseteq\V,e\in\V\setminus S,t\in[T]}\left(f_{t}(e|S)\right)$, then we can show that $|\frac{\partial^{2}F_{t}}{\partial \pi_{i_{1},m_{1}}\partial \pi_{i_{2},m_{2}}}(\uppi_{i,1}(t),\dots,\uppi_{i,n}(t))|\le M$.
\end{remark}
\begin{lemma}\label{lemma:D3}
If we set $M$ as the maximum marginal contribution of each monotone set function $f_{t}$, namely, $M\triangleq\max_{S\subseteq\V,e\in\V\setminus S,t\in[T]}\left(f_{t}(e|S)\right)$, then we can show that the Hessian matrix of the policy-based continuous extension $F_{t}$ of the previous monotone set function $f_{t}$ satisfies that,
\begin{equation}\label{equ:result:D3}
	\|\nabla^{2}F_{t}(\uppi_{1},\dots,\uppi_{n})\|^{2}_{2,\infty}\le M^{2}\sum_{i=1}^{n}\kappa_{i},
\end{equation} where $(\uppi_{1},\dots,\uppi_{n})\in\prod_{i=1}^{n}\Delta_{\kappa_{i}}$.
\end{lemma}
\begin{remark}
	For any matrix $A\in\R^{n\times n}$, the $(2,\infty)$-norm of $A$ is defined as $\|A\|_{2,\infty}=\sup\{\|A\x\|_{\infty}: \x\in\R^{n},\|\x\|_{2}=1\}$ where $\|\cdot\|_{2}$ denotes the $L_{2}$ norm.
\end{remark}
\begin{proof}
	From the definition of the norm $\|\cdot\|_{2,\infty}$, we can show that 
	\begin{equation*}
		\|\nabla^{2}F_{t}(\uppi_{1},\dots,\uppi_{n})\|^{2}_{2,\infty}=\max_{j\in\N}\|\nabla^{2}F_{t}(\uppi_{1},\dots,\uppi_{n})[j,:]\|_{2}^{2}\le M^{2}\sum_{i=1}^{n}\kappa_{i},
	\end{equation*} where $\nabla^{2}F_{t}(\uppi_{1},\dots,\uppi_{n})[j,:]$ is the $j$-th line of the Hessian matrix $\nabla^{2}F_{t}(\uppi_{1},\dots,\uppi_{n})$ and the final inequality follows from \cref{remak:D2}.
\end{proof}
From \cref{lemma:D3}, we can show that our proposed policy-based continuous extension $F_{t}$ is $\left(M\sum_{i=1}^{n}\kappa_{i}\right)$-smooth, namely,

\begin{lemma}\label{lemma:D4}
If we set $M$ as the maximum marginal contribution of each monotone set function $f_{t}$, namely, $M\triangleq\max_{S\subseteq\V,e\in\V\setminus S,t\in[T]}\left(f_{t}(e|S)\right)$, then we can show that our proposed policy-based continuous extension $F_{t}$ is $\left(M\sum_{i=1}^{n}\kappa_{i}\right)$-smooth.
\end{lemma} 
\begin{proof}
For any two point $(\uppi^{a}_{1},\dots,\uppi^{a}_{n})\in\prod_{i=1}^{n}\Delta_{\kappa_{i}}$ and $(\uppi^{b}_{1},\dots,\uppi^{b}_{n})\in\prod_{i=1}^{n}\Delta_{\kappa_{i}}$, if  $\uppi^{a}_{i}\le\uppi^{b}_{i}$ for any $i\in\N$, the following inequality holds:
\begin{equation*}
	\begin{aligned}
		&\|\nabla F_{t}(\uppi^{a}_{1},\dots,\uppi^{a}_{n})-\nabla F_{t}(\uppi^{b}_{1},\dots,\uppi^{b}_{n})\|_{2}\\
		&=\|\int_{0}^{1}\nabla^{2}F_{t}(\uppi^{a}_{1}+\lambda\uppi^{b}_{1},\dots,\uppi^{a}_{n}+\lambda\uppi^{b}_{n})\mathrm{d}\lambda(\uppi^{a}_{1}-\uppi^{b}_{1},\dots,\uppi^{a}_{n}-\uppi^{b}_{n})\|_{2}\\
		&\le\left(\sum_{i=1}^{n}\kappa_{i}\right)^{1/2}\|\int_{0}^{1}\nabla^{2}F_{t}(\uppi^{a}_{1}+\lambda\uppi^{b}_{1},\dots,\uppi^{a}_{n}+\lambda\uppi^{b}_{n})(\uppi^{a}_{1}-\uppi^{b}_{1},\dots,\uppi^{a}_{n}-\uppi^{b}_{n})\mathrm{d}\lambda\|_{\infty}\\
		&\le M(\sum_{i=1}^{n}\kappa_{i})\|(\uppi^{a}_{1},\dots,\uppi^{a}_{n})-(\uppi^{b}_{1},\dots,\uppi^{b}_{n})\|_{2},
	\end{aligned}	
\end{equation*} where the first inequality comes from $\|\x\|_{2}\le\sqrt{n}\|\x\|_{\infty}$ for any $n$-dimensional vector $\x$ and the final inequality follows from the Eq.\eqref{equ:result:D3}. Thus, $F_{t}$ is $\left(M\sum_{i=1}^{n}\kappa_{i}\right)$-smooth.
\end{proof}

With \cref{lemma:D1} and \cref{lemma:D4}, we next prove \cref{thm:result1}. At first, we show that each policy vector $(\uppi_{i,1}(t),\dots,\uppi_{i,n}(t)),\forall i\in\N,\forall t\in[T]$  of \cref{alg:SPL} is included in the constraint $\prod_{i=1}^{n}\Delta_{\kappa_{i}}$. That is, we have the following lemma:
\begin{lemma}\label{lemma:included_convex_constaint}
	 In \cref{alg:SPL}, Under \cref{ass:1}, $(\uppi_{i,1}(t),\dots,\uppi_{i,n}(t))\in\prod_{i=1}^{n}\Delta_{\kappa_{i}}$.
\end{lemma}
\begin{proof}
	We prove this theorem by induction.
	At first, from the Line 2 in  \cref{alg:SPL}, we know that $\x_{i,j}(1)\in\Delta_{\kappa_{j}}$ for any $i,j\in\N$, so $(\uppi_{i,1}(1),\dots,\uppi_{i,n}(1))\in\prod_{i=1}^{n}\Delta_{\kappa_{i}},\forall i\in\N$.
	Then, if for some $t\in[T]$, $(\uppi_{i,1}(t),\dots,\uppi_{i,n}(t))\in\prod_{i=1}^{n}\Delta_{\kappa_{i}}$. From Line 17, we know that $\uppi_{i,m}(t+1):=\sum_{j\in\mathcal{N}_{i}\cup\{i\}}w_{ij}\uppi_{j,m}(t)$ for any $m\neq i$. Then, due to \cref{ass:1}, we have $\sum_{j\in\mathcal{N}_{i}\cup\{i\}}w_{ij}=1$ and $w_{ij}\ge0$ such that $\uppi_{i,m}(t+1)\in\Delta_{\kappa_{m}},\forall m\neq i$. Moreover, from Line 19, we also know that $uppi_{i,i}(t+1)\in\Delta_{\kappa_{i}}$. Thus, $(\uppi_{i,1}(t+1),\dots,\uppi_{i,n}(t+1))\in\prod_{i=1}^{n}\Delta_{\kappa_{i}}$.
\end{proof}

Before going into the proof of \cref{thm:result1}, we define some commonly used symbols, namely, 
\begin{equation*}
	\begin{aligned}
		&\bar{\uppi}_{\cdot,j}(t)\triangleq\frac{\sum_{i=1}^{n}\uppi_{i,j}(t)}{n},\ \ \ \uppi_{\cdot,j}^{cate}(t)\triangleq[\uppi_{1,j}(t);\uppi_{2,j}(t);\dots;\uppi_{n,j}(t)]\in\R^{n*\kappa_{j}}, \ \ \forall j\in\N;\\
		&\z_{i,m}(t)\triangleq\sum_{j\in\mathcal{N}_{i}\cup\{i\}}w_{ij}\uppi_{j,m}(t),\ \ \forall i,m\in\N;\\
    	&\bar{\z}_{\cdot,j}(t)\triangleq\frac{\sum_{i=1}^{n}\z_{i,j}(t)}{n},\ \ \ \z_{\cdot,j}^{cate}(t)\triangleq[\z_{1,j}(t);\z_{2,j}(t);\dots;\z_{n,j}(t)]\in\R^{n*\kappa_{j}}, \ \ \forall j\in\N;\\
		&\mathbf{r}_{i,j}(t)\triangleq\uppi_{i,j}(t+1)-\z_{i,j}(t),\ \ \ \mathbf{r}_{\cdot,j}^{cate}(t)\triangleq[\mathbf{r}_{1,j}(t);\mathbf{r}_{2,j}(t);\dots;\mathbf{r}_{n,j}(t)]\in\R^{n*\kappa_{j}}, \ \ \forall j\in\N.
	\end{aligned}
\end{equation*}
With these symbols, we can verify that
\begin{lemma}\label{lemma:D5}  $\E(\|\mathbf{r}_{i,j}(t)\|_{2})=\E(\|\uppi_{i,j}(t+1)-\z_{i,j}(t)\|_{2})\le\sqrt{\kappa_{i}}(\int_{0}^{1}w(z)\mathrm{d}z+\frac{1}{e})\eta_{t}M$ where $M$ is the maximum marginal contribution of each monotone set function $f_{t}$.

\end{lemma}
\begin{proof}
When $i\neq j$, from Line 17, we know that $\uppi_{i,j}(t+1)=\z_{i,j}(t)$ such that $\mathbf{r}_{i,j}(t)=\mathbf{0}_{\kappa_{j}}$ and $\E(\|\mathbf{r}_{i,j}(t)\|_{2})=0\le\sqrt{\kappa_{i}}(\int_{0}^{1}w(z)\mathrm{d}z+\frac{1}{e})\eta_{t}M$. As for $i=j$, from Line 18 and Line 19, we have 
		\begin{equation*}
		\uppi_{i,i}(t+1):=\mathop{\arg\min}_{\mathbf{b}\in\Delta_{\kappa_{i}}}\|\mathbf{b}-\left(\z_{i,i}(t)+\eta_{t}\mathbf{d}_{i}(t)\right)\|_{2}.
	\end{equation*}
	Note that $\z_{i,i}(t)\in\Delta_{\kappa_{j}}$. Thus, we have 
	\begin{equation}\label{equ:lemma21}
	\|\uppi_{i,i}(t+1)-\z_{i,i}(t)\|_{2}\le\|\z_{i,i}(t)-\left(\z_{i,i}(t)+\eta_{t}\mathbf{d}_{i}(t)\right)\|_{2}=\eta_{t}\|\mathbf{d}_{i}(t)\|_{2}\le\sqrt{\kappa_{i}}(\int_{0}^{1}w(z)\mathrm{d}z+\frac{1}{e})\eta_{t}M,
	\end{equation} where the final inequality follows from \cref{lemma:D1}.
\end{proof}
 \begin{lemma}\label{lemma:D6}
 	Under \cref{ass:1}, for any $t\in[T]$ and $i\in\N$, we have that
 	\begin{equation*}
 		\begin{aligned}
 			&\sum_{i\in\N}\E(\|\uppi_{i,j}(t+1)-\bar{\uppi}_{\cdot,j}(t+1)\|)\le\sum_{m=1}^{t}\sqrt{n}\beta^{t-m}(\int_{0}^{1}w(z)\mathrm{d}z+\frac{1}{e})\eta_{m}M,\\
 			&\sum_{i\in\N}\E(\|\z_{i,j}(t+1)-\bar{\uppi}_{\cdot,j}(t+1)\|)\le\sum_{m=1}^{t}\sqrt{n}\beta^{t-m}(\int_{0}^{1}w(z)\mathrm{d}z+\frac{1}{e})\eta_{m}M,
 		\end{aligned}
 	\end{equation*} where $\tau\triangleq\max(|\lambda_{2}(\W)|,|\lambda_{n}(\W)|)$ is the second largest magnitude of the eigenvalues of the weight matrix $\W$.
 \end{lemma}
\begin{proof}
	From the definition of $\mathbf{r}_{i,j}(t)$, we can show that 
	\begin{equation}\label{lemma31}
		\x_{i,j}(t+1)=\mathbf{r}_{i,j}(t)+\z_{i,j}(t)=\mathbf{r}_{i,j}(t)+\sum_{k\in\N_{i}\cup\{i\}}w_{ik}\x_{k,j}.
	\end{equation} 
	
	As a result, from the Eq.\eqref{lemma31}, we also have that
	\begin{equation}\label{lemma32}
		\begin{aligned}
			\uppi_{\cdot,j}^{cate}(t+1)& =\mathbf{r}_{\cdot,j}^{cate}(t)+(\W\otimes\mathbf{I}_{\kappa_{j}})\uppi_{\cdot,j}^{cate}(t)\\
			&=\sum_{m=1}^{t}(\W\otimes\mathbf{I}_{\kappa_{j}})^{t-m}\mathbf{r}_{\cdot,j}^{cate}(m)\\
			&=\sum_{m=1}^{t}(\W^{t-m}\otimes\mathbf{I}_{\kappa_{j}})\mathbf{r}_{\cdot,j}^{cate}(m),
		\end{aligned}
	\end{equation} where the symbol $\otimes$ denotes the Kronecker product.
	
	If we also define $\bar{\uppi}_{\cdot,j}^{cate}(t)=[\bar{\uppi}_{\cdot,j}(t);\bar{\uppi}_{\cdot,j}(t);\dots;\bar{\uppi}_{\cdot,j}(t)]\in\R^{n\kappa_{j}}$ and from the Eq.\eqref{lemma32}, we also have that
	\begin{equation}\label{lemma33}
		\begin{aligned}				\bar{\uppi}_{\cdot,j}^{cate}(t)&=\left(\frac{\one_{n}\one_{n}^{T}}{n}\otimes\mathbf{I}_{\kappa_{j}}\right)\uppi_{\cdot,j}^{cate}(t)\\
			&=\sum_{m=1}^{t}(\frac{\one_{n}\one_{n}^{T}}{n}\otimes\mathbf{I}_{\kappa_{j}})\mathbf{r}_{\cdot,j}^{cate}(m).
		\end{aligned}
	\end{equation}
	Then, from the Eq.\eqref{lemma32} and Eq.\eqref{lemma33}, we have that , for any $i\in\N$,
	\begin{equation}\label{lemma34}
		\uppi_{\cdot,j}^{cate}(t+1)-\bar{\uppi}_{\cdot,j}^{cate}(t)=\sum_{m=1}^{t}\sum_{j\in\N_{i}\cup\{i\}}\left([\W^{t-m}]_{ij}-\frac{1}{N}\right)\mathbf{r}_{\cdot,j}^{cate}(m).
	\end{equation}
	Eq.\eqref{lemma34} indicates that
	\begin{equation*}
		\begin{aligned}
			\E\left(\left\|	\uppi_{\cdot,j}^{cate}(t+1)-\bar{\uppi}_{\cdot,j}^{cate}(t)\right\|\right)&=\E\left(\left\|\sum_{m=1}^{t}\sum_{j\in\N_{i}\cup\{i\}}\left([\W^{t-m}]_{ij}-\frac{1}{n}\right)\mathbf{r}_{\cdot,j}^{cate}(m)\right\|\right)
			\\&\le\E\left(\sum_{m=1}^{t}\sum_{j\in\N_{i}\cup\{i\}}|\left([\W^{t-m}]_{ij}-\frac{1}{n}\right)|*\|\mathbf{r}_{\cdot,j}^{cate}(m)\|_{2}\right)\\&\le\sum_{m=1}^{t}\sum_{j\in\N_{i}\cup\{i\}}|\left([\W^{t-m}]_{ij}-\frac{1}{n}\right)|\sqrt{\kappa_{i}}(\int_{0}^{1}w(z)\mathrm{d}z+\frac{1}{e})\eta_{m}M\\&\le\sum_{m=1}^{t}\sqrt{n}\tau^{t-m}(\int_{0}^{1}w(z)\mathrm{d}z+\frac{1}{e})\eta_{m}M,
		\end{aligned}
	\end{equation*} where the second inequality comes from \cref{lemma:D5} and the final inequality follows from $\sum_{j\in\N_{i}\cup\{i\}}|[\W^{t-m}]_{ij}-\frac{1}{n}|\le\sqrt{n}\tau^{t-m}$(See Proposition 1  in \citep{nedic2009distributed}).
	Due to $\z_{i,j}(t+1)=\sum_{k\in\N_{i}\cup\{i\}}w_{ik}\uppi_{k,j}(t+1)$ we also can have $\E(\|z_{i,j}(t+1)-\bar{\uppi}_{\cdot,j}(t+1)\|_{2})\le\sum_{j\in\N_{i}\cup\{i\}}w_{ij}\E(\|\uppi_{i,j}(t+1)-\bar{\uppi}_{\cdot,j}(t+1)\|_{2})\le\sum_{m=1}^{t}\sqrt{n}\tau^{t-m}(\int_{0}^{1}w(z)\mathrm{d}z+\frac{1}{e})\eta_{m}M$.
	\end{proof}
	\begin{lemma}\label{lemma:D7} 
	 When each objective function $f_{t}$  is monotone submodular with curvature $c$, $\alpha$-weakly DR-submodular or  $(\gamma,\beta)$-weakly submodular, we set the ratio $\rho=\left(1-\frac{c}{e}\right),\left(1-e^{-\alpha}\right),\left(\frac{\gamma^{2}(1-e^{-(\beta(1-\gamma)+\gamma^2)})}{\beta(1-\gamma)+\gamma^2}\right)$, respectively. Moreover, we use the symbol $M$ to denote the maximum marginal contribution of each monotone set function $f_{t}$, namely, $M\triangleq\max_{S\subseteq\V,e\in\V\setminus S,t\in[T]}\left(f_{t}(e|S)\right)$. Then, under \cref{ass:1}, if each set function $f_{t}$ is bounded $\forall t\in[T]$,  the following inequality holds for \cref{alg:SPL}:
	\begin{equation*}
		\begin{aligned}
			&\rho\sum_{t=1}^{T}f_{t}(\mathcal{A}_{t}^{*})-\sum_{t=1}^{T}F_{t}(\bar{\uppi}_{\cdot,1}(t),\dots,\bar{\uppi}_{\cdot,n}(t))\le5M^{2}n^{\frac{5}{2}}(\sum_{i=1}^{n}\kappa_{i})(\int_{0}^{1}w(z)\mathrm{d}z+\frac{1}{e})^{2}\sum_{t=1}^{T}\sum_{m=1}^{t}\tau^{t-m}\eta_{m}\\
			&+n\sqrt{\kappa_{i}}(\int_{0}^{1}w(z)\mathrm{d}z+\frac{1}{e})M\sum_{t=1}^{T}\frac{\eta_{t}}{2}+\sum_{t=1}^{T}\frac{1}{2\eta_{t}}\sum_{i=1}^{n}\Big(\|\z_{i,i}(t)-\one_{(\mathcal{A}_{t}^{*}\cap\V_{i})}\|_{2}^{2}-\|\uppi_{i,i}(t+1)-\one_{(\mathcal{A}_{t}^{*}\cap\V_{i})}\|^{2}_{2}\Big).
		\end{aligned}
	\end{equation*} where $\mathcal{A}_{t}^{*}$ is the optimal solution of problem~\eqref{equ_problem} and $\tau\triangleq\max(|\lambda_{2}(\W)|,|\lambda_{n}(\W)|)$ is the second largest magnitude of the eigenvalues of the weight matrix $\W$.
	\end{lemma}
	\begin{proof}
In order to unify the proofs in different setting, i.e., submodular objective with curvature $c$, $\alpha$-weakly DR-submodular objective and $(\gamma,\beta)$-weakly submodular objective, we define some auxiliary symbols. At first, when each objective function $f_{t}$  is monotone submodular with curvature $c$, $\alpha$-weakly DR-submodular or  $(\gamma,\beta)$-weakly submodular, we set $\rho=\left(1-\frac{c}{e}\right),\left(1-e^{-\alpha}\right),\left(\frac{\gamma^{2}(1-e^{-(\beta(1-\gamma)+\gamma^2)})}{\beta(1-\gamma)+\gamma^2}\right)$, respectively. Moreover, compared to weakly submodular scenarios, \cref{thm4} designs a different auxiliary function for submodular scenarios. Therefore, when each $f_{t}$  is monotone submodular with curvature $c$or  weakly submodular, we define $F_{t}^{a}=F_{t}^{s}$ or $F_{t}^{a}=\left(F_{t}^{s}+\frac{G_{t}}{e}\right)$, respectively. Note that the $F_{t}^{s}$ is the surrogate function considered in \cref{thm4} for different settings and $G_{t}$ is a linear function defined as $G_{t}(\uppi_{1},\dots,\uppi_{n})\triangleq \sum_{i=1}^{n}\sum_{m=1}^{\kappa_{i}}\big(f_{t}\big(v_{i,k}|\V-\{v_{i,m}\}\big)\big)\pi_{i,m}$. 

From \cref{thm:C1} and \cref{thm:C4}, we have that
		\begin{equation}
			\label{equ:lemma_result_0}
			\begin{aligned}
				&\rho f_{t}(\mathcal{A}_{t}^{*})-F_{t}(\bar{\uppi}_{\cdot,1}(t),\dots,\bar{\uppi}_{\cdot,n}(t))\\&\le\Big\langle\nabla F_{t}^{a}(\bar{\uppi}_{\cdot,1}(t),\dots,\bar{\uppi}_{\cdot,n}(t)),\one_{\mathcal{A}_{t}^{*}}-(\bar{\uppi}_{\cdot,1}(t),\dots,\bar{\uppi}_{\cdot,n}(t))\Big\rangle\\
				&=\underbrace{\Big\langle\nabla F_{t}^{a}(\bar{\uppi}_{\cdot,1}(t),\dots,\bar{\uppi}_{\cdot,n}(t))-\sum_{i\in\N}\Big[\nabla  F_{t}^{a}(\uppi_{i,1}(t),\dots,\uppi_{i,1}(t))\Big]_{\uppi_{i}},\one_{\mathcal{A}_{t}^{*}}-(\bar{\uppi}_{\cdot,1}(t),\dots,\bar{\uppi}_{\cdot,n}(t))\Big\rangle}_{\text{\textcircled{1}}}\\
				&\quad+\underbrace{\Big\langle\sum_{i\in\N}\Big[\nabla  F_{t}^{a}(\uppi_{i,1}(t),\dots,\uppi_{i,1}(t))\Big]_{\uppi_{i}},\one_{\mathcal{A}_{t}^{*}}-(\uppi_{i,1}(t),\dots,\uppi_{i,1}(t))\Big\rangle}_{\text{\textcircled{2}}}\\
				&\quad+\underbrace{\Big\langle\sum_{i\in\N}\Big[\nabla  F_{t}^{a}(\uppi_{i,1}(t),\dots,\uppi_{i,1}(t))\Big]_{\uppi_{i}},(\uppi_{i,1}(t),\dots,\uppi_{i,1}(t))-(\bar{\uppi}_{\cdot,1}(t),\dots,\bar{\uppi}_{\cdot,n}(t))\Big\rangle}_{\text{\textcircled{3}}}
			\end{aligned}
		\end{equation} where the symbol $\Big[\nabla  F_{t}^{a}(\uppi_{i,1}(t),\dots,\uppi_{i,1}(t))\Big]_{\uppi_{i}}$ is the projection over the policy $\uppi_{i}$, i.e., $\Big[\nabla  F_{t}^{a}(\uppi_{i,1}(t),\dots,\uppi_{i,1}(t))\Big]_{\uppi_{i}}$ represents a $\left(\sum_{i=1}^{n}\kappa_{i}\right)$-dimensional vector that only keeps the first-order derivative at variable $\pi_{i,m},\forall m\in[\kappa_{i}]$ and set other coordinates to $0$, that is to say, 
		\begin{equation*}
		\Big[\nabla  F_{t}^{a}(\uppi_{i,1}(t),\dots,\uppi_{i,1}(t))\Big]_{\uppi_{i}}\triangleq\Bigg(\dots,0,\dots,\underbrace{\frac{\partial F_{t}}{\partial\pi_{i,1}}(\x),\dots,\frac{\partial F_{t}}{\partial\pi_{i,\kappa_{i}}}(\x)}_{\kappa_{i}},\dots,0,\dots\Bigg),
		\end{equation*} where $\x\triangleq(\uppi_{i,1}(t),\dots,\uppi_{i,1}(t))$.

		For \textcircled{1}, we have
		\begin{equation}\label{equ:lemma_result_1}
			\begin{aligned}
				&\Big\langle\nabla F_{t}^{a}(\bar{\uppi}_{\cdot,1}(t),\dots,\bar{\uppi}_{\cdot,n}(t))-\sum_{i\in\N}\Big[\nabla  F_{t}^{a}(\uppi_{i,1}(t),\dots,\uppi_{i,n}(t))\Big]_{\uppi_{i}},\one_{\mathcal{A}_{t}^{*}}-(\bar{\uppi}_{\cdot,1}(t),\dots,\bar{\uppi}_{\cdot,n}(t))\Big\rangle\\
				&\le\left\|\nabla F_{t}^{a}(\bar{\uppi}_{\cdot,1}(t),\dots,\bar{\uppi}_{\cdot,n}(t))-\sum_{i\in\N}\Big[\nabla  F_{t}^{a}(\uppi_{i,1}(t),\dots,\uppi_{i,n}(t))\Big]_{\uppi_{i}}\right\|_{2}\left\|\one_{\mathcal{A}_{t}^{*}}-(\bar{\uppi}_{\cdot,1}(t),\dots,\bar{\uppi}_{\cdot,n}(t))\right\|_{2}\\
				&\le\sum_{i\in\N}\Big(\left\|\Big[\nabla F_{t}^{a}(\bar{\uppi}_{\cdot,1}(t),\dots,\bar{\uppi}_{\cdot,n}(t))\Big]_{\uppi_{i}}-\Big[\nabla  F_{t}^{a}(\uppi_{i,1}(t),\dots,\uppi_{i,n}(t))\Big]_{\uppi_{i}}\right\|_{2}\Big)\Big(\left\|\one_{\mathcal{A}_{t}^{*}}\|_{2}+\|(\bar{\uppi}_{\cdot,1}(t),\dots,\bar{\uppi}_{\cdot,n}(t))\right\|_{2}\Big)\\
				&\le2n\sum_{i\in\N}\Big(\left\|\nabla F_{t}^{a}(\bar{\uppi}_{\cdot,1}(t),\dots,\bar{\uppi}_{\cdot,n}(t))-\nabla  F_{t}^{a}(\uppi_{i,1}(t),\dots,\uppi_{i,n}(t))\right\|_{2}\Big)\\&=2n\sum_{i\in\N}\Big(\left\|\int_{z=1}^{1}w(z)\nabla F_{t}(z*\bar{\uppi}_{\cdot,1}(t),\dots,z*\bar{\uppi}_{\cdot,n}(t))\mathrm{d}z-\int_{z=1}^{1}w(z)\nabla  F_{t}(z*\uppi_{i,1}(t),\dots,z*\uppi_{i,n}(t))\mathrm{d}z\right\|_{2}\Big)\\
				&\le2Mn(\sum_{i=1}^{n}\kappa_{i})(\int_{0}^{1}w(z)z\mathrm{d}z)\|\sum_{i\in\N}\Big(\|(\uppi_{i,1}(t),\dots,\uppi_{i,n}(t))-(\bar{\uppi}_{\cdot,1}(t),\dots,\bar{\uppi}_{\cdot,n}(t))\|_{2}\Big)\\
				&\le2Mn(\sum_{i=1}^{n}\kappa_{i})(\int_{0}^{1}w(z)z\mathrm{d}z)\sum_{j\in\N}\sum_{i\in\N}\|\uppi_{i,j}(t)-\bar{\uppi}_{\cdot,j}(t)\|_{2})\\
				&\le2M^{2}n^{\frac{5}{2}}(\sum_{i=1}^{n}\kappa_{i})(\int_{0}^{1}w(z)z\mathrm{d}z)(\int_{0}^{1}w(z)\mathrm{d}z+\frac{1}{e})\sum_{m=1}^{t}\tau^{t-m}\eta_{m},
			\end{aligned}
		\end{equation} where  the fourth inequality follows \cref{lemma:D4} and the final inequality from \cref{lemma:D6}.
		
		For \textcircled{3}, under \cref{lemma:D1} and \cref{lemma:D6}, 
		we have, 
		\begin{equation}\label{equ:lemma_result_3}
			\begin{aligned}
			&\Big\langle\sum_{i\in\N}\Big[\nabla  F_{t}^{a}(\uppi_{i,1}(t),\dots,\uppi_{i,1}(t))\Big]_{\uppi_{i}},(\uppi_{i,1}(t),\dots,\uppi_{i,1}(t))-(\bar{\uppi}_{\cdot,1}(t),\dots,\bar{\uppi}_{\cdot,n}(t))\Big\rangle\\&\le \sqrt{\kappa_{i}}(\int_{0}^{1}w(z)\mathrm{d}z+\frac{1}{e})M\sum_{i\in\N}\sum_{j\in\N}\|\uppi_{i,j}(t)-\bar{\uppi}_{\cdot,j}(t)\|_{2}\\&\le\sqrt{\kappa_{i}}M^{2}n^{3/2}(\int_{0}^{1}w(z)\mathrm{d}z+\frac{1}{e})^{2}\sum_{m=1}^{t}\tau^{t-m}\eta_{m}.
			\end{aligned}
		\end{equation}
		As for \textcircled{2},
		we have, 
		\begin{equation}\label{equ:lemma_result_2}
			\begin{aligned}
				&\E\Big(\left\langle\sum_{i\in\N}\Big[\nabla  F_{t}^{a}(\uppi_{i,1}(t),\dots,\uppi_{i,1}(t))\Big]_{\uppi_{i}},\one_{\mathcal{A}_{t}^{*}}-(\uppi_{i,1}(t),\dots,\uppi_{i,n}(t))\right\rangle\Big)\\
				&=\E\left(\E\left(\left\langle\sum_{i\in\N}\Big[\nabla  F_{t}^{a}(\uppi_{i,1}(t),\dots,\uppi_{i,1}(t))\Big]_{\uppi_{i}},\one_{\mathcal{A}_{t}^{*}}-(\uppi_{i,1}(t),\dots,\uppi_{i,n}(t))\right\rangle\Big|\uppi_{i,j}(t),\forall i,j\in\N\right)\right)\\
				&=\E\left(\left\langle\E\left(\sum_{i\in\N}\Big[\nabla  F_{t}^{a}(\uppi_{i,1}(t),\dots,\uppi_{i,1}(t))\Big]_{\uppi_{i}}\Big|\uppi_{i,j}(t),\forall i,j\in\N\right),\one_{\mathcal{A}_{t}^{*}}-(\uppi_{i,1}(t),\dots,\uppi_{i,n}(t))\right\rangle\right)\\
				&=\E\Big(\Big\langle(\mathbf{d}_{1}(t),\dots,\mathbf{d}_{n}(t)),\one_{\mathcal{A}_{t}^{*}}-(\uppi_{i,1}(t),\dots,\uppi_{i,n}(t))\Big\rangle\Big)\\
				&=\underbrace{\E\Big(\Big\langle(\mathbf{d}_{1}(t),\dots,\mathbf{d}_{n}(t)),\one_{\mathcal{A}_{t}^{*}}-(\z_{i,1}(t),\dots,\z_{i,1}(t))\Big\rangle\Big)}_{\text{\textcircled{4}}}\\&+\underbrace{E\Big(\Big\langle(\mathbf{d}_{1}(t),\dots,\mathbf{d}_{n}(t)),(\z_{i,1}(t),\dots,\z_{i,n}(t))-(\uppi_{i,1}(t),\dots,\uppi_{i,n}(t))\Big\rangle\Big)}_{\text{\textcircled{5}}}.
			\end{aligned}
		\end{equation}
		For \textcircled{5}, we have that,
		\begin{equation}
			\label{equ:lemma_result_6}
			\begin{aligned}
				\text{\textcircled{5}}&\le\E\Big(\left\|(\mathbf{d}_{1}(t),\dots,\mathbf{d}_{n}(t))\right\|_{2}\left\|(\z_{i,1}(t),\dots,\z_{i,n}(t))-(\uppi_{i,1}(t),\dots,\uppi_{i,n}(t))\right\|_{2}\Big)
				\\
				&\le\sqrt{n\kappa_{i}}M(\int_{0}^{1}w(z)\mathrm{d}z+\frac{1}{e})\E\Big(\left\|(\z_{i,1}(t),\dots,\z_{i,n}(t))-(\uppi_{i,1}(t),\dots,\uppi_{i,n}(t))\right\|_{2}\Big)\\
				&\le\sqrt{n\kappa_{i}}M(\int_{0}^{1}w(z)\mathrm{d}z+\frac{1}{e})\E\Big(\left\|\left(\z_{i,1}(t)-\bar{\uppi}_{\cdot,1}(t),\dots,\z_{i,n}(t)-\bar{\uppi}_{\cdot,n}(t)\right)\right\|_{2}\Big)\\
				&+\sqrt{n\kappa_{i}}M(\int_{0}^{1}w(z)\mathrm{d}z+\frac{1}{e})\E\Big(\left\|\left(\uppi_{i,1}(t)-\bar{\uppi}_{\cdot,1}(t),\dots,\uppi_{i,n}(t)-\bar{\uppi}_{\cdot,n}(t)\right)\right\|_{2}\Big)\\
				&\le\sqrt{n\kappa_{i}}M(\int_{0}^{1}w(z)\mathrm{d}z+\frac{1}{e})\sum_{i\in\N}\sum_{j\in\N}\E\Big(\|\uppi_{i,j}(t)-\bar{\uppi}_{\cdot,j}(t)\|_{2}+\|\z_{i,j}(t)-\bar{\uppi}_{\cdot,j}(t)\|_{2}\Big)\\
				&\le2n^{2}\sqrt{\kappa_{i}}M^{2}(\int_{0}^{1}w(z)\mathrm{d}z+\frac{1}{e})^{2}\sum_{m=1}^{t}\tau^{t-m}\eta_{m}
			\end{aligned}
		\end{equation}%(\bar{\uppi}_{\cdot,1}(t),\dots,\bar{\uppi}_{\cdot,n}(t))-(\uppi_{i,1}(t),\dots,\uppi_{i,n}(t))
			For \textcircled{4}: At first, from the Line 19 of \cref{alg:SPL}, we know that, for any $\x\in\Delta_{\kappa_{i}}$, we have 
		 \begin{equation*}
			\|\uppi_{i,i}(t+1)-\x\|_{2}\le\|\z_{i,i}(t)+\eta_{t}\mathbf{d}_{i}(t)-\x\|_{2}.
		\end{equation*} 
	 If we slightly abuse the notation $\one_{(\mathcal{A}_{t}^{*}\cap\V_{i})}$ to denote $\kappa_{i}$-dimensional indicator vector over $(\mathcal{A}_{t}^{*}\cap\V_{i})$. Note that $\one_{\mathcal{A}_{t}^{*}}$ is a $\left(\sum_{i=1}^{n}\kappa_{i}\right)$-dimensional indicator vector over $\mathcal{A}_{t}^{*}$. Then, we have 
	 \begin{equation*}
	 	\begin{aligned}
	 	&\|\uppi_{i,i}(t+1)-\one_{(\mathcal{A}_{t}^{*}\cap\V_{i})}\|^{2}_{2}\\&\le\|\z_{i,i}(t)+\eta_{t}\mathbf{d}_{i}(t)-\one_{(\mathcal{A}_{t}^{*}\cap\V_{i})}\|_{2}^{2}\\
	 	&=\|\z_{i,i}(t)-\one_{(\mathcal{A}_{t}^{*}\cap\V_{i})}\|_{2}^{2}+2\eta_{t}\left\langle\mathbf{d}_{i}(t),\z_{i,i}(t)-\one_{(\mathcal{A}_{t}^{*}\cap\V_{i})}\right\rangle+\eta_{t}^{2}\|\mathbf{d}_{i}(t)\|_{2}^{2}.
	 	\end{aligned}
	\end{equation*} 
Therefore,
		\begin{equation*}
		\begin{aligned}
		&\left\langle\mathbf{d}_{i}(t),\one_{(\mathcal{A}_{t}^{*}\cap\V_{i})}-\z_{i,i}(t)\right\rangle\\&\le\frac{\eta_{t}}{2}\|\mathbf{d}_{i}(t)\|_{2}^{2}+\frac{1}{2\eta_{t}}\Big(\|\z_{i,i}(t)-\one_{(\mathcal{A}_{t}^{*}\cap\V_{i})}\|_{2}^{2}-\|\uppi_{i,i}(t+1)-\one_{(\mathcal{A}_{t}^{*}\cap\V_{i})}\|^{2}_{2}\Big)\\
		&\le\frac{\eta_{t}}{2}\sqrt{\kappa_{i}}(\int_{0}^{1}w(z)\mathrm{d}z+\frac{1}{e})M+\frac{1}{2\eta_{t}}\Big(\|\z_{i,i}(t)-\one_{(\mathcal{A}_{t}^{*}\cap\V_{i})}\|_{2}^{2}-\|\uppi_{i,i}(t+1)-\one_{(\mathcal{A}_{t}^{*}\cap\V_{i})}\|^{2}_{2}\Big),
		\end{aligned}	
		\end{equation*}where the final inequality follows from \cref{lemma:D1}.
		
	Note that
	\begin{equation}\label{equ:D_final1}
	\begin{aligned}
		&\E\Big(\Big\langle(\mathbf{d}_{1}(t),\dots,\mathbf{d}_{n}(t)),\one_{\mathcal{A}_{t}^{*}}-(\z_{i,1}(t),\dots,\z_{i,1}(t))\Big\rangle\Big)\\
		&=\sum_{i=1}^{n}\E\left(\left\langle\mathbf{d}_{i}(t),\one_{(\mathcal{A}_{t}^{*}\cap\V_{i})}-\z_{i,i}(t)\right\rangle\right)\\
		&\le\frac{\eta_{t}}{2}n\sqrt{\kappa_{i}}(\int_{0}^{1}w(z)\mathrm{d}z+\frac{1}{e})M+\frac{1}{2\eta_{t}}\sum_{i=1}^{n}\Big(\|\z_{i,i}(t)-\one_{(\mathcal{A}_{t}^{*}\cap\V_{i})}\|_{2}^{2}-\|\uppi_{i,i}(t+1)-\one_{(\mathcal{A}_{t}^{*}\cap\V_{i})}\|^{2}_{2}\Big).
	\end{aligned}
	\end{equation}
Merging Eq.\eqref{equ:D_final1}, Eq.\eqref{equ:lemma_result_6}, Eq.\eqref{equ:lemma_result_2},  Eq.\eqref{equ:lemma_result_3}, Eq.\eqref{equ:lemma_result_1} into Eq.\eqref{equ:lemma_result_0}, we finally have 
\begin{equation*}
\begin{aligned}
&\rho f_{t}(\mathcal{A}_{t}^{*})-F_{t}(\bar{\uppi}_{\cdot,1}(t),\dots,\bar{\uppi}_{\cdot,n}(t))\le5M^{2}n^{\frac{5}{2}}(\sum_{i=1}^{n}\kappa_{i})(\int_{0}^{1}w(z)\mathrm{d}z+\frac{1}{e})^{2}\sum_{m=1}^{t}\beta^{t-m}\eta_{m}\\
&+\frac{\eta_{t}}{2}n\sqrt{\kappa_{i}}(\int_{0}^{1}w(z)\mathrm{d}z+\frac{1}{e})M+\frac{1}{2\eta_{t}}\sum_{i=1}^{n}\Big(\|\z_{i,i}(t)-\one_{(\mathcal{A}_{t}^{*}\cap\V_{i})}\|_{2}^{2}-\|\uppi_{i,i}(t+1)-\one_{(\mathcal{A}_{t}^{*}\cap\V_{i})}\|^{2}_{2}\Big).
\end{aligned}
\end{equation*}
Therefore, we get \cref{lemma:D7}.
	\end{proof}
Next, we prove an upper bound of $\sum_{t=1}^{T}\frac{1}{2\eta_{t}}\sum_{i=1}^{n}\Big(\|\z_{i,i}(t)-\one_{(\mathcal{A}_{t}^{*}\cap\V_{i})}\|_{2}^{2}-\|\uppi_{i,i}(t+1)-\one_{(\mathcal{A}_{t}^{*}\cap\V_{i})}\|^{2}_{2}\Big)$, that is to say, 
\begin{lemma}\label{lemma:D8} Under Assumption \ref{ass:1}, we have that
	\begin{equation*}
		\sum_{t=1}^{T}\frac{1}{2\eta_{t}}\sum_{i=1}^{n}\Big(\|\z_{i,i}(t)-\one_{(\mathcal{A}_{t}^{*}\cap\V_{i})}\|_{2}^{2}-\|\uppi_{i,i}(t+1)-\one_{(\mathcal{A}_{t}^{*}\cap\V_{i})}\|^{2}_{2}\le \frac{4n}{\eta_{T+1}}+\sum_{t=1}^{T}\sum_{i=1}^{n}\frac{6}{\eta_{t+1}}\E\Big(\|\one_{(\mathcal{A}_{t}^{*}\cap\V_{i})}-\one_{(\mathcal{A}_{t+1}^{*}\cap\V_{i})}\|_{2}\Big),
	\end{equation*} where $\mathcal{A}_{t}^{*}$ is the optimal solution of problem~\eqref{equ_problem} and we slightly abuse the notation $\one_{(\mathcal{A}_{t}^{*}\cap\V_{i})}$ to denote $\kappa_{i}$-dimensional indicator vector over $(\mathcal{A}_{t}^{*}\cap\V_{i})$. Note that $\one_{\mathcal{A}_{t}^{*}}$ is a $\left(\sum_{i=1}^{n}\kappa_{i}\right)$-dimensional indicator vector over $\mathcal{A}_{t}^{*}$. 
\end{lemma}
\begin{proof}
	\begin{equation*}
		\begin{aligned}
			&\sum_{t=1}^{T}\frac{1}{2\eta_{t}}\sum_{i=1}^{n}\Big(\|\z_{i,i}(t)-\one_{(\mathcal{A}_{t}^{*}\cap\V_{i})}\|_{2}^{2}-\|\uppi_{i,i}(t+1)-\one_{(\mathcal{A}_{t}^{*}\cap\V_{i})}\|^{2}_{2}\Big)\\
			&=\underbrace{\sum_{t=1}^{T}\sum_{i=1}^{n}\Big(\frac{1}{\eta_{t}}\E(\|\z_{i,i}(t)-\one_{(\mathcal{A}_{t}^{*}\cap\V_{i})}\|_{2}^{2})-\frac{1}{\eta_{t+1}}\E(\|\z_{i,i}(t+1)-\one_{(\mathcal{A}_{t+1}^{*}\cap\V_{i})}\|_{2}^{2})\Big)}_{\text{\textcircled{1}}}\\
			&+\underbrace{\sum_{t=1}^{T}\sum_{i=1}^{n}\Bigg(\frac{1}{\eta_{t+1}}\E\Big(\|\z_{i,i}(t+1)-\one_{(\mathcal{A}_{t+1}^{*}\cap\V_{i})}\|_{2}^{2}-\|\z_{i,i}(t+1)-\one_{(\mathcal{A}_{t}^{*}\cap\V_{i})}\|_{2}^{2}\Big)\Bigg)}_{\text{\textcircled{2}}}\\
			&+\underbrace{\sum_{t=1}^{T}\sum_{i=1}^{n}\Bigg(\frac{1}{\eta_{t+1}}\E\Big(\|\z_{i,i}(t+1)-\one_{(\mathcal{A}_{t}^{*}\cap\V_{i})}\|_{2}^{2}-\|\uppi_{i,i}(t+1)-\one_{(\mathcal{A}_{t}^{*}\cap\V_{i})}\|_{2}^{2}\Big)\Bigg)}_{\text{\textcircled{3}}}\\
			&+\underbrace{\sum_{t=1}^{T}\sum_{i=1}^{n}\Big(\frac{1}{\eta_{t+1}}-\frac{1}{\eta_{t}}\Big)\E\Big(\|\uppi_{i,i}(t+1)-\one_{(\mathcal{A}_{t}^{*}\cap\V_{i})}\|_{2}^{2}\Big)}_{\text{\textcircled{4}}}.
		\end{aligned}
	\end{equation*}
	Firstly, we have \textcircled{1}$\le\frac{\sum_{i=1}^{n}\E(\|\z_{i,i}(1)-\one_{(\mathcal{A}_{1}^{*}\cap\V_{i})}\|_{2}^{2})}{\eta_{1}}\le\frac{4n}{\eta_{1}}$.
	
	Moreover, we have \textcircled{4}$\le 4n\Big(\frac{1}{\eta_{T+1}}-\frac{1}{\eta_{1}}\Big)$. 
	
	As for \textcircled{2}, we have 
\begin{equation*}
\begin{aligned}
&\sum_{t=1}^{T}\sum_{i=1}^{n}\Bigg(\frac{1}{\eta_{t+1}}\E\Big(\|\z_{i,i}(t+1)-\one_{(\mathcal{A}_{t+1}^{*}\cap\V_{i})}\|_{2}^{2}-\|\z_{i,i}(t+1)-\one_{(\mathcal{A}_{t}^{*}\cap\V_{i})}\|_{2}^{2}\Big)\Bigg)\\
&=\sum_{t=1}^{T}\sum_{i=1}^{n}\Bigg(\frac{1}{\eta_{t+1}}\E\Big(\|\one_{(\mathcal{A}_{t}^{*}\cap\V_{i})}-\one_{(\mathcal{A}_{t+1}^{*}\cap\V_{i})}\|_{2}^{2}+2\left\langle\z_{i,i}(t+1)-\one_{(\mathcal{A}_{t}^{*}\cap\V_{i})},\one_{(\mathcal{A}_{t}^{*}\cap\V_{i})}-\one_{(\mathcal{A}_{t+1}^{*}\cap\V_{i})}\right\rangle\Big)\Bigg)\\
&\le\sum_{t=1}^{T}\sum_{i=1}^{n}\frac{6}{\eta_{t+1}}\E\Big(\|\one_{(\mathcal{A}_{t}^{*}\cap\V_{i})}-\one_{(\mathcal{A}_{t+1}^{*}\cap\V_{i})}\|_{2}\Big),
\end{aligned}.
\end{equation*} where the final inequality from $\left\langle\z_{i,i}(t+1)-\one_{(\mathcal{A}_{t}^{*}\cap\V_{i})},\one_{(\mathcal{A}_{t}^{*}\cap\V_{i})}-\one_{(\mathcal{A}_{t+1}^{*}\cap\V_{i})}\right\rangle\le\|\z_{i,i}(t+1)-\one_{(\mathcal{A}_{t}^{*}\cap\V_{i})}\|_{2}*\|\one_{(\mathcal{A}_{t}^{*}\cap\V_{i})}-\one_{(\mathcal{A}_{t+1}^{*}\cap\V_{i})}\|_{2}\le2\|\one_{(\mathcal{A}_{t}^{*}\cap\V_{i})}-\one_{(\mathcal{A}_{t+1}^{*}\cap\V_{i})}\|_{2}$.
	
	Then, we have 
	\begin{equation*}
		\begin{aligned}
			&\text{\textcircled{3}}=\sum_{t=1}^{T}\sum_{i=1}^{n}\Bigg(\frac{1}{\eta_{t+1}}\E\Big(\|\z_{i,i}(t+1)-\one_{(\mathcal{A}_{t}^{*}\cap\V_{i})}\|_{2}^{2}-\|\uppi_{i,i}(t+1)-\one_{(\mathcal{A}_{t}^{*}\cap\V_{i})}\|_{2}^{2}\Big)\Bigg)\\
			&=\sum_{t=1}^{T}\sum_{i=1}^{n}\Bigg(\frac{1}{\eta_{t+1}}\E\Big(\|\sum_{j\in\N_{i}\cup\{i\}}w_{ij}\uppi_{j,i}(t+1)-\one_{(\mathcal{A}_{t}^{*}\cap\V_{i})}\|_{2}^{2}-\|\uppi_{i,i}(t+1)-\one_{(\mathcal{A}_{t}^{*}\cap\V_{i})}\|_{2}^{2}\Big)\Bigg)\\		
			&\le\sum_{t=1}^{T}\Bigg(\frac{1}{\eta_{t+1}}\E\Big(\sum_{i\in\N}\sum_{j\in\N_{i}\cup\{i\}}\Big(w_{ij}\|\uppi_{j,i}(t+1)-\one_{(\mathcal{A}_{t}^{*}\cap\V_{i})}\|_{2}^{2}\Big)-\sum_{i\in\N}\|\uppi_{i,i}(t+1)-\one_{(\mathcal{A}_{t}^{*}\cap\V_{i})}\|_{2}^{2}\Big)\Bigg)\\		
			&=\sum_{t=1}^{T}\Bigg(\frac{1}{\eta_{t+1}}\E\Big(\sum_{i\in\N}\Big((\sum_{j\in\N_{i}\cup\{i\}}w_{ji})\|\uppi_{j,i}(t+1)-\one_{(\mathcal{A}_{t}^{*}\cap\V_{i})}\|_{2}^{2}\Big)-\sum_{i\in\N}\|\uppi_{i,i}(t+1)-\one_{(\mathcal{A}_{t}^{*}\cap\V_{i})}\|_{2}^{2}\Big)\Bigg)\\		
			&=0,
		\end{aligned}
	\end{equation*} where the first inequality follows from the convexity of $\|\cdot\|^{2}_{2}$, the third equality is from $w_{ij}=w_{ji}$ and the final equality follows from \cref{ass:1}.
	
	We finally get 
	\begin{equation*}
		\sum_{t=1}^{T}\frac{1}{2\eta_{t}}\sum_{i=1}^{n}\Big(\|\z_{i,i}(t)-\one_{(\mathcal{A}_{t}^{*}\cap\V_{i})}\|_{2}^{2}-\|\uppi_{i,i}(t+1)-\one_{(\mathcal{A}_{t}^{*}\cap\V_{i})}\|^{2}_{2}\le \frac{4n}{\eta_{T+1}}+\sum_{t=1}^{T}\sum_{i=1}^{n}\frac{6}{\eta_{t+1}}\E\Big(\|\one_{(\mathcal{A}_{t}^{*}\cap\V_{i})}-\one_{(\mathcal{A}_{t+1}^{*}\cap\V_{i})}\|_{2}\Big).
	\end{equation*}
\end{proof}
As a result, we can show that 
\begin{lemma}\label{lemma:D9}
	Under Assumption \ref{ass:1}, when each set objective function $f_{t}$ is monotone submodular with curvature $c$, $\alpha$-weakly DR-submodular or  $(\gamma,\beta)$-weakly submodular, if we set the weight function $w(z)$ according to \cref{thm4}, we have that
\begin{equation*}
	\begin{aligned}
		&\rho\sum_{t=1}^{T}f_{t}(\mathcal{A}_{t}^{*})-\sum_{t=1}^{T}\E(f_{t}\left(\cup_{i=1}^{n}\{a_{i}(t)\}\right))\le6M^{2}n^{\frac{5}{2}}(\sum_{i=1}^{n}\kappa_{i})(\int_{0}^{1}w(z)\mathrm{d}z+\frac{1}{e})^{2}\sum_{t=1}^{T}\sum_{m=1}^{t}\beta^{t-m}\eta_{m}\\
		&+n\sqrt{\kappa_{i}}(\int_{0}^{1}w(z)\mathrm{d}z+\frac{1}{e})M\sum_{t=1}^{T}\frac{\eta_{t}}{2}+\frac{4n}{\eta_{T+1}}+\sum_{t=1}^{T}\sum_{i=1}^{n}\frac{6}{\eta_{t+1}}\E\Big(\|\one_{(\mathcal{A}_{t}^{*}\cap\V_{i})}-\one_{(\mathcal{A}_{t+1}^{*}\cap\V_{i})}\|_{2}\Big).
	\end{aligned}
\end{equation*}
\end{lemma}
\begin{proof}
	Merging \cref{lemma:D8} into \cref{lemma:D7}, we can get that
	\begin{equation*}
		\begin{aligned}
			&\rho\sum_{t=1}^{T}f_{t}(\mathcal{A}_{t}^{*})-\sum_{t=1}^{T}F_{t}(\bar{\uppi}_{\cdot,1}(t),\dots,\bar{\uppi}_{\cdot,n}(t))\le5M^{2}n^{\frac{5}{2}}(\sum_{i=1}^{n}\kappa_{i})(\int_{0}^{1}w(z)\mathrm{d}z+\frac{1}{e})^{2}\sum_{t=1}^{T}\sum_{m=1}^{t}\tau^{t-m}\eta_{m}\\
			&+n\sqrt{\kappa_{i}}(\int_{0}^{1}w(z)\mathrm{d}z+\frac{1}{e})M\sum_{t=1}^{T}\frac{\eta_{t}}{2}+\frac{4n}{\eta_{T+1}}+\sum_{t=1}^{T}\sum_{i=1}^{n}\frac{6}{\eta_{t+1}}\E\Big(\|\one_{(\mathcal{A}_{t}^{*}\cap\V_{i})}-\one_{(\mathcal{A}_{t+1}^{*}\cap\V_{i})}\|_{2}\Big).
		\end{aligned}
	\end{equation*}
	From \cref{thm1}, we also can show that $|F_{t}(\x)-F_{t}(\y)|\le nM\|\x-\y\|_{2}$ for any $t\in[T]$ such that we have
	\begin{equation*}
	\begin{aligned}
	&|F_{t}(\bar{\uppi}_{\cdot,1}(t),\dots,\bar{\uppi}_{\cdot,n}(t))-F_{t}(\uppi_{1,1}(t),\uppi_{2,2}(t),\dots,\uppi_{n,n}(t))|\\&
	\le nM\|(\bar{\uppi}_{\cdot,1}(t),\dots,\bar{\uppi}_{\cdot,n}(t))-(\uppi_{1,1}(t),\dots,\uppi_{n,n}(t))\|_{2}\\&= nM\sum_{i=1}^{n}\|\bar{\uppi}_{\cdot,i}(t)-\uppi_{i,i}(t)\|_{2}\le nM\sum_{j=1}^{n}\sum_{i=1}^{n}\|\bar{\uppi}_{\cdot,i}(t)-\uppi_{j,i}(t)\|_{2}\le n^{5/2}M^{2}(\int_{0}^{1}w(z)\mathrm{d}z+\frac{1}{e})\sum_{m=1}^{t}\tau^{t-m}\eta_{m},
	\end{aligned}
	\end{equation*} where the final inequality follows from \cref{lemma:D6}.
	
	Moreover, from the monotonicity of $F_{t}$(See \textbf{2)} in \cref{thm1}) and Lines 5-8, we also can infer that
	\begin{equation*}
		\E\left(f_{t}\left(\cup_{i=1}^{n}\{a_{i}(t)\}\right)\right)=F_{t}(\p_{1}(t),\dots,\p_{n}(t))\ge F_{t}(\uppi_{1,1}(t),\uppi_{2,2}(t),\dots,\uppi_{n,n}(t)),
	\end{equation*} where the final inequality comes from $\p_{i}(t) :=\frac{\uppi_{i,i}(t)}{\|\uppi_{i,i}(t)\|_{1}}\ge\uppi_{i,i}(t)$.
Therefore, we have that
\begin{equation*}
	\begin{aligned}
		&\rho\sum_{t=1}^{T}f_{t}(\mathcal{A}_{t}^{*})-\sum_{t=1}^{T}\E(f_{t}\left(\cup_{i=1}^{n}\{a_{i}(t)\}\right))\le6M^{2}n^{\frac{5}{2}}(\sum_{i=1}^{n}\kappa_{i})(\int_{0}^{1}w(z)\mathrm{d}z+\frac{1}{e})^{2}\sum_{t=1}^{T}\sum_{m=1}^{t}\tau^{t-m}\eta_{m}\\
		&+n\sqrt{\kappa_{i}}(\int_{0}^{1}w(z)\mathrm{d}z+\frac{1}{e})M\sum_{t=1}^{T}\frac{\eta_{t}}{2}+\frac{4n}{\eta_{T+1}}+\sum_{t=1}^{T}\sum_{i=1}^{n}\frac{6}{\eta_{t+1}}\E\Big(\|\one_{(\mathcal{A}_{t}^{*}\cap\V_{i})}-\one_{(\mathcal{A}_{t+1}^{*}\cap\V_{i})}\|_{2}\Big).
	\end{aligned}
\end{equation*}
\end{proof}
From \cref{lemma:D9}, if we set the step size $\eta_{t}=\mathcal{O}\big(\sqrt{\frac{(1-\tau)\mathcal{P}_{T}}{T}}\big)$ where $\mathcal{P}_{T}\triangleq\sum_{t=1}^{T}|\mathcal{A}_{t+1}^{*}\triangle\mathcal{A}_{t}^{*}|$ and $\triangle$ denotes the symmetric difference, we can show that
\begin{equation*}
\begin{aligned}
&\rho\sum_{t=1}^{T}F_{t}(\one_{\mathcal{A}_{t}^{*}})-\sum_{t=1}^{T}\E(f_{t}\left(\cup_{i=1}^{n}\{a_{i}(t)\}\right))\\&\le\mathcal{O}(\sum_{t=1}^{T}\sum_{m=1}^{t}\tau^{t-m}\eta_{m})+\mathcal{O}(\sum_{t=1}^{T}\eta_{t})+\mathcal{O}(\frac{1}{\eta_{T+1}})+\mathcal{O}\Big(\sum_{t=1}^{T}\sum_{i=1}^{n}\frac{1}{\eta_{t+1}}\E(\|\one_{(\mathcal{A}_{t}^{*}\cap\V_{i})}-\one_{(\mathcal{A}_{t+1}^{*}\cap\V_{i})}\|_{2})\Big)\\
&=\mathcal{O}\left(\sqrt{\frac{\mathcal{P}_{T}T}{1-\tau}}\right),
\end{aligned}	
\end{equation*} where the final equality follows from $\mathcal{P}_{T}\triangleq\sum_{t=1}^{T}\sum_{i=1}^{n}\|\one_{(\mathcal{A}_{t}^{*}\cap\V_{i})}-\one_{(\mathcal{A}_{t+1}^{*}\cap\V_{i})}\|_{2}\triangleq\sum_{t=1}^{T}\sum_{i=1}^{n}\|\one_{(\mathcal{A}_{t}^{*}\cap\V_{i})}-\one_{(\mathcal{A}_{t+1}^{*}\cap\V_{i})}\|_{1}\triangleq\sum_{t=1}^{T}|\mathcal{A}_{t+1}^{*}\triangle\mathcal{A}_{t}^{*}|$.
\section{Convergence  Analysis of \texorpdfstring{\texttt{MA-MPL}}{} Algorithm}\label{appendix_proof_thm2}
In this section, we prove \cref{thm:result2}.

At first, we verify that each policy vector $(\uppi_{i,1}^{(k)}(t),\dots,\uppi_{i,n}^{(k)}(t))\in\prod_{i=1}\Delta_{\kappa_{i}}$, namely, we have the following lemma:
\begin{lemma}\label{lemma:E1}
When the communication graph $G(\V,\mathcal{E})$ is connected, we can show that $(\uppi_{i,1}^{(k)}(t),\dots,\uppi_{i,n}^{(k)}(t))\in\prod_{i=1}\Delta_{\kappa_{i}}$ for any $i\in\N$ and $0\le k\le K$ as well as $t\in[T]$.
\end{lemma}
\begin{proof}
From the Line 1 of \cref{alg:MPL}, we can know that $(\uppi^{(0)}_{i,1}(t),\dots,\uppi^{(0)}_{i,n}(t)):=\boldsymbol{0},\forall t\in[T]$ such that $(\uppi^{(0)}_{i,1}(t),\dots,\uppi^{(0)}_{i,n}(t))\in\prod_{i=1}\Delta_{\kappa_{i}}$. Moreover, from Lines 8, for a fixed $t\in[T]$, we know that only agent $i\in\N$ can change the policy $\uppi_{i,i}^{(k)}(t)$. Furthermore, from Line 10, we also can infer that, for other agent $j\neq i$, its policy $\uppi_{j,i}^{(k)}(t)$ is copy of some past iteration of $\uppi_{i,i}^{(k)}(t)$, namely, there exists a $k_{1}\le k$ such that $\uppi_{j,i}^{(k)}(t)=\uppi_{i,i}^{(k_{1})}(t)$. According to the Line 8 and Line 10, we have $\uppi_{i,i}^{(k)}(t)=\uppi_{i,i}^{(k-1)}(t)+\frac{1}{K}\mathbf{v}_{i}^{(k)}(t)$, so we have  $\uppi_{i,i}^{(k)}(t)=\frac{1}{K}\sum_{m=1}^{k}\mathbf{v}_{i}^{(m)}(t)\in\Delta_{\kappa_{i}}$ for any $k\in[K]$ ,$t\in[T]$ and $i\in\N$. Therefore, we  have $(\uppi_{i,1}^{(k)}(t),\dots,\uppi_{i,n}^{(k)}(t))\in\prod_{i=1}\Delta_{\kappa_{i}}$ for any $i\in\N$ and $k\in[K]$ as well as $t\in[T]$.
\end{proof}
Note that, from Lines 11-13, each agent $i\in\N$ only uses the policy $\uppi_{i,i}^{(K)}(t)$ to make decision, so the real policy vector taken by all agents is $(\uppi_{1,1}^{(K)}(t),\dots,\uppi_{n,n}^{(K)}(t))$. Motivated by this finding, we next investigate some relationships between the policy vector $(\uppi_{1,1}^{(k)}(t),\dots,\uppi_{n,n}^{(k)}(t))$ and $(\uppi_{i,1}^{(k)}(t),\dots,\uppi_{i,n}^{(k)}(t))$. That is, 
\begin{lemma}\label{lemma:E2}
When the communication graph $G(\V,\mathcal{E})$ is connected, if we set $\kappa\triangleq\sum_{i=1}^{n}\kappa_{i}$ and utilize the symbol $\one_{\kappa}$ to denote the $\kappa$-dimensional vector whose all elements are $1$, we can show that
\begin{equation*}
0\le\frac{1}{n}\left\langle\one_{\kappa},(\uppi_{1,1}^{(k)}(t),\dots,\uppi_{n,n}^{(k)}(t))-(\uppi_{i,1}^{(k)}(t),\dots,\uppi_{i,n}^{(k)}(t))\right\rangle\le\frac{d(G)}{K},
\end{equation*} where $d(G)$ is the diameter of a graph, that is to say, the length of the shortest path between the most distanced nodes.
\end{lemma}
\begin{proof}
From Lines 6-10, we can show that, for fixed $k\in[K]$ and $t\in[T]$, each $\uppi_{i,j}^{(k)}(t)=\uppi_{j,j}^{(k_{j})}(t)$ where $k_{j}\in[k-d(G),k]$. Moreover, according to the Line 8 and Line 10, we have $\uppi_{j,j}^{(k)}(t)=\uppi_{j,j}^{(k-1)}(t)+\frac{1}{K}\mathbf{v}_{j}^{(k)}(t)$, so we have  $\uppi_{j,j}^{(k)}(t)=\frac{1}{K}\sum_{m=1}^{k}\mathbf{v}_{j}^{(m)}(t)$ for any $k\in[K]$ ,$t\in[T]$ and $i\in\N$. Therefore, we have that $\uppi_{j,j}^{(k)}(t)-\uppi_{i,j}^{(k)}(t)=\frac{1}{K}\sum_{m=k_{j}}^{j}\mathbf{v}_{j}^{(m)}(t)$ such that $0\le\left\langle\one_{\kappa_{j}},\uppi_{j,j}^{(k)}(t)-\uppi_{i,j}^{(k)}(t)\right\rangle\le \frac{k-k_{j}}{K}\le\frac{d(G)}{K}$. So we have the result of \cref{lemma:E2}.
\end{proof}

Next, we investigate the relationships between  $(\uppi_{1,1}^{(k)}(t),\dots,\uppi_{n,n}^{(k)}(t))$  and  $(\uppi_{1,1}^{(k-1)}(t),\dots,\uppi_{n,n}^{(k-1)}(t))$.

\begin{lemma}\label{lemma:E3}
	When the communication graph $G(\V,\mathcal{E})$ is connected, we can show that
	\begin{equation*}
		\begin{aligned}
		&(\uppi_{1,1}^{(k)}(t),\dots,\uppi_{n,n}^{(k)}(t))-(\uppi_{1,1}^{(k-1)}(t),\dots,\uppi_{n,n}^{(k-1)}(t))=\frac{1}{K}(\mathbf{v}_{1}^{(k)}(t),\dots,(\mathbf{v}_{n}^{(k)}(t));\\
		&0\le\frac{1}{n}\left\langle\one_{\kappa},(\uppi_{1,1}^{(k)}(t),\dots,\uppi_{n,n}^{(k)}(t))-(\uppi_{1,1}^{(k-1)}(t),\dots,\uppi_{n,n}^{(k-1)}(t))\right\rangle\le\frac{1}{K},
		\end{aligned}
	\end{equation*} where $\kappa\triangleq\sum_{i=1}^{n}\kappa_{i}$ and $d(G)$ is the diameter of a graph, i.e., the length of the shortest path between the most distanced nodes.
\end{lemma}
\begin{proof}
	According to the Line 8 and Line 10, we have $\uppi_{j,j}^{(k)}(t)=\uppi_{j,j}^{(k-1)}(t)+\frac{1}{K}\mathbf{v}_{j}^{(k)}(t)$. Therefore, we get result of \cref{lemma:E3}.
\end{proof}
With \cref{lemma:E2} and \cref{lemma:E3}, we next give an upper bound about the error between the gradient $\nabla F_{t}(\uppi_{1,1}^{(k)}(t),\dots,\uppi_{n,n}^{(k)}(t))$ and $\nabla F_{t}(\uppi_{i,1}^{(k)}(t),\dots,\uppi_{i,n}^{(k)}(t))$.
\begin{lemma}\label{lemma:E4}
	When the communication graph $G(\V,\mathcal{E})$ is connected, if we set $M$ as the maximum marginal contribution of each monotone set function $f_{t}$, namely, $M\triangleq\max_{S\subseteq\V,e\in\V\setminus S,t\in[T]}\left(f_{t}(e|S)\right)$, we can show that
	\begin{equation*}
		\begin{aligned}
		 \left\|\nabla F_{t}(\uppi_{1,1}^{(k)}(t),\dots,\uppi_{n,n}^{(k)}(t))-\nabla F_{t}(\uppi_{i,1}^{(k)}(t),\dots,\uppi_{i,n}^{(k)}(t))\right\|_{2}\le\frac{\sqrt{\kappa}nMd(G)}{K}
		\end{aligned}
	\end{equation*} where $\kappa=\sum_{i=1}^{n}\kappa_{i}$ and $d(G)$ is the diameter of a graph, i.e., the length of the shortest path between the most distanced nodes.
\end{lemma}
\begin{proof}
At first, we investigate the error  between the first-order derivative $\frac{\partial F_{t}}{\partial\pi_{i,m}}(\uppi_{1,1}^{(k)}(t),\dots,\uppi_{n,n}^{(k)}(t))$ and the first-order derivative $\frac{\partial F_{t}}{\partial\pi_{i,m}}(\uppi_{i,1}^{(k)}(t),\dots,\uppi_{i,n}^{(k)}(t))$ for some $i\in\N$ and $m\in[\kappa_{i}]$.
\begin{equation*}
\begin{aligned}
&\Big|\frac{\partial F_{t}}{\partial\pi_{i,m}}(\uppi_{1,1}^{(k)}(t),\dots,\uppi_{n,n}^{(k)}(t))-\frac{\partial F_{t}}{\partial\pi_{i,m}}(\uppi_{i,1}^{(k)}(t),\dots,\uppi_{i,n}^{(k)}(t))\Big|\\
&=\Big|\left\langle\left(\int_{0}^{1}\frac{\partial^{2} F_{t}}{\partial\pi_{i,m}\partial\pi_{1,1}},\dots,\frac{\partial^{2} F_{t}}{\partial\pi_{i,m}\partial\pi_{n,\kappa_{n}}}\right)(\x),\left(\uppi_{1,1}^{(k)}(t)-\uppi_{i,1}^{(k)}(t),\dots,\uppi_{n,n}^{(k)}(t)-\uppi_{i,n}^{(k)}(t)\right)\right\rangle\Big|\\
&\le M\Big|\left\langle\one_{\kappa},\left(\uppi_{1,1}^{(k)}(t)-\uppi_{i,1}^{(k)}(t),\dots,\uppi_{n,n}^{(k)}(t)-\uppi_{i,n}^{(k)}(t)\right)\right\rangle\Big|\le\frac{nMd(G)}{K},
 \end{aligned}
\end{equation*} where the first equality comes from $\x=\left(\lambda\uppi_{1,1}^{(k)}(t)+(1-\lambda)\uppi_{i,1}^{(k)}(t),\dots,\lambda\uppi_{n,n}^{(k)}(t)+(1-\lambda)\uppi_{i,n}^{(k)}(t)\right)$ for some $\lambda\in[0,1]$, the first inequality follows from \cref{remak:D2} and the final inequality comes from \cref{lemma:E2}.
Then, we have that 
\begin{equation*}
\begin{aligned}
&\left\|\nabla F_{t}(\uppi_{1,1}^{(k)}(t),\dots,\uppi_{n,n}^{(k)}(t))-\nabla F_{t}(\uppi_{i,1}^{(k)}(t),\dots,\uppi_{i,n}^{(k)}(t))\right\|_{2}\\
&\le\sqrt{\kappa}\max_{\pi_{i,m}}\Big|\frac{\partial F_{t}}{\partial\pi_{i,m}}(\uppi_{1,1}^{(k)}(t),\dots,\uppi_{n,n}^{(k)}(t))-\frac{\partial F_{t}}{\partial\pi_{i,m}}(\uppi_{i,1}^{(k)}(t),\dots,\uppi_{i,n}^{(k)}(t))\Big|\le\frac{\sqrt{\kappa}nMd(G)}{K},
\end{aligned}
\end{equation*} where  $\kappa=\sum_{i=1}^{n}\kappa_{i}$ and $d(G)$ is the diameter of a graph, i.e., the length of the shortest path between the most distanced nodes.
\end{proof}

Now, we prove \cref{thm:result2}.
\begin{lemma}\label{lemma:E5}
	When the communication graph $G(\V,\mathcal{E})$ is connected, if we set $M$ as the maximum marginal contribution of each monotone set function $f_{t}$, namely, $M\triangleq\max_{S\subseteq\V,e\in\V\setminus S,t\in[T]}\left(f_{t}(e|S)\right)$, we can show that, 
	
	$\textbf{i):)}$ if $f_{t}$ is a monotone $\alpha$-weakly submodular function: \begin{equation*}
		\begin{aligned}
			&(1-e^{-\alpha})\sum_{t=1}^{T}f_{t}(\mathcal{A}_{t}^{*})-\sum_{t=1}^{T}\E\left(F_{t}(\uppi_{1,1}^{(K)}(t),\dots,\uppi_{n,n}^{(k)}(t))\right)\\&\le\left(\frac{nM(\sum_{i=1}^{n}\kappa_{i})}{2\alpha }+\frac{2\sqrt{\kappa}n^{2}Md(G)}{\alpha }\right)\frac{T}{K}+\frac{2nM\sqrt{\kappa}}{\alpha}\frac{T}{\sqrt{L}}\\
			&+\frac{1}{K}\sum_{k=1}^{K}\sum_{i=1}^{n}\sum_{t=1}^{T}\left\langle\mathbf{d}_{i}^{(k)}(t),\one_{(\mathcal{A}_{t}^{*}\cap\V_{i})}-\mathbf{v}_{i}^{(k)}(t)\right\rangle;
		\end{aligned}
	\end{equation*} 
	$\textbf{ii):)}$ if $f_{t}$ is a monotone $(\gamma,\beta)$-weakly submodular: \begin{equation*}
		\begin{aligned}
			&\frac{\gamma^{2}(1-e^{-\phi(\gamma,\beta)})}{\phi(\gamma,\beta)}\sum_{t=1}^{T}f_{t}(\mathcal{A}_{t}^{*})-\sum_{t=1}^{T}\E\left(F_{t}(\uppi_{1,1}^{(K)}(t),\dots,\uppi_{n,n}^{(k)}(t))\right)\\&\le\left(\frac{nM(\sum_{i=1}^{n}\kappa_{i})}{2\phi(\gamma,\beta) }+\frac{2\sqrt{\kappa}n^{2}Md(G)}{\phi(\gamma,\beta) }\right)\frac{T}{K}+\frac{2nM\sqrt{\kappa}}{\phi(\gamma,\beta)}\frac{T}{\sqrt{L}}\\
			&-\frac{1}{K}\sum_{k=1}^{K}\sum_{i=1}^{n}\sum_{t=1}^{T}\left\langle\mathbf{d}_{i}^{(k)}(t),\one_{(\mathcal{A}_{t}^{*}\cap\V_{i})}-\mathbf{v}_{i}^{(k)}(t)\right\rangle,
		\end{aligned}
	\end{equation*} where where we slightly abuse the notation $\one_{(\mathcal{A}_{t}^{*}\cap\V_{i})}$ to denote $\kappa_{i}$-dimensional indicator vector over $(\mathcal{A}_{t}^{*}\cap\V_{i})$. Note that $\one_{\mathcal{A}_{t}^{*}}$ is a $\left(\sum_{i=1}^{n}\kappa_{i}\right)$-dimensional indicator vector over $\mathcal{A}_{t}^{*}$ and $\phi(\gamma,\beta)=\beta(1-\gamma)+\gamma^{2}$.
Note the $\sum_{i=1}^{n}\sum_{t=1}^{T}\left\langle\mathbf{d}_{i}^{(k)}(t),\one_{(\mathcal{A}_{t}^{*}\cap\V_{i})}-\mathbf{v}_{i}^{(k)}(t)\right\rangle$ is the dynamic regret of linear maximization oracle $Q_{i}^{(k)}$ over the competitive sequence $\left(\one_{(\mathcal{A}_{1}^{*}\cap\V_{i})},\dots,\one_{(\mathcal{A}_{T}^{*}\cap\V_{i})}\right)$.  Therefore, we also rewrite the previous two results as:

$\textbf{i):)}$ if $f_{t}$ is a monotone $\alpha$-weakly submodular function: \begin{equation*}
	\begin{aligned}
		&(1-e^{-\alpha})\sum_{t=1}^{T}f_{t}(\mathcal{A}_{t}^{*})-\sum_{t=1}^{T}\E\left(F_{t}(\uppi_{1,1}^{(K)}(t),\dots,\uppi_{n,n}^{(k)}(t))\right)\\&\le\left(\frac{nM(\sum_{i=1}^{n}\kappa_{i})}{2\alpha }+\frac{2\sqrt{\kappa}n^{2}Md(G)}{\alpha }\right)\frac{T}{K}+\frac{2nM\sqrt{\kappa}}{\alpha}\frac{T}{\sqrt{L}}\\
		&+\frac{1}{K}\sum_{k=1}^{K}\sum_{i=1}^{n}\mathcal{R}^{Q_{i}^{(k)}}\left(\one_{(\mathcal{A}_{1}^{*}\cap\V_{i})},\dots,\one_{(\mathcal{A}_{T}^{*}\cap\V_{i})}\right);
	\end{aligned}
\end{equation*} 
$\textbf{ii):)}$ if $f_{t}$ is a monotone $(\gamma,\beta)$-weakly submodular: \begin{equation*}
	\begin{aligned}
		&\frac{\gamma^{2}(1-e^{-\phi(\gamma,\beta)})}{\phi(\gamma,\beta)}\sum_{t=1}^{T}f_{t}(\mathcal{A}_{t}^{*})-\sum_{t=1}^{T}\E\left(F_{t}(\uppi_{1,1}^{(K)}(t),\dots,\uppi_{n,n}^{(k)}(t))\right)\\&\le\left(\frac{nM(\sum_{i=1}^{n}\kappa_{i})}{2\phi(\gamma,\beta) }+\frac{2\sqrt{\kappa}n^{2}Md(G)}{\phi(\gamma,\beta) }\right)\frac{T}{K}+\frac{2nM\sqrt{\kappa}}{\phi(\gamma,\beta)}\frac{T}{\sqrt{L}}\\
		&-\frac{1}{K}\sum_{k=1}^{K}\mathcal{R}^{Q_{i}^{(k)}}\left(\one_{(\mathcal{A}_{1}^{*}\cap\V_{i})},\dots,\one_{(\mathcal{A}_{T}^{*}\cap\V_{i})}\right),
	\end{aligned}
\end{equation*} where $\mathcal{R}^{Q_{i}^{(k)}}\left(\one_{(\mathcal{A}_{1}^{*}\cap\V_{i})},\dots,\one_{(\mathcal{A}_{T}^{*}\cap\V_{i})}\right)$ is the dynamic regret of linear maximization oracle $Q_{i}^{(k)}$ over the competitive sequence $\left(\one_{(\mathcal{A}_{1}^{*}\cap\V_{i})},\dots,\one_{(\mathcal{A}_{T}^{*}\cap\V_{i})}\right)$.
\end{lemma}
\begin{proof}
\begin{equation*}
\begin{aligned}
&F_{t}(\uppi_{1,1}^{(k)}(t),\dots,\uppi_{n,n}^{(k)}(t))-F_{t}(\uppi_{1,1}^{(k-1)}(t),\dots,\uppi_{n,n}^{(k-1)}(t))\\
&\ge\left\langle\nabla F_{t}(\uppi_{1,1}^{(k-1)}(t),\dots,\uppi_{n,n}^{(k-1)}(t)), (\uppi_{1,1}^{(k)}(t)-\uppi_{1,1}^{(k-1)}(t),\dots,(\uppi_{1,1}^{(k)}(t)-\uppi_{n,n}^{(k-1)}(t))\right\rangle\\
&-\frac{M(\sum_{i=1}^{n}\kappa_{i})}{2}\left\|\left(\uppi_{1,1}^{(k)}(t)-\uppi_{1,1}^{(k-1)}(t),\dots,(\uppi_{1,1}^{(k)}(t)-\uppi_{n,n}^{(k-1)}(t)\right)\right\|_{2}^{2}\\
&\ge\frac{1}{K}\left\langle\nabla F_{t}(\uppi_{1,1}^{(k-1)}(t),\dots,\uppi_{n,n}^{(k-1)}(t)),(\mathbf{v}_{1}^{(k)}(t),\dots,\mathbf{v}_{n}^{(k)}(t))\right\rangle-\frac{nM(\sum_{i=1}^{n}\kappa_{i})}{2K^{2}}\\
&=\frac{1}{K}\left\langle\nabla F_{t}(\uppi_{1,1}^{(k-1)}(t),\dots,\uppi_{n,n}^{(k-1)}(t)),\one_{\mathcal{A}_{t}^{*}}\right\rangle-\frac{nM(\sum_{i=1}^{n}\kappa_{i})}{2K^{2}}\\
&+\frac{1}{K}\left\langle\nabla F_{t}(\uppi_{1,1}^{(k-1)}(t),\dots,\uppi_{n,n}^{(k-1)}(t)),(\mathbf{v}_{1}^{(k)}(t),\dots,\mathbf{v}_{n}^{(k)}(t))-\one_{\mathcal{A}_{t}^{*}}\right\rangle\\
&=\frac{1}{K}\underbrace{\left\langle\nabla F_{t}(\uppi_{1,1}^{(k-1)}(t),\dots,\uppi_{n,n}^{(k-1)}(t)),\one_{\mathcal{A}_{t}^{*}}\right\rangle}_{\text{\textcircled{1}}}-\frac{nM(\sum_{i=1}^{n}\kappa_{i})}{2K^{2}}\\
&+\frac{1}{K}\underbrace{\left\langle\nabla F_{t}(\uppi_{1,1}^{(k-1)}(t),\dots,\uppi_{n,n}^{(k-1)}(t))-\sum_{i\in\N}\Big[\nabla F_{t}(\uppi_{i,1}^{(k-1)}(t),\dots,\uppi_{i,n}^{(k-1)}(t))\Big]_{\uppi_{i}},(\mathbf{v}_{1}^{(k)}(t),\dots,\mathbf{v}_{n}^{(k)}(t))-\one_{\mathcal{A}_{t}^{*}}\right\rangle}_{\text{\textcircled{2}}}\\
&+\frac{1}{K}\underbrace{\left\langle\sum_{i\in\N}\Big[\nabla F_{t}(\uppi_{i,1}^{(k-1)}(t),\dots,\uppi_{i,n}^{(k-1)}(t))\Big]_{\uppi_{i}}-\left(\mathbf{d}_{1}^{(k)}(t),\dots,\mathbf{d}_{n}^{(k)}(t)\right),(\mathbf{v}_{1}^{(k)}(t),\dots,\mathbf{v}_{n}^{(k)}(t))-\one_{\mathcal{A}_{t}^{*}}\right\rangle}_{\text{\textcircled{3}}}\\
&+\frac{1}{K}\left\langle\left(\mathbf{d}_{1}^{(k)}(t),\dots,\mathbf{d}_{n}^{(k)}(t)\right),(\mathbf{v}_{1}^{(k)}(t),\dots,\mathbf{v}_{n}^{(k)}(t))-\one_{\mathcal{A}_{t}^{*}}\right\rangle,
\end{aligned}
\end{equation*} where the first inequality follows the \cref{lemma:D4}, namely, our proposed policy-based continuous extension $F_{t}$ is $\left(M\sum_{i=1}^{n}\kappa_{i}\right)$-smooth andthe symbol $\Big[\nabla  F_{t}^{a}(\uppi_{i,1}(t),\dots,\uppi_{i,1}(t))\Big]_{\uppi_{i}}$ is the projection over the policy $\uppi_{i}$, namely, $\Big[\nabla  F_{t}^{a}(\uppi_{i,1}(t),\dots,\uppi_{i,1}(t))\Big]_{\uppi_{i}}$ represents a $\left(\sum_{i=1}^{n}\kappa_{i}\right)$-dimensional vector that only keeps the first-order derivative at variable $\pi_{i,m},\forall m\in[\kappa_{i}]$ and set other coordinates to $0$, that is to say, 
\begin{equation*}
\Big[\nabla  F_{t}^{a}(\uppi_{i,1}(t),\dots,\uppi_{i,1}(t))\Big]_{\uppi_{i}}\triangleq\Bigg(\dots,0,\dots,\underbrace{\frac{\partial F_{t}}{\partial\pi_{i,1}}(\x),\dots,\frac{\partial F_{t}}{\partial\pi_{i,\kappa_{i}}}(\x)}_{\kappa_{i}},\dots,0,\dots\Bigg),
\end{equation*} where $\x\triangleq(\uppi_{i,1}(t),\dots,\uppi_{i,1}(t))$.

For \textcircled{1}, according to $\textbf{4)}$ of \cref{thm1}, when $f_{t}$ is monotone $\alpha$-weakly DR-submodular, 
\begin{equation*}
\left\langle\nabla F_{t}(\uppi_{1,1}^{(k-1)}(t),\dots,\uppi_{n,n}^{(k-1)}(t)),\one_{\mathcal{A}_{t}^{*}}\right\rangle\ge\alpha\left(f_{t}(\mathcal{A}_{t}^{*})-F_{t}(\uppi_{1,1}^{(k-1)}(t),\dots,\uppi_{n,n}^{(k-1)}(t))\right);
\end{equation*}
As for the setting that $f_{t}$ is monotone $(\gamma,\beta)$-weakly submodular, we also have 
\begin{equation*}
	\left\langle\nabla F_{t}(\uppi_{1,1}^{(k-1)}(t),\dots,\uppi_{n,n}^{(k-1)}(t)),\one_{\mathcal{A}_{t}^{*}}\right\rangle\ge\gamma^{2}f_{t}(\mathcal{A}_{t}^{*})-(\beta(1-\gamma)+\gamma^{2})F_{t}\left(\uppi_{1,1}^{(k-1)}(t),\dots,\uppi_{n,n}^{(k-1)}(t)\right).
\end{equation*}

For \textcircled{2}, according to \cref{lemma:E4}, we have
\begin{equation*}
\begin{aligned}
	&\Big|\left\langle\nabla F_{t}(\uppi_{1,1}^{(k-1)}(t),\dots,\uppi_{n,n}^{(k-1)}(t))-\sum_{i\in\N}\Big[\nabla F_{t}(\uppi_{i,1}^{(k-1)}(t),\dots,\uppi_{i,n}^{(k-1)}(t))\Big]_{\uppi_{i}},(\mathbf{v}_{1}^{(k)}(t),\dots,\mathbf{v}_{n}^{(k)}(t))-\one_{\mathcal{A}_{t}^{*}}\right\rangle\Big|\\
&=\Big|\left\langle\sum_{i\in\N}\Big[\nabla F_{t}(\uppi_{1,1}^{(k-1)}(t),\dots,\uppi_{n,n}^{(k-1)}(t))\Big]_{\uppi_{i}}-\sum_{i\in\N}\Big[\nabla F_{t}(\uppi_{i,1}^{(k-1)}(t),\dots,\uppi_{i,n}^{(k-1)}(t))\Big]_{\uppi_{i}},(\mathbf{v}_{1}^{(k)}(t),\dots,\mathbf{v}_{n}^{(k)}(t))-\one_{\mathcal{A}_{t}^{*}}\right\rangle\Big|\\
&\le\left\|\sum_{i\in\N}\Big[\nabla F_{t}(\uppi_{1,1}^{(k-1)}(t),\dots,\uppi_{n,n}^{(k-1)}(t))\Big]_{\uppi_{i}}-\sum_{i\in\N}\Big[\nabla F_{t}(\uppi_{i,1}^{(k-1)}(t),\dots,\uppi_{i,n}^{(k-1)}(t))\Big]_{\uppi_{i}}\right\|_{2}\left\|(\mathbf{v}_{1}^{(k)}(t),\dots,\mathbf{v}_{n}^{(k)}(t))-\one_{\mathcal{A}_{t}^{*}}\right\|_{2}\\
&\le2n\sum_{i\in\N}\left\|\Big[\nabla F_{t}(\uppi_{1,1}^{(k-1)}(t),\dots,\uppi_{n,n}^{(k-1)}(t))\Big]_{\uppi_{i}}-\Big[\nabla F_{t}(\uppi_{i,1}^{(k-1)}(t),\dots,\uppi_{i,n}^{(k-1)}(t))\Big]_{\uppi_{i}}\right\|_{2}\\
&\le2n\sum_{i\in\N}\left\|\nabla F_{t}(\uppi_{1,1}^{(k-1)}(t),\dots,\uppi_{n,n}^{(k-1)}(t))-\nabla F_{t}(\uppi_{i,1}^{(k-1)}(t),\dots,\uppi_{i,n}^{(k-1)}(t))\right\|_{2}\le\frac{2\sqrt{\kappa}n^{2}Md(G)}{K}.
\end{aligned}
\end{equation*}

As for \textcircled{3},
\begin{equation*}
	\begin{aligned}
		&\E\Big|\left\langle\sum_{i\in\N}\Big[\nabla F_{t}(\uppi_{i,1}^{(k-1)}(t),\dots,\uppi_{i,n}^{(k-1)}(t))\Big]_{\uppi_{i}}-\left(\mathbf{d}_{1}^{(k)}(t),\dots,\mathbf{d}_{n}^{(k)}(t)\right),(\mathbf{v}_{1}^{(k)}(t),\dots,\mathbf{v}_{n}^{(k)}(t))-\one_{\mathcal{A}_{t}^{*}}\right\rangle\Big|\\
		&\le\E\left(\left\|\sum_{i\in\N}\Big[\nabla F_{t}(\uppi_{i,1}^{(k-1)}(t),\dots,\uppi_{i,n}^{(k-1)}(t))\Big]_{\uppi_{i}}-\left(\mathbf{d}_{1}^{(k)}(t),\dots,\mathbf{d}_{n}^{(k)}(t)\right)\right\|_{2}\left\|(\mathbf{v}_{1}^{(k)}(t),\dots,\mathbf{v}_{n}^{(k)}(t))-\one_{\mathcal{A}_{t}^{*}}\right\|_{2}\right)\\
		&\le2n\E\left(\left\|\sum_{i\in\N}\Big[\nabla F_{t}(\uppi_{i,1}^{(k-1)}(t),\dots,\uppi_{i,n}^{(k-1)}(t))\Big]_{\uppi_{i}}-\left(\mathbf{d}_{1}^{(k)}(t),\dots,\mathbf{d}_{n}^{(k)}(t)\right)\right\|_{2}\right)\\
	    &\le2n\sqrt{\E\left(\frac{1}{L}\left\|\left(\mathbf{g}_{1}^{(k)}(t),\dots,\mathbf{g}_{n}^{(k)}(t)\right)\right\|_{2}\right)}\\
		&\le2n\sqrt{\frac{\kappa M^{2}}{L}},
	\end{aligned}
\end{equation*} where the third inequality follows Lines 15-20, namely, $\left(\mathbf{d}_{1}^{(k)}(t),\dots,\mathbf{d}_{n}^{(k)}(t)\right)$ is $L$-batch stochastic gradient for $\sum_{i\in\N}\Big[\nabla F_{t}(\uppi_{i,1}^{(k-1)}(t),\dots,\uppi_{i,n}^{(k-1)}(t))\Big]_{\uppi_{i}}$ and the final inequality follows from \cref{remak:D2} and $\kappa=\sum_{i=1}^{n}\kappa_{i}$.

Therefore, when $f_{t}$ is monotone  $\alpha$-weakly DR-submodular, we have that
\begin{equation*}
\begin{aligned}
&\E\left(F_{t}(\uppi_{1,1}^{(k)}(t),\dots,\uppi_{n,n}^{(k)}(t))-F_{t}(\uppi_{1,1}^{(k-1)}(t),\dots,\uppi_{n,n}^{(k-1)}(t))\right)\\
&\ge\frac{\alpha}{K}\E\left(f_{t}(\mathcal{A}_{t}^{*})-F_{t}(\uppi_{1,1}^{(k-1)}(t),\dots,\uppi_{n,n}^{(k-1)}(t))\right)-\frac{nM(\sum_{i=1}^{n}\kappa_{i})}{2K^{2}}-\frac{2\sqrt{\kappa}n^{2}Md(G)}{K^{2}}-\frac{2nM}{K}\sqrt{\frac{\kappa}{L}}\\
&+\frac{1}{K}\left\langle\left(\mathbf{d}_{1}^{(k)}(t),\dots,\mathbf{d}_{n}^{(k)}(t)\right),(\mathbf{v}_{1}^{(k)}(t),\dots,\mathbf{v}_{n}^{(k)}(t))-\one_{\mathcal{A}_{t}^{*}}\right\rangle.
\end{aligned}	
\end{equation*}

Then, we can have that
\begin{equation*}
\begin{aligned}
	&\E\left(f_{t}(\mathcal{A}_{t}^{*})-F_{t}(\uppi_{1,1}^{(k)}(t),\dots,\uppi_{n,n}^{(k)}(t))\right)\\
	&\le\left(1-\frac{\alpha}{K}\right)\E\left(f_{t}(\mathcal{A}_{t}^{*})-F_{t}(\uppi_{1,1}^{(k-1)}(t),\dots,\uppi_{n,n}^{(k-1)}(t))\right)+\frac{nM(\sum_{i=1}^{n}\kappa_{i})}{2K^{2}}+\frac{2\sqrt{\kappa}n^{2}Md(G)}{K^{2}}+\frac{2nM}{K}\sqrt{\frac{\kappa}{L}}\\
	&-\frac{1}{K}\left\langle\left(\mathbf{d}_{1}^{(k)}(t),\dots,\mathbf{d}_{n}^{(k)}(t)\right),(\mathbf{v}_{1}^{(k)}(t),\dots,\mathbf{v}_{n}^{(k)}(t))-\one_{\mathcal{A}_{t}^{*}}\right\rangle.
\end{aligned}	
\end{equation*}
As a result, 
\begin{equation*}
	\begin{aligned}
		&\E\left(f_{t}(\mathcal{A}_{t}^{*})-F_{t}(\uppi_{1,1}^{(K)}(t),\dots,\uppi_{n,n}^{(k)}(t))\right)\\
		&\le\left(1-\frac{\alpha}{K}\right)^{K}\E\left(f_{t}(\mathcal{A}_{t}^{*})-F_{t}(\uppi_{1,1}^{(0)}(t),\dots,\uppi_{n,n}^{(k-1)}(t))\right)\\
		&-\frac{1}{K}\sum_{k=1}^{K}\left\langle\left(\mathbf{d}_{1}^{(k)}(t),\dots,\mathbf{d}_{n}^{(k)}(t)\right),(\mathbf{v}_{1}^{(k)}(t),\dots,\mathbf{v}_{n}^{(k)}(t))-\one_{\mathcal{A}_{t}^{*}}\right\rangle\\
		&+\sum_{m=1}^{K}\left(1-\frac{\alpha}{K}\right)^{K-m}\left(\frac{nM(\sum_{i=1}^{n}\kappa_{i})}{2K^{2}}+\frac{2\sqrt{\kappa}n^{2}Md(G)}{K^{2}}+\frac{2nM}{K}\sqrt{\frac{\kappa}{L}}\right)\\
		&\le\left(1-\frac{\alpha}{K}\right)^{K}\E\left(f_{t}(\mathcal{A}_{t}^{*})\right)+\frac{nM(\sum_{i=1}^{n}\kappa_{i})}{2\alpha K}+\frac{2\sqrt{\kappa}n^{2}Md(G)}{\alpha K}+\frac{2nM}{\alpha}\sqrt{\frac{\kappa}{L}}\\
		&-\frac{1}{K}\sum_{k=1}^{K}\left\langle\left(\mathbf{d}_{1}^{(k)}(t),\dots,\mathbf{d}_{n}^{(k)}(t)\right),(\mathbf{v}_{1}^{(k)}(t),\dots,\mathbf{v}_{n}^{(k)}(t))-\one_{\mathcal{A}_{t}^{*}}\right\rangle\\
		&\le e^{-\alpha}\E\left(f_{t}(\mathcal{A}_{t}^{*})\right)+\frac{nM(\sum_{i=1}^{n}\kappa_{i})}{2\alpha K}+\frac{2\sqrt{\kappa}n^{2}Md(G)}{\alpha K}+\frac{2nM}{\alpha}\sqrt{\frac{\kappa}{L}}\\
		&-\frac{1}{K}\sum_{k=1}^{K}\left\langle\left(\mathbf{d}_{1}^{(k)}(t),\dots,\mathbf{d}_{n}^{(k)}(t)\right),(\mathbf{v}_{1}^{(k)}(t),\dots,\mathbf{v}_{n}^{(k)}(t))-\one_{\mathcal{A}_{t}^{*}}\right\rangle,
	\end{aligned}	
\end{equation*} where the second inequality follows from $\sum_{m=1}^{K}\left(1-\frac{\alpha}{K}\right)^{K-m}\le\frac{K}{\alpha},F_{t}(\mathbf{0})=0$ and the final inequality from from $(1-\frac{\alpha}{K})^{K}\le e^{-\alpha}$ when $K\ge3$.

So we can get the following result, 
\begin{equation*}
\begin{aligned}
&(1-e^{-\alpha})\sum_{t=1}^{T}f_{t}(\mathcal{A}_{t}^{*})-\sum_{t=1}^{T}\E\left(F_{t}(\uppi_{1,1}^{(K)}(t),\dots,\uppi_{n,n}^{(k)}(t))\right)\\&\le\left(\frac{nM(\sum_{i=1}^{n}\kappa_{i})}{2\alpha }+\frac{2\sqrt{\kappa}n^{2}Md(G)}{\alpha }\right)\frac{T}{K}+\frac{2nM\sqrt{\kappa}}{\alpha}\frac{T}{\sqrt{L}}\\
&-\frac{1}{K}\sum_{k=1}^{K}\sum_{i=1}^{n}\sum_{t=1}^{T}\left\langle\mathbf{d}_{i}^{(k)}(t),\mathbf{v}_{i}^{(k)}(t)-\one_{(\mathcal{A}_{t}^{*}\cap\V_{i})}\right\rangle
\end{aligned}
\end{equation*} where we slightly abuse the notation $\one_{(\mathcal{A}_{t}^{*}\cap\V_{i})}$ to denote $\kappa_{i}$-dimensional indicator vector over $(\mathcal{A}_{t}^{*}\cap\V_{i})$. Note that $\one_{\mathcal{A}_{t}^{*}}$ is a $\left(\sum_{i=1}^{n}\kappa_{i}\right)$-dimensional indicator vector over $\mathcal{A}_{t}^{*}$. 

Similarly, when $f_{t}$ is monotone $(\gamma,\beta)$-weakly submodular, we also have 
\begin{equation*}
	\begin{aligned}
		&\E\left(F_{t}(\uppi_{1,1}^{(k)}(t),\dots,\uppi_{n,n}^{(k)}(t))-F_{t}(\uppi_{1,1}^{(k-1)}(t),\dots,\uppi_{n,n}^{(k-1)}(t))\right)\\
		&\ge\frac{1}{K}\E\left(\gamma^{2}f_{t}(\mathcal{A}_{t}^{*})-\phi(\gamma,\beta)F_{t}(\uppi_{1,1}^{(k-1)}(t),\dots,\uppi_{n,n}^{(k-1)}(t))\right)-\frac{nM(\sum_{i=1}^{n}\kappa_{i})}{2K^{2}}-\frac{2\sqrt{\kappa}n^{2}Md(G)}{K^{2}}-\frac{2nM}{K}\sqrt{\frac{\kappa}{L}}\\
		&+\frac{1}{K}\left\langle\left(\mathbf{d}_{1}^{(k)}(t),\dots,\mathbf{d}_{n}^{(k)}(t)\right),(\mathbf{v}_{1}^{(k)}(t),\dots,\mathbf{v}_{n}^{(k)}(t))-\one_{\mathcal{A}_{t}^{*}}\right\rangle,
	\end{aligned}	
\end{equation*} where $\phi(\gamma,\beta)=\beta(1-\gamma)+\gamma^{2}$.

Then, we have
\begin{equation*}
	\begin{aligned}
		&\E\left(\gamma^{2}f_{t}(\mathcal{A}_{t}^{*})-\phi(\gamma,\beta)F_{t}(\uppi_{1,1}^{(k)}(t),\dots,\uppi_{n,n}^{(k)}(t))\right)\\
		&\le\left(1-\frac{\phi(\gamma,\beta)}{K}\right)\E\left(\gamma^{2}f_{t}(\mathcal{A}_{t}^{*})-\phi(\gamma,\beta)F_{t}(\uppi_{1,1}^{(k-1)}(t),\dots,\uppi_{n,n}^{(k-1)}(t))\right)\\
		&+\phi(\gamma,\beta)\left(\frac{nM(\sum_{i=1}^{n}\kappa_{i})}{2K^{2}}+\frac{2\sqrt{\kappa}n^{2}Md(G)}{K^{2}}+\frac{2nM}{K}\sqrt{\frac{\kappa}{L}}\right)\\
		&-\frac{\phi(\gamma,\beta)}{K}\left\langle\left(\mathbf{d}_{1}^{(k)}(t),\dots,\mathbf{d}_{n}^{(k)}(t)\right),(\mathbf{v}_{1}^{(k)}(t),\dots,\mathbf{v}_{n}^{(k)}(t))-\one_{\mathcal{A}_{t}^{*}}\right\rangle.
	\end{aligned}	
\end{equation*}

As a result, we can get the following result, 
\begin{equation*}
	\begin{aligned}
		&\frac{\gamma^{2}(1-e^{-\phi(\gamma,\beta)})}{\phi(\gamma,\beta)}\sum_{t=1}^{T}f_{t}(\mathcal{A}_{t}^{*})-\sum_{t=1}^{T}\E\left(F_{t}(\uppi_{1,1}^{(K)}(t),\dots,\uppi_{n,n}^{(k)}(t))\right)\\&\le\left(\frac{nM(\sum_{i=1}^{n}\kappa_{i})}{2\phi(\gamma,\beta) }+\frac{2\sqrt{\kappa}n^{2}Md(G)}{\phi(\gamma,\beta) }\right)\frac{T}{K}+\frac{2nM\sqrt{\kappa}}{\phi(\gamma,\beta)}\frac{T}{\sqrt{L}}\\
		&-\frac{1}{K}\sum_{k=1}^{K}\sum_{i=1}^{n}\sum_{t=1}^{T}\left\langle\mathbf{d}_{i}^{(k)}(t),\mathbf{v}_{i}^{(k)}(t)-\one_{(\mathcal{A}_{t}^{*}\cap\V_{i})}\right\rangle,
	\end{aligned}
\end{equation*}
\end{proof}
As a result, when  $K=\sqrt{T}$ and $L=T$, if \cref{ass:2} holds, that is, each linear maximization oracle $Q_{i}^{(k)}$ can achieve the a regret bound of $\mathcal{O}(\sqrt{V_{T}T})$ where $V_{T}$ is the variation of any feasible path $(\mathbf{u}_1,\dots,\mathbf{u}_T)$ where $\mathbf{u}_{t}\in\Delta_{\kappa_{i}},\forall t\in[T]$, that is, $V_{T}\triangleq\sum_{t=2}^{T}\|\mathbf{u}_{t}-\mathbf{u}_{t-1}\|_{2}$ for any path $(\mathbf{u}_1,\dots,\mathbf{u}_T)\in\prod_{t=1}^{T}\Delta_{\kappa_{i}}$, then we can have the following results:

\textbf{i)}: when $f_{t}$ is a monotone $\alpha$-weakly submodular function, \cref{alg:MPL} achieves:
\begin{equation*}
	\begin{aligned}
		&(1-e^{-\alpha})\sum_{t=1}^{T}f_{t}(\mathcal{A}_{t}^{*})-\sum_{t=1}^{T}\E\left(F_{t}(\uppi_{1,1}^{(K)}(t),\dots,\uppi_{n,n}^{(k)}(t))\right)\\&\le\left(\frac{nM(\sum_{i=1}^{n}\kappa_{i})}{2\alpha }+\frac{2\sqrt{\kappa}n^{2}Md(G)}{\alpha }\right)\frac{T}{K}+\frac{2nM\sqrt{\kappa}}{\alpha}\frac{T}{\sqrt{L}}\\
		&+\frac{1}{K}\sum_{k=1}^{K}\sum_{i=1}^{n}\mathcal{R}^{Q_{i}^{(k)}}\left(\one_{(\mathcal{A}_{1}^{*}\cap\V_{i})},\dots,\one_{(\mathcal{A}_{T}^{*}\cap\V_{i})}\right)\\
		&\le\frac{1}{K}\sum_{k=1}^{K}\sum_{i=1}^{n}\mathcal{O}\left(\sqrt{T\sum_{t=1}^{T-1}\left\|\one_{(\mathcal{A}_{t+1}^{*}\cap\V_{i})}-\one_{(\mathcal{A}_{t}^{*}\cap\V_{i})}\right\|_{2}}\right)+\mathcal{O}\left(d(G)\sqrt{T}\right)\\
		&\le\sum_{i=1}^{n}\mathcal{O}\left(\sqrt{T\sum_{t=1}^{T-1}\left\|\one_{\mathcal{A}_{t+1}^{*}}-\one_{\mathcal{A}_{t}^{*}}\right\|_{2}}\right)+\mathcal{O}\left(d(G)\sqrt{T}\right)\\
		&\le\mathcal{O}\left(d(G)\sqrt{T\sum_{t=1}^{T-1}\left\|\one_{\mathcal{A}_{t+1}^{*}}-\one_{\mathcal{A}_{t}^{*}}\right\|_{2}}\right)=\mathcal{O}\left(d(G)\sqrt{T\mathcal{P}_{T}}\right),
	\end{aligned}
\end{equation*} where the third inequality follows from the concavity of $\sqrt{\cdot}$ function and the final equality comes from $\mathcal{P}_{T}\triangleq\sum_{t=1}^{T}|\mathcal{A}_{t+1}^{*}\triangle\mathcal{A}_{t}^{*}|\triangleq\sum_{t=1}^{T-1}\|\one_{\mathcal{A}_{t}^{*}}-\one_{\mathcal{A}_{t+1}^{*})}\|_{1}\triangleq\sum_{t=1}^{T-1}\|\one_{\mathcal{A}_{t}^{*}}-\one_{\mathcal{A}_{t+1}^{*}}\|_{2}$.

\textbf{ii)}: Similarly, when $f_{t}$ is a monotone $(\gamma,\beta)$-weakly submodular:
\begin{equation*}
	\begin{aligned}
		&\frac{\gamma^{2}(1-e^{-\phi(\gamma,\beta)})}{\phi(\gamma,\beta)}\sum_{t=1}^{T}f_{t}(\mathcal{A}_{t}^{*})-\sum_{t=1}^{T}\E\left(F_{t}(\uppi_{1,1}^{(K)}(t),\dots,\uppi_{n,n}^{(k)}(t))\right)\\
		&\le\mathcal{O}\left(d(G)\sqrt{T\sum_{t=1}^{T-1}\left\|\one_{\mathcal{A}_{t+1}^{*}}-\one_{\mathcal{A}_{t}^{*}}\right\|_{2}}\right)=\mathcal{O}\left(d(G)\sqrt{T\mathcal{P}_{T}}\right).
	\end{aligned}
\end{equation*}  
\section{Limitation and Broader Impact}\label{sec:limitation}
In this work, in order to eliminate the dependence on the unknown DR ratio and submodularity ratio in our proposed \texttt{MA-SPL} algorithm,  we introduce a \emph{parameter-free} online algorithm named \texttt{MA-MPL} for the MA-OC problem. However, this new \texttt{MA-MPL} algorithm typically  incurs greater communication complexity, as shown in \cref{tab:Comparison}. Notably, recent studies~\citep{liao2023improved,zhang2023communication} employed a blocking procedure from \citep{zhang2019online} to reduce the number of communication in the decentralized online Frank-Wolfe algorithm~\citep{zhu2021projection}. Since our \texttt{MA-SPL} algorithm also can be viewed a variant of decentralized online Frank-Wolfe algorithm~\citep{zhu2021projection}, we believe that this  blocking strategy is a promising technique to help reduce the communication complexity of our proposed \texttt{MA-MPL} algorithm. We plan to explore this in future work. Furthermore, this work focuses on theoretically exploring MA-OC problem with $(\gamma,\beta)$-weakly submodular and $\alpha$-weakly DR-submodular objectives. So we do not foresee any form of negative social impact induced by our work.
\section{NeurIPS Paper Checklist}
The checklist is designed to encourage best practices for responsible machine learning research, addressing issues of reproducibility, transparency, research ethics, and societal impact. Do not remove the checklist: {\bf The papers not including the checklist will be desk rejected.} The checklist should follow the references and follow the (optional) supplemental material.  The checklist does NOT count towards the page
limit. 

Please read the checklist guidelines carefully for information on how to answer these questions. For each question in the checklist:
\begin{itemize}
	\item You should answer \answerYes{}, \answerNo{}, or \answerNA{}.
	\item \answerNA{} means either that the question is Not Applicable for that particular paper or the relevant information is Not Available.
	\item Please provide a short (1–2 sentence) justification right after your answer (even for NA). 
	% \item {\bf The papers not including the checklist will be desk rejected.}
\end{itemize}

{\bf The checklist answers are an integral part of your paper submission.} They are visible to the reviewers, area chairs, senior area chairs, and ethics reviewers. You will be asked to also include it (after eventual revisions) with the final version of your paper, and its final version will be published with the paper.

The reviewers of your paper will be asked to use the checklist as one of the factors in their evaluation. While "\answerYes{}" is generally preferable to "\answerNo{}", it is perfectly acceptable to answer "\answerNo{}" provided a proper justification is given (e.g., "error bars are not reported because it would be too computationally expensive" or "we were unable to find the license for the dataset we used"). In general, answering "\answerNo{}" or "\answerNA{}" is not grounds for rejection. While the questions are phrased in a binary way, we acknowledge that the true answer is often more nuanced, so please just use your best judgment and write a justification to elaborate. All supporting evidence can appear either in the main paper or the supplemental material, provided in appendix. If you answer \answerYes{} to a question, in the justification please point to the section(s) where related material for the question can be found.

IMPORTANT, please:
\begin{itemize}
	\item {\bf Delete this instruction block, but keep the section heading ``NeurIPS Paper Checklist"},
	\item  {\bf Keep the checklist subsection headings, questions/answers and guidelines below.}
	\item {\bf Do not modify the questions and only use the provided macros for your answers}.
\end{itemize}

%%% END INSTRUCTIONS %%%

\begin{enumerate}
	
	\item {\bf Claims}
	\item[] Question: Do the main claims made in the abstract and introduction accurately reflect the paper's contributions and scope?
	\item[] Answer: \answerYes{} % Replace by \answerYes{}, \answerNo{}, or \answerNA{}.
	\item[] Justification: The abstract and introduction clearly state the paper’s contribution and scope.
	
	\item {\bf Limitations}
	\item[] Question: Does the paper discuss the limitations of the work performed by the authors?
	\item[] Answer: \answerYes{} % Replace by \answerYes{}, \answerNo{}, or \answerNA{}.
	\item[] Justification:  Some limitations of our proposed \texttt{MA-MPL} algorithm have been discussed in the \cref{sec:limitation}.
	
	\item {\bf Theory assumptions and proofs}
	\item[] Question: For each theoretical result, does the paper provide the full set of assumptions and a complete (and correct) proof?
	\item[] Answer: \answerYes{} % Replace by \answerYes{}, \answerNo{}, or \answerNA{}.
	\item[] Justification: We have clearly stated the required assumptions and an accompanying complete proof in the appendix for each theory result.
	
	\item {\bf Experimental result reproducibility}
	\item[] Question: Does the paper fully disclose all the information needed to reproduce the main experimental results of the paper to the extent that it affects the main claims and/or conclusions of the paper (regardless of whether the code and data are provided or not)?
	\item[] Answer: \answerYes{} % Replace by \answerYes{}, \answerNo{}, or \answerNA{}.
	\item[] Justification: We present the detailed experiment setups and results in \cref{appendix:experiments}.

	\item {\bf Open access to data and code}
	\item[] Question: Does the paper provide open access to the data and code, with sufficient instructions to faithfully reproduce the main experimental results, as described in supplemental material?
	\item[] Answer: \answerYes{} % Replace by \answerYes{}, \answerNo{}, or \answerNA{}.
	\item[] Justification: We reveal the codes about our experiments in supplemental materials.

	\item {\bf Experimental setting/details}
	\item[] Question: Does the paper specify all the training and test details (e.g., data splits, hyperparameters, how they were chosen, type of optimizer, etc.) necessary to understand the results?
	\item[] Answer: \answerYes{} % Replace by \answerYes{}, \answerNo{}, or \answerNA{}.
	\item[] Justification: We present the detailed experiment setups and results in \cref{appendix:experiments}.
	
	\item {\bf Experiment statistical significance}
	\item[] Question: Does the paper report error bars suitably and correctly defined or other appropriate information about the statistical significance of the experiments?
	\item[] Answer: \answerYes{} % Replace by \answerYes{}, \answerNo{}, or \answerNA{}.
	\item[] Justification: We repeat each experiments five times and report the average utility in \cref{appendix:experiments}.
	
	\item {\bf Experiments compute resources}
	\item[] Question: For each experiment, does the paper provide sufficient information on the computer resources (type of compute workers, memory, time of execution) needed to reproduce the experiments?
	\item[] Answer: \answerYes{} % Replace by \answerYes{}, \answerNo{}, or \answerNA{}.
	\item[] Justification: We present the detailed experiment setups and results in \cref{appendix:experiments}.
	
	\item {\bf Code of ethics}
	\item[] Question: Does the research conducted in the paper conform, in every respect, with the NeurIPS Code of Ethics \url{https://neurips.cc/public/EthicsGuidelines}?
	\item[] Answer: \answerYes{} % Replace by \answerYes{}, \answerNo{}, or \answerNA{}.
	\item[] Justification: Our research conforms, in every respect, to the NeurIPS Code of Ethics.

	\item {\bf Broader impacts}
	\item[] Question: Does the paper discuss both potential positive societal impacts and negative societal impacts of the work performed?
	\item[] Answer: \answerYes{} % Replace by \answerYes{}, \answerNo{}, or \answerNA{}.
	\item[] Justification: Our work is primarily of theoretical nature and has no immediate societal impacts.
	
	\item {\bf Safeguards}
	\item[] Question: Does the paper describe safeguards that have been put in place for responsible release of data or models that have a high risk for misuse (e.g., pretrained language models, image generators, or scraped datasets)?
	\item[] Answer: \answerNA{} % Replace by \answerYes{}, \answerNo{}, or \answerNA{}.
	\item[] Justification: No high risk data or model have been used.
	
	\item {\bf Licenses for existing assets}
	\item[] Question: Are the creators or original owners of assets (e.g., code, data, models), used in the paper, properly credited and are the license and terms of use explicitly mentioned and properly respected?
	\item[] Answer: \answerNA{} % Replace by \answerYes{}, \answerNo{}, or \answerNA{}.
	\item[] Justification: No existing asset has been used in the paper.
	
	\item {\bf New assets}
	\item[] Question: Are new assets introduced in the paper well documented and is the documentation provided alongside the assets?
	\item[] Answer: \answerNA{} % Replace by \answerYes{}, \answerNo{}, or \answerNA{}.
	\item[] Justification:  No new asset is introduced in the paper
	
	\item {\bf Crowdsourcing and research with human subjects}
	\item[] Question: For crowdsourcing experiments and research with human subjects, does the paper include the full text of instructions given to participants and screenshots, if applicable, as well as details about compensation (if any)? 
	\item[] Answer: \answerNA{} % Replace by \answerYes{}, \answerNo{}, or \answerNA{}.
	\item[] Justification: No experiments with human subjects were conducted.
	
	\item {\bf Institutional review board (IRB) approvals or equivalent for research with human subjects}
	\item[] Question: Does the paper describe potential risks incurred by study participants, whether such risks were disclosed to the subjects, and whether Institutional Review Board (IRB) approvals (or an equivalent approval/review based on the requirements of your country or institution) were obtained?
	\item[] Answer: \answerNA{} % Replace by \answerYes{}, \answerNo{}, or \answerNA{}.
	\item[] Justification: We conducted no experiments with human subjects.
	
	\item {\bf Declaration of LLM usage}
	\item[] Question: Does the paper describe the usage of LLMs if it is an important, original, or non-standard component of the core methods in this research? Note that if the LLM is used only for writing, editing, or formatting purposes and does not impact the core methodology, scientific rigorousness, or originality of the research, declaration is not required.
	%this research? 
	\item[] Answer: \answerNA{} % Replace by \answerYes{}, \answerNo{}, or \answerNA{}.
	\item[] Justification: The core methodology of our work does not involve the use of LLMs as any important, original, or non-standard components.
	
\end{enumerate}
\end{document}